%% file: 0_main.tex
\documentclass[moor,nonblindrev]{style/informs4_arxiv}

 \OneAndAHalfSpacedXI 

\usepackage[margin=1in]{geometry}
\usepackage{endnotes}
\let\footnote=\endnote
\let\enotesize=\normalsize
\def\notesname{Endnotes}%
\def\makeenmark{$^{\theenmark}$}
\def\enoteformat{\rightskip0pt\leftskip0pt\parindent=1.75em
	\leavevmode\llap{\theenmark.\enskip}}

\usepackage{graphicx}
\usepackage{multirow}
\usepackage{hhline}

\usepackage[small, margin=1cm]{caption}

\usepackage{appendix}
\usepackage{color}
\definecolor{strcolor}{rgb}{0.6, 0.2, 0.6}
\definecolor{commentcolor}{rgb}{0.3125, 0.5, 0.3125}
\definecolor{keycol}{rgb}{0, 0, 1}

\usepackage{bbm}

\usepackage{listings}
\usepackage{url}

\usepackage{amsmath}
\usepackage{mathtools}
\usepackage{dsfont}
\usepackage{aligned-overset}
\usepackage{tikz}
\usetikzlibrary{calc,positioning, math}
\usepackage{pgfplots}
\pgfplotsset{compat=1.18}
\usepackage{stmaryrd}
\SetSymbolFont{stmry}{bold}{U}{stmry}{m}{n}
\usepackage{soul}
\usepackage{dcolumn,booktabs}
\usepackage{csvsimple}

\usepackage[linesnumbered,ruled,vlined]{algorithm2e}
\SetKwInput{KwInput}{Input}                
\SetKwInput{KwOutput}{Output}              
\usepackage{subfig}
\graphicspath{{figures/}}

\newcommand{\V}{\mathcal{V}}
\newcommand{\E}{\mathcal{E}}
\newcommand{\C}{\mathcal{C}}

\newcommand{\I}{\mathcal{I}}

\newcommand{\R}{\mathbb{R}}
\newcommand{\Z}{\mathbb{Z}}

\newcommand{\F}{\mathcal{F}}
\newcommand{\Fh}{\mathcal{F}_{\textnormal{high}}}
\newcommand{\Fl}{\mathcal{F}_{\textnormal{low}}}
\newcommand{\Fe}{\mathcal{F}_{\textnormal{eq}}}

\newcommand{\RMP}[1]{\textnormal{(LP(#1))}}

\DeclareMathOperator{\Supp}{Supp}
\DeclareMathOperator*{\dprime}{\prime \prime}
\DeclarePairedDelimiter{\ceil}{\lceil}{\rceil}
\DeclarePairedDelimiter{\floor}{\lfloor}{\rfloor}
\newcommand{\DRE}[2]{ \textnormal{D}_{\textnormal{RE}}\left( #1 \, \middle\| \, #2 \right)} 
\newcommand{\A}{\textnormal{A}}
\newcommand{\D}{\textnormal{D}}
\renewcommand{\AA}{\mathcal{A}}
\newcommand{\AD}{\mathcal{D}}
\newcommand{\rA}{r_{\textnormal{A}}}
\newcommand{\rD}{r_{\textnormal{D}}}
\newcommand{\sigmaA}{\sigma^{\textnormal{A}}}
\newcommand{\sigmaD}{\sigma^{\textnormal{D}}}
\newcommand{\sigmaAEq}{\sigma^{\textnormal{A}^{*}}}
\newcommand{\sigmaDEq}{\sigma^{\textnormal{D}^{*}}}
\newcommand{\sigmaDHat}{\widehat{\sigma\,}^{\textnormal{D}}}
\newcommand{\rhoAHat}{\widehat{\rho\,}^{\textnormal{A}}}
\newcommand{\DeltaA}{\Delta_{\textnormal{A}}}
\newcommand{\DeltaD}{\Delta_{\textnormal{D}}}
\newcommand{\rhoA}{\rho^{\textnormal{A}}}
\newcommand{\rhoAEq}{\rho^{\textnormal{A}^{*}}}
\newcommand{\ThetaA}{\Theta_{\textnormal{A}}}
\newcommand{\LI}{\textnormal{I}}    
\newcommand{\LII}{\textnormal{II}}    
\newcommand{\LIII}{\textnormal{III}}    
\newcommand{\LIV}{\textnormal{IV}}    

\definecolor{fred}{RGB}{237,41,57}
\setstcolor{fred}

\newcommand{\eg}{\textit{e.g.}}
\newcommand{\ie}{\textit{i.e.}}

\def\blot{\quad \mbox{$\vcenter{ \vbox{ \hrule height.4pt
				\hbox{\vrule width.4pt height.9ex \kern.9ex \vrule width.4pt}
				\hrule height.4pt}}$}}

\usepackage[sort&compress]{natbib}
 \bibpunct[, ]{[}{]}{,}{n}{}{,}%
 \def\bibfont{\small}%
\TheoremsNumberedThrough     
\ECRepeatTheorems

\EquationsNumberedThrough    

\usepackage[colorlinks=true,breaklinks=true,bookmarks=true,urlcolor=blue,citecolor=blue,linkcolor=blue,bookmarksopen=false,draft=false]{hyperref}

\gdef\AQ#1{}
\gdef\CQ#1{}
\begin{document}
	
\def\COPYRIGHTHOLDER{INFORMS}%
\def\COPYRIGHTYEAR{2017}%
\def\DOI{\fontsize{7.5}{9.5}\selectfont\sf\bfseries\noindent https://doi.org/10.1287/opre.2017.1714\CQ{Word count = 9740}}

	\RUNAUTHOR{Bahamondes and Dahan} %

	\RUNTITLE{Inspection Game with Location-Specific Detection Capabilities}

\TITLE{Inspection Game with Location-Specific Detection Capabilities: Exact and Approximate Algorithms for Strategic Resource Coordination}


\ARTICLEAUTHORS{

\AUTHOR{Basti\'an Bahamondes, Mathieu Dahan}
\AFF{H. Milton Stewart  School of Industrial and Systems Engineering, Georgia Institute of Technology, Atlanta, Georgia 30332, \{\EMAIL{bbahamondes3@gatech.edu}, \EMAIL{mathieu.dahan@isye.gatech.edu}\}}




}

\ABSTRACT{We consider a zero-sum inspection game, in which a defender positions detectors across a critical system to detect multiple attacks caused by an attacker. We assume that detection is imperfect, and each detector location is associated with a probability of detecting attacks within its set of monitored system components. The objective of the defender (resp. attacker) is to minimize (resp. maximize) the expected number of undetected attacks. To compute Nash equilibria for this large-scale zero-sum game, we formulate a linear program with a small number of constraints that can be solved via column generation. We provide an exact mixed-integer program for the pricing problem and leverage its supermodular structure to design two efficient approaches for computing approximate Nash equilibria with theoretical guarantees: A column generation and a multiplicative weights update algorithm, both with approximate best responses. To address the computational challenges posed by combinatorial attacker strategies, each iteration of our multiplicative weights update algorithm computes a projection onto the polytope of marginal attack probabilities under the unnormalized relative entropy, for which we derive a closed-form expression and a linear-time algorithm. Computational results on real-world gas distribution networks illustrate the performance and scalability of our solution approaches.}

\SUBJECTCLASS{Military: search/surveillance; Games/group decisions: noncooperative; Programming: linear: large-scale systems}



\maketitle

\input{1_introduction}

\input{2_model_description}

\input{3_exact_solution_methods}

\input{4_approximate_solution_methods}

\input{5_computational_study}

\input{6_conclusions}



 \ACKNOWLEDGMENT{This work was supported by the Office of Naval Research under Grant No. N00014-24-1-2047.}

\bibliographystyle{style/informs2014} 
\bibliography{references} 



\input{7_appendix}
	\end{document}

%% file: 1_introduction.tex
\section{Introduction}\label{sec:introduction}

\subsection{Motivation}

Modern societies' welfare significantly relies on critical infrastructure systems, including power grids, transportation and telecommunication systems, and water, oil or gas distribution pipelines. However, these systems have increasingly become targets for malicious cyber-physical attacks, as evidenced by past incidents that have significantly disrupted their operations. In 2015, Ukraine experienced several foreign cyberattacks that compromised the supervisory control and data acquisition (SCADA) systems managing its power grid. These attacks resulted in widespread power outages, affecting 225,000 customers (\citet{lee2016analysis}). In 2021, the Colonial Pipeline ransomware attack led to substantial disruptions in fuel supply along the East Coast of the United States, causing losses of $\$4{.}4$ million USD (\citet{beerman2023review}). 
Such targeted attacks pose a severe threat to the reliability and functionality of critical infrastructures, highlighting the urgent need for robust inspection systems capable of detecting both random anomalies and adversarial attacks. Developing effective inspection mechanisms has become a paramount requirement to ensure the continuous operation, safety, and resilience of these networks (\citet{mo2011cyber,sandberg2015cyberphysical,weerakkody2019challenges}).

Numerous efforts have been deployed in the development of inspection technologies that monitor critical parameters such as temperature, pressure, flow rates, and voltage levels to assess the operational status and performance of various networks. These technological advances, often relying on online data acquisition and continual monitoring through the flexible allocation of sensors within the system enable operators to promptly identify anomalies or potential threats and perform proactive maintenance tasks and timely mitigation responses (\citet{vidacs2015wireless}).

In the context of large-scale systems and costly detection technology, the limited availability of sensing resources can pose significant limitations for the continuous monitoring of all the system components. Moreover, sensing reliability can be locally undermined by several factors, such as obstacles and harsh environmental conditions. For instance, detection systems based on infrared thermography (IRT) can identify irregular temperature patterns around gas pipelines, signaling potential leaks. Yet, high-resolution infrared cameras are expensive, and IRT's efficiency might significantly decrease for buried pipelines or within the presence of covering materials such as concrete (\citet{adegboye2019recent}). 


Hence, in this article we aim to investigate the following research question: \emph{How to effectively coordinate limited inspection resources within a critical system with location-specific detection capabilities, in order to optimize the detection of multiple adversarial attacks on its components?}

\subsection{Contributions}


%

To address the research question, we investigate a two-person zero-sum inspection game between a defender and an attacker who strategically coordinate multiple resources within a system. Specifically, the defender randomizes the placement of detectors across a set of locations, while the attacker targets multiple system components. Placing a detector at a location enables the defender to monitor a subset of components, which we refer to as \emph{monitoring set}. To capture the local effects undermining detection capabilities, we assume that upon inspection, the defender identifies each attack within a monitoring set with a location-specific probability. Our model allows monitoring sets to overlap, enabling the defender to simultaneously monitor components from multiple locations, which is particularly relevant when inspecting critical infrastructure systems.
The defender (resp. attacker) aims to minimize (resp. maximize) the expected number of undetected attacks.


Our model extends previous works in the literature, which typically involve a single unit of resource for one or both players (\citet{von1953certain,washburn1995two,karlin2016game}), assume perfect detection capabilities (\citet{dahan2022network}), or restrict detection to nonoverlapping monitoring sets (\citet{bahamondes2024hide}). More generally, our contribution expands the literature on combinatorial games and randomized robust optimization with \emph{supermodular} structure, which has been less explored compared to similar settings with submodular structure, enhancing our understanding of strategic interactions in complex systems.


Computing Nash equilibrium (NE) strategies for this large-scale zero-sum game is challenging, as the number of players' strategies grows combinatorially with the size of the system and the number of available resources. Instead, we leverage the structure of our game to compactly represent the attacker's strategies in terms of marginal probabilities of targeting individual components. Using this representation, we formulate a linear program (LP) with a small number of constraints that can be solved via column generation to obtain exact NE strategies (Proposition~\ref{prop:LP1 reformulation}). The pricing problem, corresponding to the defender's pure best response problem, consists in finding a detector positioning that minimizes the supermodular expected number of undetected attacks, subject to a cardinality constraint. We show that this problem is NP-hard (Proposition~\ref{prop:NP_hardness_Defender_BR_2}) and provide a compact mixed-integer programming (MIP) formulation (Proposition~\ref{prop:MIP_optimality}).

Next, leveraging the supermodular structure of the defender's best response problem, we rely on both the forward and reverse greedy algorithms to compute approximate best responses, whose guarantees depend on the curvature of the defender's payoff function given the attacker's strategy (\citet{ilev2001approximation,sviridenko2017optimal}). This enables us to compute approximate equilibrium strategies for our game by devising a column generation algorithm with approximate best responses for the pricing problem. This method generates approximate equilibria by inheriting the approximation guarantees for the pricing problem.
%
%
%
%
%
%
%
%
We extend a result by \citet{guo2016optimal} by allowing early termination of column generation at the cost of a small additive error, and by simplifying the convergence criterion (Theorem~\ref{thm:column_generation_with_approx_best_response}).

We derive another approach for computing approximate equilibria, by integrating the Multiplicative Weights Update (MWU) algorithm with an approximation algorithm for the defender's best response problem, as in \citet{krause2011randomized}.
%
%
%
This integration yields similar approximation guarantees to those of column generation, but in polynomial time (Theorem~\ref{thm:MWU_epsilon_NE}). However, the combinatorial number of attacker's strategies in our game renders the update step in \citeauthor{krause2011randomized}'s implementation of the MWU algorithm computationally prohibitive. We address this issue by implementing the MWU algorithm on the attacker's marginal probabilities, instead of the original mixed strategies used in the standard implementation. This adaptation requires solving an additional projection problem at each iteration to recover feasibility of the attacker's marginal attack probabilities after the update step.

We analytically solve the projection problem arising in our implementation of the MWU algorithm by characterizing the projection under the unnormalized relative entropy onto the \emph{full-dimensional} capped simplex polytope (Theorem~\ref{thm:projection_closed_form}). To achieve this result, we identify a precise condition that determines whether the resource constraint is either tight or not at optimality. To handle the former case, we employ a Lagrangian formulation of the projection problem and circumvent the nonsmoothness of the unnormalized relative entropy by dualizing the equality constraint and then solving independent single-variable optimization problems for each entry of the projection. The resulting closed-form solution involves a key index parameter, which we compute using a linear-time algorithm (Theorem~\ref{thm:fast_projection}). Our algorithm's running time matches that of previous optimal algorithms designed for variants of the \emph{lower-dimensional} capped simplex and related polytopes (\citet{herbster2001tracking,krichene2015efficient,lim2016efficient,robinson2021improved}), and is faster than other more general-purpose algorithms intended for computing Bregman projections (\citet{suehiro2012online,gupta2016solving}). 

    
We conduct a computational study to evaluate the performance of our solution approaches in real-world instances from gas distribution networks across the United States and Europe. By considering various network configurations, we provide insights into the performance of our exact and approximate solution methods. Furthermore, we analyze the quality of approximations and scalability of our approximate solution methods with different detector allocation scenarios, providing insights into the practical applicability of each method. Our computational results validate the theoretical underpinnings of our proposed algorithms and offer practical guidelines for efficiently addressing real-world network security challenges.

\section{Literature Review}
\label{sec:literature_review}

Game theory has emerged as a powerful tool applied to security related problems (see, for example, \citet{pita2008deployed,kiekintveld2011computing,zhu2015game,gupta2016dynamic,hota2016optimal,bertsimas2016power,miao2018hybrid,pirani2021strategic}). One of such practical applications arises in problems of strategic sensor placement for infrastructure inspection, in which a defender positions sensors in locations of a critical infrastructure system to detect attacks caused by a strategic attacker. Such models typically account for the detection range of the sensors: Positioning a sensor allows the defender to monitor a subset of components---often referred to as a monitoring set---and potentially detect attacks occurring within it. As a result, attacks may be detected from multiple locations. This overlapping feature renders such games challenging to solve.


A body of work on two-person zero-sum games has investigated this problem, assuming that neither player can observe the opponent's strategy before committing to their own action. \citet{dahan2022network} studied a restricted version of our model under the assumption of perfect detection and obtained approximate equilibria by means of minimum set covers and maximum set packings. \citet{milosevic2023strategic} extended this model by considering components with heterogeneous criticality values. Our model generalizes the classical hide-and-seek game of \citet{von1953certain}, as well as the network interdiction game of \citet{washburn1995two}, handling multiple checkpoints and coordinated evaders in networks with a reasonable number of paths (which typically occurs in illicit supply chains).

\citet{bahamondes2022network,bahamondes2024hide} introduced imperfect, location-specific detection capabilities into the inspection model and solved the game restricted to nonoverlapping (\ie, mutually disjoint) monitoring sets. This assumption renders the model akin to a hide-and-seek game with imperfect detection and heterogeneous hiding capacities, allowing closed-form equilibria in terms of marginal inspection probabilities and expected number of items hidden within each location. Our preliminary analysis of the more general setting with overlapping monitoring sets and imperfect detection appeared in \citet{bahamondes2022network}, where we proposed heuristic solutions based on minimum weighted set covers. In contrast, we now derive both exact and approximate solution approaches with proven approximation guarantees, which require novel structural results.

Another related stream of literature stems from the study of Stackelberg security games (\citet{10.5555/3304652.3304789}), where a defender first deploys resources to protect a set of targets, and an attacker subsequently launches an attack after potentially observing the defender's strategy. These models typically assume that the attacker targets a single location and that detection is perfect (\citet{10.5555/2898607.2898735,10.5555/2030470.2030518,kiekintveld2011computing,bucarey2021coordinating,bustamante2024playing,paruchuri08bayesian}), both assumptions that our model relaxes. More recently, this literature has incorporated uncertainty regarding the attacker's type (\citet{bucarey2021coordinating,bustamante2024playing,paruchuri08bayesian}) and behavior (\citet{10.5555/2891460.2891560,10.5555/2772879.2773329,10.5555/2540128.2540187}). Closely related to our work is the study of \citet{wang2017security}, which extends the security game of \citet{kiekintveld2011computing} by introducing multiple attacker resources and nonadditive utilities. Unlike our model, they assume an unconstrained number of defender resources but a more general nonzero-sum utility structure. 


Our work also relates to randomized robust optimization of submodular functions. \citet{krause2011randomized} and \citet{chen2017robust} studied robust sensor positioning problems with bounded submodular functions and used the MWU algorithm to compute approximate equilibrium strategies. \citet{kawase2019randomized} extended these results for payoffs taking values within $\R_{\geq 0}$, while \citet{gupta2016solving} employed similar techniques to solve combinatorial zero-sum games with bilinear payoffs, where each player's strategy space corresponds to the vertices of a polytope. 
\citet{wilder2018equilibrium} developed pseudopolynomial-time algorithms for submodular settings with exponentially many adversary strategies. \citet{adibi2022minimax} addressed minmax problems with submodular functions, showing hardness results and introducing approximation algorithms that combine nonlinear optimization techniques with greedy submodular maximization. Our approximate solution approach follows that of \citet{krause2011randomized}, but the supermodularity in our game's payoff prevents constant-factor approximations. Instead, we use instance-dependent approximations and an additional projection step in the MWU algorithm for updating the attacker's marginal distributions with respect to the unnormalized relative entropy.


The MWU algorithm (\citet{freund1999adaptive}) can be viewed as a special case of both the Mirror Descent algorithm (\citet{nemirovski1983problem}) and the Follow the Perturbed Leader algorithm (\citet{kalai2005efficient}), and versions of it introducing projection steps have been proposed by \citet{robinson2021improved,hembold2009learning,koolen2010hedging}, and \citet{warmuth2008randomized}, but primarily in the context of online learning over structured concept classes, where the focus is on minimizing the learner's regret in the adversarial setting. \citet{hellerstein2019solving} used approximate best responses for zero-sum search games, but their approach does not apply directly to our model, as it requires one player to have a linear number of strategies, whereas in our game \emph{both} players face combinatorially many strategies. Other approaches simulate the MWU algorithm by maintaining product distributions, which enable efficient updates provided that the players’ pure strategies correspond to the vertices of a 0/1 polytope of polynomial dimension and that an associated polyhedral kernel function can be evaluated in polynomial time (\citet{gupta2016solving,farina2022kernelized}). However, these methods are not applicable to our setting, since the overlapping structure of the monitoring sets and the non-additivity of the defender's payoff function in our game precludes the existence of such a low-dimensional 0/1 polytope representation of the defender's strategies.

Efficient algorithms for computing projections with respect to the unnormalized relative entropy, as well as more general Bregman divergences, onto the simplex and related polytopes---including variants of the capped simplex and the permutahedron---have been studied by \citet{herbster2001tracking,warmuth2008randomized,yasutake2011online,suehiro2012online,krichene2015efficient,gupta2016solving,lim2016efficient,robinson2021improved}. Unlike these prior works, our game allows players to utilize fewer resources than their respective budgets. Hence, our contribution focuses on the study of projections onto the \emph{full-dimensional} capped simplex. This polytope encompasses vectors within the unit cube, with the constraint that the sum of their entries is \emph{at most} a given resource budget. This contrasts with the \emph{lower-dimensional} capped simplex polytope, which requires the sum of their entries to be \emph{equal} to that budget. We derive a closed-form solution to our projection problem, extending the technique of \citet{ang2021fast} used for analytically solving the projection problem under the Euclidean distance onto the capped simplex. Our solutions lead us to derive a linear-time algorithm that is faster than more general-purpose algorithms for minimizing separable convex functions on base polytopes of polymatroids (\citet{suehiro2012online,gupta2016solving}), and that does not require adding perturbations to the relative entropy function to handle its nonsmoothness (\citet{krichene2015efficient}).

{\bf Outline. }This article is organized as follows. In Section \ref{sec:model description}, we introduce the inspection game. We then review preliminary results that help us investigate our game in Section~\ref{sec:preliminaries}. In Sections \ref{sec:exact equilibrium strategies} and \ref{sec:approximate equilibrium strategies}, we derive algorithmic approaches for computing exact and approximate equilibrium strategies, respectively. In Section \ref{sec:computational study}, we present our computational results. We then discuss our conclusions and plans for future work in Section \ref{sec:conclusion}. Finally, the proofs of our results are detailed in the appendix.

%% file: 2_model_description.tex
\section{Model Description}\label{sec:model description}


We consider a critical system comprising a set of $m$ \emph{components} $\E$ that can be potential targets for an attacker, and a set of $n$ \emph{locations} $\V$ that can receive detectors from a defender for monitoring the system and detecting the attacks. By placing a detector at a location $v \in \V$, the defender can monitor a subset of components $\C_v \subseteq \E$, which we refer to as the \emph{monitoring set} of location $v$. Without loss of generality, we assume that each monitoring set is nonempty and each component can be monitored from at least one location. An instance of the detection model is shown in Figure~\ref{fig:network model general monitoring sets}.

%

\begin{figure}[htbp]
    \centering
    \input{figures/example_network_model.tikz}
    \caption{Detection model instance. The set of detector locations is $\V = \{v_{1},\dots, v_{4}\}$ and the set of components is $\E =\{e_{1}, \dots, e_{7}\}$. The monitoring sets are $\C_{v_1}=\{e_1,e_2\}$, $\C_{v_2}=\{e_{2},e_{3}\}$, $\C_{v_3} = \{e_{3},\dots,e_7\}$ and $\C_{v_4}=\{e_{5}\}$.}
    \label{fig:network model general monitoring sets}
\end{figure}


Due to inherent characteristics undermining the detection capabilities at each location, we assume imperfect detection of attacks. Specifically, by placing a detector at a location $v$, the defender identifies an attack on each targeted component in $\C_v$ with a probability $p_v$, independent of other targeted components and detectors positioned. We refer to $p_v$ as the \emph{detection probability} of location $v$. Therefore, an attack on a component will only be detected if the defender successfully identifies it through at least one of the detectors positioned at locations monitoring such component.

We assume that both the attacker and the defender are strategic, and hence adopt a game-theoretic framework to study their behaviors. We define a simultaneous-move two-person zero-sum game $\Gamma$ between a defender $\D$ and an attacker $\A$. In this game, both players face exogenous resource constraints: $\D$ can select up to $\rD \in \Z_{>0}$ locations from $\V$ to place detectors and $\A$ can select up to $\rA \in \Z_{>0}$ components of $\E$ to target. Without loss of generality, we assume that each location can host at most one detector and each component can be the target of at most one attack. Hence, we consider that $\rD \leq n$ and $\rA\leq m$, and we let $\AD \coloneqq\{ S\subseteq \V: |S| \leq \rD\}$ and $\AA \coloneqq\{ T\subseteq \E: |T| \leq \rA\}$ be the action sets for $\D$ and $\A$, respectively. We assume that both players have complete information regarding the critical system (detector locations, components, monitoring sets), detection probabilities, and their own and their opponent's action sets. Such information may be inferred from public sources and regulatory requirements.





We consider that in the event of a successful detection of an attack, $\D$ can start a response mechanism to mitigate its damage. Thus, in our model, an attack is successful if and only if it remains undetected by $\D$. Let $u: \AD \times \AA \rightarrow \mathbb{R}$ be the \emph{undetection function}, defined as the average number of undetected attacks:
\begin{align*}
    u(S,T) \coloneqq  \sum_{e\in T} \prod_{v \in S:\, e\in \C_v} (1-p_v), \quad \forall \, (S,T)\in \AD\times\AA.
\end{align*}

We use $u(S,e)$ (resp. $u(v, T)$) to denote the case when $T=\{e\}$ (resp. $S=\{v\}$), for some $e\in \E$ (resp. $v\in \V$). We note that the undetection function satisfies $u(S,T) = \sum_{e\in T} u(S,e)$ for every $(S,T) \in \AD \times \AA$. The term $u(S,e)=\prod_{v\in S:\, e\in \C_v} (1-p_v)$ represents the undetection probability of an attack on $e$ when $\D$ selects the detector positioning $S\in \AD$. 

We allow the players to use mixed strategies, defined as probability distributions over their sets of pure actions. To this aim, we let $\DeltaD  \coloneqq \ \{ \sigmaD \in [0,1]^{\AD}:\, \sum_{S\in \AD}\sigmaD_S = 1 \}$ and $\DeltaA  \coloneqq  \{ \sigmaA \in [0,1]^{\AA}:\, \sum_{T\in \AA}\sigmaA_T = 1 \}$
be the sets of mixed strategies of $\D$ and $\A$, respectively. For each $\sigmaD \in \DeltaD$ (resp. $\sigmaA \in \DeltaA$), $\sigmaD_S$ (resp. $\sigmaA_T$) represents the probability that action $S\in \AD$ (resp. $T\in \AA$) is executed.

Given a strategy profile $(\sigmaD, \sigmaA) \in \DeltaD \times \DeltaA$, the expected number of undetected attacks is then defined as $U(\sigmaD, \sigmaA)\coloneqq \sum_{S\in \AD} \sum_{T\in \AA} \sigmaD_S\, \sigmaA_T\, u(S,T)$. We assume that $\D$ (resp. $\A$) aims to minimize (resp. maximize) $U$.
 For ease of notation, we use $U(S,\sigmaA)$ (resp. $U(\sigmaD, T)$) to denote when $\sigmaD_S=1$ (resp. $\sigmaA_T=1$) for some $S\in \AD$ (resp. $T\in \AA$), and omit the brackets when $S=\{v\}$ (resp. $T=\{e\}$). 

Our subject of interest is the study of strategy profiles $(\sigmaDEq, \sigmaAEq)$ that represent \emph{Nash equilibria} (NE) of the game $\Gamma$, \ie, that satisfy
\begin{align*}
   U(\sigmaDEq, \sigmaA) \leq  U(\sigmaDEq, \sigmaAEq) \leq U(\sigmaD, \sigmaAEq), \quad \forall\, \sigmaD \in \DeltaD,\; \forall\, \sigmaA\in \DeltaA.
\end{align*}

NE are relevant in settings where neither player can observe the opponent’s actual allocation before committing to their own action. This assumption reflects conditions in critical infrastructure defense---such as water, gas, and power distribution networks---which span geographically wide and often restricted areas and operate in challenging environments. In these settings, defensive allocations---such as drone patrols or mobile inspection teams---are deliberately concealed, frequently redeployed, and at best only partially observable to an adversary. Finally, we remark that because we model the interaction as a zero-sum game, every Nash equilibrium is also a (strong) Stackelberg equilibrium (\citet{korzhyk2011stackelberg}).

We refer to $U(\sigmaDEq,\sigmaAEq)$ as the \emph{value of the game} $\Gamma$. Because $\Gamma$ is a two-person zero-sum game with a finite number of player actions, an NE is guaranteed to exist and the value of the game is unique. Alternatively, NE can be characterized by strategy profiles $(\sigmaDEq, \sigmaAEq)$ where both $\sigmaDEq$ and $\sigmaAEq$ are respectively optimal solutions of the following pair of minmax and maxmin problems:
\begin{align*}
    \underset{\sigmaD\in \DeltaD}{\min}\,\underset{\sigmaA\in \DeltaA}{\max} \, U(\sigmaD, \sigmaA), \qquad
    \underset{\sigmaA\in \DeltaA}{\max}\,\underset{\sigmaD\in \DeltaD}{\min} \, U(\sigmaD, \sigmaA).
\end{align*}

Furthermore, the optimal values of both these problems are identical and coincide with the value of the game (\citet{von1928theorie}; see also \citet{von2004theory}). From this equivalence, $\sigmaDEq$ and $\sigmaAEq$ are also referred to as \emph{optimal strategies} of the game $\Gamma$ and can be computed by solving the following linear program (LP):
\begin{align}
    \label{LP: LP1}
    \tag{LP}
    \begin{alignedat}{4}
        & \underset{\sigmaD \in \R^{\AD},\, z \in \R}{\min} & & z \\
        & \textnormal{subject to}& \quad & \begin{aligned}[t]
                                            \sum_{S\in \AD} \sigmaD_S\,u(S, T) & \leq z  & \forall\, & T\in \AA,\\
                                            \sum_{S \in \AD}\sigmaD_S & = 1 & \quad  &  \\
                                            \sigmaD_S & \geq 0 & \forall\, & S\in \AD.
                                        \end{aligned}
    \end{alignedat}
\end{align}


\eqref{LP: LP1} and its dual are reformulations of the minmax and maxmin problems, respectively. Thus, equilibirum inspection (resp. attack) strategies and the value of the game $\Gamma$ are given by the optimal primal (resp. dual) solutions and optimal value of \eqref{LP: LP1}, respectively. Nonetheless, because the cardinality of $\AD$ (resp. $\AA$) grows combinatorially with $\rD$ (resp. $\rA$), this LP becomes computationally challenging to solve, even for small-sized instances. Hence, we derive solution approaches that leverage the structure of the undetection function, allowing us to compute exact and approximate equilibria of the game $\Gamma$.

\section{Preliminaries}
\label{sec:preliminaries}

In our analysis, it will be useful to characterize the attacker's strategies in terms of their marginal probabilities of targeting each component. For every mixed attack strategy $\sigmaA\in \DeltaA$, we denote as $\rho(\sigmaA) = (\rho_e(\sigmaA))_{e \in \E}$ the vector of marginal probabilities of targeting each component when $\A$ plays $\sigmaA$, given by $\rho_e(\sigmaA) \coloneqq  \sum_{T\in \AA:\, e\in T}  \sigmaA_T$ for every $e \in \E$.
The additive structure of the undetection function $u(S,T)$ with respect to $\A$'s pure actions permits us to represent the expected number of undetected attacks in terms of these marginal attack probabilities and the undetection probabilities resulting from $\D$'s mixed inspection strategy.
\begin{lemma}
    \label{lem:undetection_function}
    For every $(\sigmaD,\sigmaA) \in \DeltaD \times \DeltaA$, $U(\sigmaD,\sigmaA) = \sum_{e\in \E} \rho_e(\sigmaA)\, U(\sigmaD,e)$.
\end{lemma}
From Lemma~\ref{lem:undetection_function}, the expected number of undetected attacks is given by the sum, over each component, of the probability that such component is targeted by $\A$ and the attack remains undetected by $\D$. While every mixed attack strategy induces a vector of marginal attack probabilities in $\R^{\E}$, we are also interested in the inverse question of identifying the vectors in $\R^{\E}$ that represent marginal attack probabilities resulting from some mixed attack strategy. This bidirectional mapping will facilitate transitioning between the spaces of mixed strategies and marginal attack probabilities. The following lemma provides necessary and sufficient conditions that such vectors must satisfy.

\begin{lemma}[\citet{bahamondes2024hide}]
    \label{lem:convex_hull}
    Consider a resource budget $\rA \in \Z_{>0}$, and a vector $\rho \in \R^{\E}$. Then, there exists a probability distribution $\sigmaA \in \DeltaA$ satisfying $\rho_e(\sigmaA) = \rho_{e}$ for all $e\in \E$ if and only if $\rho\in [0,1]^{\E}$ and $\sum_{e\in \E}\rho_e \leq \rA$. 
    Furthermore, such $\sigmaA$ with a support of size at most $m+1$ can be computed in time $O(m^2)$.
\end{lemma}
As a consequence of Lemma~\ref{lem:convex_hull}, the vectors in $\R^{\E}$ that represent marginal attack probabilities resulting from $\A$'s mixed strategies are precisely the feasible solutions of the \emph{full-dimensional capped simplex} polytope $\ThetaA \coloneqq \left\{ \rho \in [0,1]^{\E}:\, \sum_{e\in \E} \rho_e \leq \rA \right\}$.
This provides an alternative representation of $\A$'s strategies as vectors in $\ThetaA$, which offers a computational advantage, as the elements of $\ThetaA$ are vectors of size $m$, while $\A$'s mixed strategies require vectors of size $\sum_{k=0}^{\rA} \binom{m}{k}$. 

In view of Lemmas~\ref{lem:undetection_function} and \ref{lem:convex_hull}, for the remainder of this article we denote by $U(\sigmaD,\rhoA) \coloneqq \sum_{e\in \E} \rhoA_e\, U(\sigmaD,e)$ the expected number of undetected attacks when $\D$ plays the \emph{mixed} inspection strategy $\sigmaD\in \DeltaD$ and $\A$ plays the \emph{marginal} attack strategy $\rhoA \in \ThetaA$. The reader can verify that equilibria of the game $\Gamma$ can be equivalently represented  by strategy profiles $(\sigmaDEq,\rhoAEq)\in \DeltaD \times \ThetaA$ satsifying $U(\sigmaDEq,\rhoA) \leq U(\sigmaDEq,\rhoAEq) \leq U(\sigmaD,\rhoAEq)$ for every $\sigmaD \in \DeltaD$ and $\rhoA \in \ThetaA$, and that the value of the game is $U(\sigmaDEq, \rhoAEq)$.  Similarly, $(\sigmaDEq,\rhoAEq)$ is an equilibrium of $\Gamma$ if and only if $\sigmaDEq$ is an optimal solution of the minmax problem $\min_{\sigmaD \in \DeltaD} \max_{\rhoA \in \ThetaA} U(\sigmaD, \rhoA)$, and $\rhoAEq$ is an optimal solution of the maxmin problem $\max_{\rhoA \in \ThetaA} \min_{\sigmaD \in \DeltaD}  U(\sigmaD, \rhoA)$.



%% file: figures/example_network_model.tikz
    \definecolor{network_yellow}{RGB}{253,179,056}
    \definecolor{network_blue}{RGB}{047,103,177}
    \begin{tikzpicture}[node/.style={circle,fill=network_yellow,draw,minimum size=1em,inner sep=2pt]}, scale=0.8]
        \tikzmath{
            \m = 1; 
            \n{1} = 2; \n{2} = 1; \n{3} = 3; \n{4} = 1;
            \x{1,1} = -1.5; \y{1,1} = 0;  
            \x{1,2} = 1.5 - 0.25; \y{1,2} = 0;  
            \x{2,1} = 4.5 -0.75; \y{2,1} = 0;  
            \x{3,1} = 5.3; \y{3,1} = -0.7;  
            \x{3,2} = 6; \y{3,2} = 0.7;  
            \x{3,3} = 6.7; \y{3,3} = -0.7;  
            \x{4,1} = 8.5; \y{4,1} = -0.5;  
        }
        
        \draw[thick, fill=network_blue, fill opacity=0.2] (0, 0) circle [x radius=2cm, y radius=1cm] {};
        
        \node[] (v1) at (0, 0) {$v_{1}$};
        
        \draw[thick, fill=network_blue, fill opacity=0.2] (3 - 0.5, 0) circle [x radius=2cm, y radius=1cm] {};
        
        \node[] (v2) at (3 -0.5, 0) {$v_{2}$};
        
        \draw[thick, fill=network_blue, fill opacity=0.2] (7-0.5, 0) circle [x radius=3.5cm, y radius=1.8cm] {};
        
        \node[] (v3) at (7 -0.5, 0) {$v_{3}$};
        
        \draw[thick, fill=network_blue, fill opacity=0.2] (8.5, 0) circle [x radius=1.2cm, y radius=1.1cm] {};
        
        \node[] (v4) at (8.5, 0.25) {$v_{4}$};

        \node[node] (e1) at (\x{1,1}, \y{1,1}) {$e_{1}$};
        
        \node[node] (e2) at (\x{1,2}, \y{1,2}) {$e_{2}$};
        
        \node[node] (e3) at (\x{2,1}, \y{2,1}) {$e_{3}$};
        
        \node[node] (e4) at (\x{3,1}, \y{3,1}) {$e_{4}$};
        
        \node[node] (e5) at (\x{4,1}, \y{4,1}) {$e_{5}$};
        
        \node[node] (e6) at (\x{3,2}, \y{3,2}) {$e_{6}$};
        
        \node[node] (e7) at (\x{3,3}, \y{3,3}) {$e_{7}$};
        
\end{tikzpicture}

%% file: 3_exact_solution_methods.tex
\section{Exact Solution Method}
\label{sec:exact equilibrium strategies}


 In this section, we derive an exact solution approach for solving the game $\Gamma$. We recall that equilibrium strategies of $\Gamma$ are optimal primal and dual solutions of \eqref{LP: LP1}, which is challenging to solve due to its combinatorial number of variables and constraints. However, by leveraging the compact representation of $\A$'s strategies described in Section~\ref{sec:preliminaries}, we can formulate an alternative LP with a significantly smaller number of constraints.
 
\begin{proposition}
\label{prop:LP1 reformulation}
\eqref{LP: LP1} is equivalent to:
\begin{align}
    \label{LP: LP1 reformulation}
    \tag{LP($\AD$)}
    \begin{alignedat}{4}
        & \underset{\sigmaD \in \R^{\AD},\, \lambda \in \R^{\E},\, \gamma \in \R}{\min} & & \rA \gamma + \sum_{e\in \E}\lambda_{e} \\
        & \textnormal{subject to}& \quad & \begin{aligned}[t]
                                            \gamma + \lambda_e & \geq \sum_{S\in \AD}\sigmaD_S\, u(S, e) & \forall\, & e\in \E,\\
                                            \sum_{S \in \AD}\sigmaD_S & = 1 & \quad  &  \\
                                            \sigmaD_S & \geq 0 & \forall\, & S\in \AD,\\
                                            \lambda_e & \geq 0 & \forall\, & e\in \E,\\
                                            \gamma & \geq 0. &  &
                                        \end{aligned}
    \end{alignedat}
\end{align}


In particular, let $(\sigmaDEq, \lambda^*,\gamma^*) \in \R^{\AD}\times \R^{\E}\times \R$ be an optimal solution of \eqref{LP: LP1 reformulation}, and let $\rhoAEq \in \R^{\E}$ be the vector of optimal dual variables associated with its first set of constraints. Then, $(\sigmaDEq, \rhoAEq)$ is an equilibrium of the game $\Gamma$, and the value of the game is given by $\rA \gamma^* + \sum_{e\in \E} \lambda^*_e$.

%

\end{proposition}

From Proposition \ref{prop:LP1 reformulation}, we observe that mixed inspection strategies and marginal attack strategies in equilibrium of $\Gamma$ are respectively optimal primal and dual solutions of an LP with $|\AD| + m + 1$ variables and $m + 1$ constraints. Notably, the size of \eqref{LP: LP1 reformulation} is independent of the number of attack resources $\rA$. This is a consequence of the alternative representation of the attacker's strategies in terms of their marginal probabilities, which, as feasible solutions of $\ThetaA$, only require $m+1$ constraints for their description.

We can solve \eqref{LP: LP1 reformulation} using a column generation algorithm, provided that we can derive a relatively small formulation of the pricing problem. To this aim, we consider a subset of $\D$'s pure actions $\I \subseteq \AD$ and define the associated restricted master problem as {\RMP{$\I$}}.
Let $(\rhoA, \nu)\in \R^{\E} \times \R$ 
be optimal dual variables associated with the first set of constraints and second constraint of {\RMP{$\I$}}, respectively. One can easily check that $\rhoA\in \ThetaA$, and that it corresponds to the marginal attack probabilities of $\A$'s best response to the optimal solution of {\RMP{$\I$}}. Consequently, the reduced costs of the variables $\gamma$ and $\lambda_e$ are nonnegative. On the other hand, the reduced costs of the variables $\sigmaD_S$ are given by $\bar{c}_{S} = -\nu + \sum_{e \in \E}\rhoA_{e} \,u(S,e) = -\nu + U(S, \rhoA)$ for every $S \in \AD$.
Therefore, the pricing problem is the following:
\begin{align}
    \label{pb:DBR}
    \text{Find } S^* \in \underset{S\in \AD}{\argmin} \; U(S,\rhoA).    \tag{$\textnormal{DBR}$}
\end{align}
It aims to find a deterministic detector positioning that minimizes the expected number of undetected attacks against $\A$ playing the marginal attack strategy $\rhoA$. However, solving this pure best response problem for $\D$ is challenging due to its combinatorial nature and the nonlinearity of the undetection function $U$, even for instances where the monitoring sets exhibit minimal overlap:
%
%
%
%
%
\begin{proposition}
    \label{prop:NP_hardness_Defender_BR_2}
    Problem~\eqref{pb:DBR} is \textnormal{NP}-hard, even if $\rA=1$, the detection probabilities are homogeneous, and every component belongs to at most two monitoring sets. 
\end{proposition}

The proof of Proposition~\ref{prop:NP_hardness_Defender_BR_2} follows by reduction from the NP-hard problem VERTEX COVER.  Despite this negative result, an exact solution for Problem~\eqref{pb:DBR} can be computed using a mixed-integer program (MIP). For this purpose, we set an arbitrary enumeration of the detector locations, denoted as $\V = \{v_1,\ldots,v_n\}$. Then, for every location $v_i \in \V$ we define a binary variable $x_i$ as
\begin{align*}
    x_{i} = 
    \begin{cases}
        1 & \text{if a detector is positioned at location $v_i$}, \\
        0 & \text{otherwise}.
    \end{cases}
\end{align*}
Thus, if $S = \{ v_i \in \V:\, x_i=1 \}$ is a detector positioning determined by the variables $x$, we have $S\in \AD$ if and only if $\sum_{v_i \in \V} x_i \leq \rD$. Next, we must linearize the term $u(S,e)=\prod_{v_i\in S:\, e\in \C_{v_i}} \left( 1-p_{v_i} \right)$ in the undetection function. To this end, we define the following decision variables:
\begin{align*}
    u_{e,i} = \prod_{v_j\in \V:\, j\leq i} \hspace{-0.5em} \left( 1-p_{v_j} \mathds{1}_{ \C_{v_j} }(e) x_{j} \right) \quad \forall\, e\in \E, \, \forall\, v_i\in \V,
\end{align*}
where for every $e\in \E$ and $v_{j}\in \V$, $\mathds{1}_{ \C_{v_j} }(e) \in \{0,1\}$ is equal to 1 if and only if $e\in \C_{v_j}$. Then, $u(S,e) = u_{e,n}$ for all $e\in \E$. Finally, we can linearize the variables $u_{e,i}$ by leveraging the following relations derived from their definition:
\begin{align*}
    \left( 1-p_{v_1} \mathds{1}_{ \C_{v_1} }(e) x_{1}\right) &= u_{e,1}  &\forall\, e\in \E,\\
    u_{e,i} \left( 1-p_{v_{i+1}} \mathds{1}_{ \C_{v_{i+1}} }(e) x_{i+1}\right) &= u_{e,i+1}   &\forall\, e\in \E, \, \forall\, v_i \in \V \setminus \left\{v_n\right\}, \\
    0 &\leq u_{e,i} \leq 1  &\forall\, v_i \in \V.
\end{align*}


This results in the following reformulation of \eqref{pb:DBR}:

\begin{proposition}
    \label{prop:MIP_optimality}
    Problem~\eqref{pb:DBR} can be formulated as the following MIP:
    \begin{align}
        \label{MIP:pricing problem}
        \tag{MIP}
        \begin{alignedat}{3}
            & \min_{x\in \R^{n},\, u \in \R^{\E \times \{1,\ldots,n\} } } \quad && \sum_{e\in \E} \rhoA_e\, u_{e,n} \\
            & \text{subject to} \quad && \sum_{v_i\in \V} x_{i} \leq \rD  \\
            &  && 1-p_{v_1} \mathds{1}_{ \C_{v_1} }(e) x_{1} \leq u_{e,1} \quad & \forall\, e\in \E,\\
            &  && u_{e,i} \left(1-p_{v_{i+1}} \mathds{1}_{ \C_{v_{i+1}} }(e) \right) \leq u_{e,i+1} \quad & \forall\, e\in \E, \, \forall \,  v_i \in \V \setminus \left\{ v_n \right\},\\
            &  && u_{e,i} - x_{i+1} \leq u_{e,i+1} \quad & \forall\, e\in \E, \, \forall \, v_i \in \V \setminus \left\{ v_n \right\},\\
            &  && 0 \leq u_{e,i} \leq 1 \quad &  \forall \,  v_i  \in \V,\\
            &  && x_{i} \in \{0,1\} & \forall\, v_i \in \V.
        \end{alignedat}
    \end{align}
    In particular, let $(x^{*}, u^{*})\in \R^{n} \times \R^{\E \times \{1,\ldots,n\} }$ be an optimal solution of \eqref{MIP:pricing problem}. Then, $S^{*} \coloneqq \{ v_i \in \V:\, x^*_i=1\}$ is an optimal solution of \eqref{pb:DBR}, and $U(S^*, \rhoA) = \sum_{e\in \E} \rhoA_{e} \, u^{*}_{e,n}$. 
\end{proposition}


Notably, \eqref{MIP:pricing problem} comprises $O(mn)$ variables and constraints, and can be solved to determine the variable $\sigmaD_S$ of the master problem \eqref{LP: LP1 reformulation} with lowest reduced cost. Thus, NE of $\Gamma$ can be computed by solving \eqref{LP: LP1 reformulation} with a column generation algorithm that solves \eqref{MIP:pricing problem} for the pricing problem.
%
%
%
%
%

%% file: 4_approximate_solution_methods.tex
\section{Approximate Solution Methods}
\label{sec:approximate equilibrium strategies}
We now delve into the computation of approximate equilibrium strategies for the game $\Gamma$. We begin by analyzing the structure of Problem \eqref{pb:DBR} and exploit the value of approximation algorithms. We then propose two solution approaches that leverage these approximations: one utilizing column generation, and the other employing the multiplicative weights update algorithm.

\subsection{Approximate Solutions for Defender's Best Response Problem}
\label{subsec:approximation_algorithms_for_DBR}
We first examine structural properties of Problem \eqref{pb:DBR} that will allow us to efficiently compute approximate best responses for a large class of instances. 
\begin{lemma}
    \label{prop:supermodularity}
    For every fixed $\rhoA \in \ThetaA$, the set function $U(S,\rhoA) \coloneqq \sum_{e\in \E} \rhoA_e \, u(S, e)$, defined for every $S\subseteq \V$, is nonincreasing and supermodular.
\end{lemma}
Therefore, for a fixed $\A$'s marginal attack strategy $\rhoA\in \ThetaA$, $\D$'s best response problem \eqref{pb:DBR} consists in minimizing the nonnegative and nonincreasing supermodular function $U(\cdot, \rhoA)$, subject to a cardinality constraint given by $S\in \AD$. Let $S^{*}\in \AD$ be an optimal solution of \eqref{pb:DBR}. For $\alpha\geq 1$, a detector positioning $\widehat{S} \in \AD$ is an $\alpha$-approximate best response for $\D$ if it satisfies $U(S^*,\rhoA) \leq  U(\widehat{S},\rhoA) \leq \alpha  U(S^*,\rhoA)$.
To approximate \eqref{pb:DBR}, \citet{ilev2001approximation} and \citet{sviridenko2017optimal} showed that a \emph{reverse} greedy algorithm---also referred to as greedy descent or stingy algorithm---finds a solution whose approximation guarantee depends on the \emph{curvature} of the function $U(\cdot,\rhoA)$, defined as $c \coloneqq 1 -  \min_{\{v \in \V:\,  U_{\emptyset}(v, \rhoA) > 0\}}\,  U_{\V \setminus \{v\}}(v, \rhoA)  / U_{\emptyset}(v, \rhoA) $,
where $U_{S}(v, \rhoA) \coloneqq  U(S, \rhoA) - U(S \cup \{v\}, \rhoA)$ is the marginal decrease in the expected number of undetected attacks when $v\in \V$ is added to the set of detector locations $S\subseteq \V$. We let $c\coloneqq 0$ if $U_{\emptyset}(v, \rhoA) = 0$ for every $v\in \V$. The curvature parameter $c$ lies in the interval $[0,1]$ and measures how far $U(\cdot, \rhoA)$ is from being an additive set function. In particular, when $c=0$, the function $U(\cdot, \rhoA)$ is additive.

In practical scenarios, it is often the case that each component can be inspected from a small number of detector locations. In such cases, we can provide the following upper bound on the curvature of $U(\cdot, \rhoA)$.
\begin{lemma}
    \label{lem:upper_bound_curvature}
    The curvature parameter $c$ of the function $U(\cdot, \rhoA)$ satisfies $c \leq 1 - (1- \max_{v \in \V} p_v  )^{d}$, where $d$ is the maximum number of locations that can monitor a component.
\end{lemma}

Starting from the full set of detector locations $\V$, each iteration of the reverse greedy algorithm selects a location that minimizes the marginal decrease in the expected number of undetected attacks that would result from its removal. Then, it removes such location from the current solution. This process is repeated until $\rD$ locations are left, resulting in a feasible detector positioning. We refer to Algorithm \ref{alg:reverse_greedy} for a pseudocode of the reverse greedy algorithm.

\begin{algorithm}[!htbp]
\DontPrintSemicolon
    Initialize $S \gets \V$ \;
    \While{$|S| > \rD$}{
        $v \in \underset{w \in \V }{\argmin} \, U_{S \setminus \{w\} }(w, \rhoA) $ \;
        $S \gets S \setminus\left\{ v\right\}$
    }
    \Return{S}
\caption{Reverse Greedy Algorithm for Approximate $\D$'s Best Response}
\label{alg:reverse_greedy}
\end{algorithm}

The reverse greedy algorithm exhibits the following approximation guarantee.
\begin{theorem}[\citet{sviridenko2017optimal}]
    \label{thm:reverse_greedy_guarantee}
    The reverse greedy algorithm (Algorithm~\ref{alg:reverse_greedy}) returns a detector positioning $\widehat{S} \in \AD$ satisfying $U(\widehat{S}, \rhoA) \leq \left( 1 / (1-c) \right) U(S^*, \rhoA)$. 
\end{theorem}


Therefore, by applying the bound provided in \mbox{Lemma~\ref{lem:upper_bound_curvature}} to \mbox{Theorem~\ref{thm:reverse_greedy_guarantee}}, we derive the following instance-dependent approximation result.



\begin{corollary}
    \label{cor:approximation_factor}
    The reverse greedy algorithm (Algorithm~\ref{alg:reverse_greedy}) returns a detector positioning $\widehat{S} \in \AD$ satisfying $U(\widehat{S}, \rhoA) \leq \left( 1 / \left(1 - \max_{v \in \V} p_v \right)^{d} \right) U(S^*, \rhoA)$, where $d$ represents the maximum number of locations that can monitor a component. 
\end{corollary}

From Theorem~\ref{thm:reverse_greedy_guarantee}, the reverse greedy algorithm finds a detector positioning for which the expected number of undetected attacks against $\A$'s marginal attack strategy $\rhoA$ is within a factor of $1 / (1-c)$ of that of $\D$'s best response. This approximation factor is strictly increasing in the curvature parameter $c$, and it degrades as $c$ approaches $1$. In the limit $c=1$, we can show that Problem \eqref{pb:DBR} becomes notably challenging to approximate.

\begin{proposition}
    \label{prop:inapproximability}
    If $c=1$, then, unless $\textnormal{P} = \textnormal{NP}$, there is no polynomial time approximation algorithm for \eqref{pb:DBR}.
\end{proposition}

In scenarios with expensive detection resources, the number of available detectors $\rD$ may be significantly smaller than the number of locations $n$. In such instances, the reverse greedy algorithm needs to remove a large number of locations from $\V$ before termination. This can potentially lower the performance of solution approaches that frequently utilize reverse greedy as a subroutine. To mitigate this issue, we also consider the \emph{forward} greedy algorithm as an approximation method for \eqref{pb:DBR}. This algorithm begins with an empty set of detector locations, and in each iteration, selects a location that maximizes the marginal decrease in the expected number of undetected attacks, subsequently adding it to the current solution, until $\rD$ locations are selected. We refer to Algorithm~\ref{alg:forward_greedy} for a pseudocode of the forward greedy algorithm. 
\begin{algorithm}[!htbp]
\DontPrintSemicolon
    Initialize $S \gets \emptyset$ \;
    \While{$|S| < \rD$}{
        $v \in \underset{w \in \V }{\argmax} \, U_{S}(w, \rhoA) $ \;
        $S \gets S \cup \left\{ v\right\}$
    }
    \Return{S}
\caption{Forward Greedy Algorithm for Approximate $\D$'s Best Response}
\label{alg:forward_greedy}
\end{algorithm}



By leveraging the approximation guarantee of \citet{conforti1984submodular} for submodular maximization with cardinality constraints, we can show that the forward greedy algorithm provides the following  guarantee:

\begin{proposition}
    \label{prop:forward_greedy_guarantee}
    The forward greedy algorithm (Algorithm~\ref{alg:forward_greedy}) returns a detector positioning $\widehat{S}\in \AD$ satisfying $U(\widehat{S}, \rhoA) \leq \left( \frac{1 - \exp(-c)}{c} \right) U(S^*, \rhoA) + \left(1 - \frac{1 - \exp(-c)}{c} \right)  \rA$.
\end{proposition}

The approximation guarantee of Proposition~\ref{prop:forward_greedy_guarantee} is weaker than that of Theorem~\ref{thm:reverse_greedy_guarantee}. Notably, when the optimal value of \eqref{pb:DBR} is zero, the reverse greedy algorithm returns an optimal solution, whereas the forward greedy algorithm may yield a solution whose relative approximation error is unbounded.


\subsection{Column Generation with Approximate Best Response}
\label{subsec:colgen with approx best response}
The exact column generation algorithm of Section~\ref{sec:exact equilibrium strategies} requires solving \eqref{MIP:pricing problem} in each iteration to compute optimal solutions for the pricing problem \eqref{pb:DBR}, which may be computationally expensive. Alternatively, we can utilize the forward and reverse greedy algorithms to obtain fast approximate solutions to \eqref{pb:DBR}. We then show that the multiplicative approximation guarantee is inherited by the strategy profile returned by the column generation algorithm. Additionally, we can enable early termination once the reduced cost of the variable of \eqref{LP: LP1 reformulation} associated with the approximate best response is sufficiently close to zero, at the cost of an additional small additive error. 

\begin{theorem}
\label{thm:column_generation_with_approx_best_response}
    Let $\varepsilon \geq 0$. Consider a solution $(\sigmaDHat, \widehat{\lambda}, \widehat{\gamma}) \in \R^{\AD} \times \R^{\E} \times \R$ of \eqref{LP: LP1 reformulation}, together with a vector of dual variables $\rhoAHat \in \R^{\E}$ associated with its first set of constraints, attained through a column generation algorithm that solves in each iteration  the restricted master problem {\RMP{$\I$}} and computes an $\alpha$-approximation $\widehat{S} \in \AD$ for \eqref{pb:DBR}, and that terminates once the reduced cost of the variable $\sigmaD_{\widehat{S}}$ satisfies $\bar{c}_{\widehat{S}} \geq -\varepsilon$. Then, the strategy profile $(\sigmaDHat, \rhoAHat)$ satisfies $U(\sigmaDHat, \rhoAHat) = \rA \widehat{\gamma} + \sum_{e\in \E} \widehat{\lambda}_e$ and
    \begin{align*}
        U(\sigmaDHat, \rhoA) \leq U(\sigmaDHat, \rhoAHat) \leq \alpha \, U(\sigmaD, \rhoAHat) + \varepsilon, \quad \forall\, \sigmaD \in \DeltaD,\; \forall\, \rhoA \in \ThetaA.
    \end{align*}
    Furthermore, $\sigmaDHat$ and $\rhoAHat$ respectively satisfy the following bounds for the worst-case expected number of undetected attacks with respect to the value of the game:
    \begin{align*}
        U(\sigmaDEq, \rhoAEq) &\leq \underset{\rhoA \in \ThetaA}{\max} U(\sigmaDHat, \rhoA) \leq \alpha \, U(\sigmaDEq, \rhoAEq) + \varepsilon, \\
        \frac{1}{\alpha} \left( U(\sigmaDEq, \rhoAEq) - \varepsilon \right) &\leq \underset{\sigmaD \in \DeltaD}{\min} U(\sigmaD, \rhoAHat) \leq U(\sigmaDEq, \rhoAEq).
    \end{align*}
\end{theorem}

The strategy profile $(\sigmaDHat, \rhoAHat)$ returned by the column generation algorithm with $\alpha$-approximate best response for \eqref{pb:DBR} described in Theorem~\ref{thm:column_generation_with_approx_best_response} is such that $\D$ can decrease the expected number of undetected attacks by a factor of at most $1/\alpha - \varepsilon / U(\sigmaDHat, \rhoAHat)$ through unilateral deviation from $(\sigmaDHat, \rhoAHat)$, while $\A$'s strategy $\rhoAHat$ is already a best response to $\sigmaDHat$. We observe that the parameter $\varepsilon \geq 0$ used for the early termination criterion translates into an additive error in the approximation guarantees. In particular, when $\varepsilon = 0$, $\sigmaDHat$ is an $\alpha$-approximation of the optimal inspection strategy $\sigmaDEq$, whereas $\rhoAHat$ is an $(1/\alpha)$-approximation of the optimal marginal attack strategy $\rhoAEq$. When detection is imperfect, Corollary~\ref{cor:approximation_factor} guarantees an approximation factor of $\alpha \leq 1 / \left(1 - \max_{v \in \V} p_v \right)^{d}$, where $d$ is the maximum number of locations that can monitor a component.

Interestingly, the termination criterion stated in Theorem~\ref{thm:column_generation_with_approx_best_response} only requires the reduced cost of the variable in \eqref{LP: LP1 reformulation} associated with the $\alpha$-approximate best response $\widehat{S}$ to be at least $-\varepsilon$. In particular, using strong duality arguments, we can show that it is unnecessary to require that $\widehat{S}$ already belong to the set of columns $\I$ upon termination, a condition stipulated by \citet{guo2016optimal}. This observation potentially leads to significantly faster convergence in practice.

\subsection{Multiplicative Weights Update with Approximate Best Response}
\label{subsec:mwu with projection step}

While the column generation algorithm with approximate best response of Section~\ref{subsec:colgen with approx best response} provides approximate equilibrium strategies for the game $\Gamma$, its running time may remain exponential due to the potential need to add an exponential number of columns before convergence. We now propose a second solution approach for computing approximate equilibria of $\Gamma$ in polynomial time. Here, we adapt the method introduced by \citet{krause2011randomized}, which employs the Multiplicative Weights Update (MWU) algorithm in conjunction with an approximate best response oracle to compute randomized sensing strategies for a game with submodular structure.

Although our game shares similarities with that of \citeauthor{krause2011randomized}, a key distinction makes ours more intricate to solve. In their setting, it is assumed that the adversary has a polynomial number of pure strategies, allowing each iteration of the MWU algorithm---which involves updating the adversary's mixed strategy vector---to be implementable in polynomial time. In contrast, in our game, the adversary $\A$ has $|\AA| = \sum_{k=0}^{\rA} \binom{m}{k}$ pure strategies, a number that grows combinatorially with $\rA$, rendering the update step of the adversary's strategies computationally expensive. To address this challenge, we propose a novel implementation of the MWU algorithm that updates $\A$'s marginal attack probabilities, instead of their original mixed strategies. 

We initialize our implementation of the MWU algorithm by setting uniform marginal attack probabilities, namely, $\rho^{(1)}_e = \rA/m$ for every $e\in \E$. Each iteration begins with a vector of marginal attack probabilities $\rho^{(t)} \in \Theta_{\A}$ for $\A$, after which $\D$ computes an approximate best response $S^{(t)}\in \AD$. Next, $\A$ uses $S^{(t)}$ to update $\rho^{(t)}$ to a new marginal attack strategy $\rho^{(t+1)}$. This update is performed in two steps: First, an intermediate updated vector $\widetilde{\rho}^{\,(t+1)}$ is constructed by applying the multiplicative weights rule, that is, $\widetilde{\rho}^{\,(t+1)}_e \gets \rho^{(t)}_e \exp\left(\eta u(S^{(t)}, e)\right)$ for every $e\in \E$, where $\eta>0$ is a parameter of the algorithm. This step aims to increase the marginal attack probabilities for components with a higher undetection probability under $S^{(t)}$. However, the resulting vector $\widetilde{\rho}^{\,(t+1)}$ may not be a feasible solution of $\ThetaA$, and therefore, may not be the marginal distribution resulting from some mixed strategy in $\DeltaA$ (Lemma~\ref{lem:convex_hull}). To rectify this issue, we proceed to the second part of the update step, wherein we project $\widetilde{\rho}^{\,(t+1)}$ onto $\ThetaA$ with respect to the unnormalized relative entropy, which is defined as
\begin{align*}
    \DRE{\rho}{ \widetilde{\rho}\,} \coloneqq
            \sum_{e\in \E} \left( \rho_e \ln \frac{ \rho_e }{ \widetilde{\rho}_e } + \widetilde{\rho}_e - \rho_e \right),
        \quad \forall\, \rho \in \R^{\E}_{> 0}, \forall\, \widetilde{\rho} \in \R^{\E}_{> 0}.
\end{align*}


The unnormalized relative entropy is chosen as the distance measure in the projection step because it aligns naturally with the multiplicative weights update rule. By strongly penalizing deviations near zero, it preserves the positivity of all marginal probabilities and avoids premature loss of support. In contrast, metrics such as the Euclidean distance can distort probability mass and drive some coordinates of the projection to vanish, which may trap the dynamics in degenerate strategies and prevent convergence to the equilibrium distribution.


We refer to Algorithm \ref{alg:MWU} for a pseudocode of our implementation of the MWU algorithm.
\begin{algorithm}[htbp]
\DontPrintSemicolon
    \SetKwInOut{Parameters}{Parameters}
    \Parameters{$\tau \in \mathbb{Z}_{>0}$: number of iterations, $\eta \in \R_{>0}$: multiplicative weights factor}
    Initialize $\rho^{(1)}_e \gets \frac{\rA}{m}$, \ $\forall \, e \in \E$ \;
   \For{$t=1,\ldots,\tau$}{
         Let $S^{(t)}\in \AD$ satisfying  $U(S^{(t)}, \rho^{(t)}) \leq \alpha \, \underset{S\in \AD}{\min} \, U(S, \rho^{(t)}) $ \; \label{alg:best_response_step}
         \For{$e \in \E$}{
            \label{alg:start_update_step}
            $\widetilde{\rho}^{\,(t+1)}_e \gets \rho^{(t)}_e \exp\left(\eta u(S^{(t)}, e)\right)$
         }
         \label{alg:end_update_step}
         Let $\rho^{(t+1)} \in  \underset{\rho \in \ThetaA}{\argmin}\, \DRE{ \rho }{ \widetilde{\rho}^{\,(t+1)} }$ \label{alg:projection_step}
   }
    \Return{$\sigmaDHat \coloneqq \frac{1}{\tau}\sum_{t=1}^{\tau} \mathds{1}_{S^{(t)}}$ and $\rhoAHat \coloneqq \frac{1}{\tau}\sum_{t=1}^{\tau} \rho^{(t)}$}
\caption{Multiplicative Weights Update with Approximate Best Response}
\label{alg:MWU}
\end{algorithm}

The MWU algorithm (Algorithm~\ref{alg:MWU}) returns a mixed inspection strategy $\sigmaDHat \coloneqq \frac{1}{\tau}\sum_{t=1}^{\tau} \mathds{1}_{S^{(t)}}$, where $ \mathds{1}_{S^{(t)}} \in \{0,1\}^{\AD}$ is the characteristic vector of $S^{(t)}$---that is, for every $S\in \AD$, $ \mathds{1}_{S^{(t)}}(S)=1$ if and only if $S=S^{(t)}$. In other words, $\sigmaDHat_{S}$ is equal to the relative frequency with which $S$ was selected as $\D$'s best response (Line~\ref{alg:best_response_step}) throughout the execution of the algorithm. Therefore, the size of the support of $\sigmaDHat$ is at most  the number of iterations $\tau$. The MWU algorithm also returns a marginal attack strategy $\rhoAHat \coloneqq \frac{1}{\tau}\sum_{t=1}^{\tau} \rho^{(t)}$, given by the average marginal attack probabilities obtained after completing each update step (Line~\ref{alg:projection_step}). These strategies satisfy the following approximation guarantees:

\begin{theorem}
    \label{thm:MWU_epsilon_NE}
    Let $\varepsilon>0$. Then, after $\tau \geq 4\rA^2 \max\left\{ \ln  (m /\rA), 1\right\} / \varepsilon^2$ iterations, Algorithm~\ref{alg:MWU} with $\eta = \sqrt{\max\left\{ \ln  (m /\rA), 1\right\} / \tau}$ returns a strategy profile $(\sigmaDHat, \rhoAHat) \in \DeltaD \times \ThetaA$ satisfying 
    \begin{align*}
        \frac{1}{\alpha} \left( U(\sigmaDHat,\rhoA ) - \varepsilon \right) \leq U(\sigmaDHat,\rhoAHat) \leq \alpha U(\sigmaD,\rhoAHat ) + \varepsilon,\quad \forall\, \sigmaD \in \DeltaD, \; \forall\, \rhoA \in \ThetaA.
    \end{align*}
    Furthermore, $\sigmaDHat$ and $\rhoAHat$ respectively satisfy the following bounds for the worst-case expected number of undetected attacks with respect to the value of the game:
    \begin{align*}
        U(\sigmaDEq, \rhoAEq) \leq \underset{\rhoA \in \ThetaA}{\max}\, U(\sigmaDHat, \rhoA) &\leq \alpha U(\sigmaDEq, \rhoAEq) + \varepsilon, \\
        \frac{1}{\alpha} \left( U(\sigmaDEq, \rhoAEq) - \varepsilon \right) \leq \underset{\sigmaD \in \DeltaD}{\min}\, U(\sigmaD, \rhoAHat) &\leq U(\sigmaDEq, \rhoAEq).
    \end{align*}   
\end{theorem}


From Theorem~\ref{thm:MWU_epsilon_NE}, the strategy profile $(\sigmaDHat, \rhoAHat)$ returned by the MWU algorithm with $\alpha$-approximate best response (Algorithm~\ref{alg:MWU}) offers similar approximation guarantees to those provided by the column generation algorithm with approximate best response from Theorem~\ref{thm:column_generation_with_approx_best_response}. However, unlike the latter, in this case, $\A$'s marginal attack strategy $\rhoAHat$ is not necessarily a best response to $\sigmaDHat$. Instead, $\A$ can increase the expected number of undetected attacks by at most a factor of $\alpha + \varepsilon / U(\sigmaDHat, \rhoAHat)$ through unilateral deviation from $(\sigmaDHat, \rhoAHat)$. Moreover, the MWU algorithm cannot be executed with $\varepsilon=0$, as it requires $\Omega(1 / \varepsilon^2)$ iterations to achieve convergence.


We next delve into the projection step of Algorithm~\ref{alg:MWU} (Line~\ref{alg:projection_step}). Given $\widetilde{\rho} \in \R^{\E}_{> 0}$, the projection problem consists in finding a vector $\rho\in \ThetaA$ that minimizes the unnormalized relative entropy with respect to $\widetilde{\rho}$, given by $ \DRE{\rho}{ \widetilde{\rho}\,}$. This is a strictly convex minimization problem, which renders the projection unique. In the following theorem, we show that this projection problem has a closed-form solution.

\begin{theorem}
    \label{thm:projection_closed_form}
    Let $\widetilde{\rho} \in \R^{\E}_{> 0}$, and let $\rho^{*} \in \R^{\E}$ denote its projection onto $\ThetaA$ with respect to the unnormalized relative entropy. 

        If $\sum_{e\in \E} \min\{ \widetilde{\rho}_e, 1\} \leq \rA$, then $\rho^{*}$ is given by $\rho^{*}_e = \min\left\{\widetilde{\rho}_e,\, 1 \right\}$ for every $e\in \E$.

        If $\sum_{e\in \E} \min\{ \widetilde{\rho}_e, 1\} > \rA$, we sort the components of $\E$ such that $\widetilde{\rho}_{e_1} \geq \cdots \geq \widetilde{\rho}_{e_m}$, breaking ties arbitrarily, and define the parameter $k^* \coloneqq \max\left\{ k \in \{0,\ldots,\rA\} :\, k + \left( 1 / \widetilde{\rho}_{e_k} \right) \sum_{j=k+1}^{m} \widetilde{\rho}_{e_j}  \leq \rA \right\}$,
        where we let $\widetilde{\rho}_{e_0} \coloneqq +\infty$. Then, $\rho^{*}$ is given by $\rho^{*}_{e} = \min \left\{ \left( \rA-k^* \right) \widetilde{\rho}_{e} / \sum_{j=k^*+1}^{m} \widetilde{\rho}_{e_j}   ,\, 1 \right\}$ for every $e\in \E$.
\end{theorem}

From Theorem~\ref{thm:projection_closed_form}, we obtain an analytic characterization for the projection of a vector $\widetilde{\rho} \in \R^{\E}_{> 0}$ onto $\ThetaA$ under the unnormalized relative entropy. In general, this projection is given by a truncated scaling of $\widetilde{\rho}$, namely $\rho^{*}_e = \min\left\{ \mu \widetilde{\rho}_e, 1 \right\}$ for every $e\in \E$, where the scaling factor $\mu$ is determined by two cases. If $\sum_{e\in \E} \min\{ \widetilde{\rho}_e, 1\} \leq \rA$, the truncated scaling with $\mu = 1$ is already an optimal solution of the projection problem. Indeed, $\DRE{\cdot}{\widetilde{\rho}\,}$ is separable, and its $e$-th summand is minimized in the interval $[0,1]$ by $\min\left\{ \widetilde{\rho}_e, 1 \right\}$. On the other hand, if $\sum_{e\in \E} \min\{ \widetilde{\rho}_e, 1\} > \rA$, the truncated scaling with $\mu=1$ becomes infeasible for $\ThetaA$. Furthermore, KKT conditions are not directly applicable to the projection problem, due to $\DRE{\cdot}{\widetilde{\rho}\,}$ being nonsmooth. To handle this case, we use a Lagrangian dual approach to determine that the appropriate scaling factor is given by $\mu = (\rA - k^*) / \sum_{j=k^*+1}^{m} \widetilde{\rho}_{e_{j}}$.

As a direct consequence of Theorem~\ref{thm:projection_closed_form}, the projection of $\widetilde{\rho}$ onto $\ThetaA$ can be computed in time $O(m \log m)$, where the running time is dominated by sorting the elements of $\E$ to determine $k^*$. However, it is possible to compute the projection in linear time. The key insights are that the function $g(k) \coloneqq k + 1 / \widetilde{\rho}_{e_k}  \sum_{j=k+1}^{m}  \widetilde{\rho}_{e_j}$ arising in the definition of the parameter $k^*$ is nondecreasing in $k$, and that its evaluation is determined by a splitting component $e_k$ that induces a partition of $\E$ into three sets, namely $\F_{\textnormal{high}} \coloneqq  \left\{ e\in \E:\, \widetilde{\rho}_{e} > \widetilde{\rho}_{e_k}  \right\}$, $\F_{\textnormal{eq}} \coloneqq  \left\{ e\in \E:\, \widetilde{\rho}_{e} = \widetilde{\rho}_{e_k}  \right\}$, and $\F_{\textnormal{low}} \coloneqq  \left\{ e\in \E:\, \widetilde{\rho}_{e} < \widetilde{\rho}_{e_k}  \right\}$, so that $g(k) = |\F_{\textnormal{high}}| + |\F_{\textnormal{eq}}| + (1/\widetilde{\rho}_{e_k}) \sum_{e\in \F_{\textnormal{low}}} \widetilde{\rho}_e$.

This structure enables us to find $k^*$ using a tailored binary search algorithm without sorting $\E$. The algorithm maintains two auxiliary variables to account for the contributions of previously discarded components in the calculation of $g(k)$: $k$, which counts the number of discarded components whose value of $\widetilde{\rho}_e$ is greater than or equal to the current threshold $\widetilde{\rho}_{e_k}$; and $s$, which accumulates the total weight of discarded components whose value of $\widetilde{\rho}_e$ is less than or equal to this threshold.

Starting with initial values of $\F = \E$, $k=0$, and $s=0$, each iteration selects a splitting component $e_{k}$ and tests whether $g(k) \leq \rA$. The auxiliary variables $k$ and $s$ ensure that the current value of $g(k)$ can be obtained as $k + |\F_{\textnormal{high}} | + |\F_{\textnormal{eq}} | + \left( 1 /\widetilde{\rho}_{e} \right) \left( \sum_{e^{\prime} \in \F_{\textnormal{low}}} \widetilde{\rho}_{e^{\prime}} + s \right)$. If the test is satisfied, the next iteration selects the splitting component among the elements of $\F_{\textnormal{low}}$ (\ie, we set $\F = \F_{\textnormal{low}}$) and increases the current value of $k$ by $|\F_{\textnormal{high}} | + |\F_{\textnormal{eq}} |$; otherwise, it selects it among the elements of $\F_{\textnormal{high}}$ (\ie, we set $\F = \F_{\textnormal{high}}$) and increases the value of $s$ by $\widetilde{\rho}_e |\F_{\textnormal{eq}} | + \sum_{e^{\prime} \in \F_{\textnormal{low}}} \widetilde{\rho}_{e^{\prime}}$. In either case, the next splitting component is never chosen from $\F_{\textnormal{eq}}$, as doing so would yield the same value of $g(k)$ and therefore not advance the search. Upon termination, the auxiliary variables satisfy $k = k^{*}$ and $s = \sum_{j=k^{*}+1}^{m}\widetilde{\rho}_{e_j}$.

The splitting component can be identified in linear time using a selection algorithm such as the median-of-medians algorithm (\citet{blum1973time}), ensuring that each iteration discards a constant fraction of the remaining elements. Consequently, the overall binary search runs in $O(m)$ time. We refer to Algorithm~\ref{alg:fast_projection} for the pseudocode of the projection algorithm.


\begin{algorithm}[!htbp]
\DontPrintSemicolon
    \SetKwInOut{Input}{Input}
    \SetKwInOut{Output}{Output}
    \Input{A vector $\widetilde{\rho} \in \R^{\E}_{>0}$ and a resource budget $\rA \in \Z_{> 0}$.}
    \Output{A vector $\rho^{*} \in \R^{\E}$ satisfying $\rho^{*} \in  \underset{\rho \in \ThetaA}{\argmin}\, \DRE{ \rho }{ \widetilde{\rho\,} }$.}
    \If{$\sum_{e\in \E} \min\{ \widetilde{\rho}_e, 1\} \leq \rA$}{
        $\mu \gets 1$ \;
    }
    \Else{
        Let $\F \gets \E$,\, $k \gets 0$,\, $s \gets 0$ \;
       \While{$\F \neq \emptyset$ \label{alg:while_start}}{
            Let $e \gets$ index of $\ceil*{ \frac{|\F|}{2} }$-th largest entry of $\left( \widetilde{\rho}_e \right)_{e\in \F}$   \; 
            $\F_{\textnormal{high}} \gets \left\{e^{\prime} \in \F:\, \widetilde{\rho}_{e^{\prime}}  > \widetilde{\rho}_{e} \right\}$,\, $\F_{\textnormal{low}} \gets \left\{e^{\prime} \in \F:\, \widetilde{\rho}_{e^{\prime}}  < \widetilde{\rho}_{e} \right\}$, \, $\F_{\textnormal{eq}} \gets \left\{e^{\prime} \in \F:\, \widetilde{\rho}_{e^{\prime}}  = \widetilde{\rho}_{e} \right\}$   \; 
            \If{$k + |\F_{\textnormal{high}} | + |\F_{\textnormal{eq}} | + \left( 1 /\widetilde{\rho}_{e} \right) \left( \sum_{e^{\prime} \in \F_{\textnormal{low}}} \widetilde{\rho}_{e^{\prime}} + s \right) \leq \rA$} {
                    $\F \gets \F_{\textnormal{low}}$, \, $k \gets k + |\F_{\textnormal{high}} | + |\F_{\textnormal{eq}} |$ \;
            }
            \Else{
                $\F \gets \F_{\textnormal{high}}$, \, $s \gets s + \widetilde{\rho}_e |\F_{\textnormal{eq}} | + \sum_{e^{\prime} \in \F_{\textnormal{low}}} \widetilde{\rho}_{e^{\prime}}$ \;
            } \label{alg:while_end}
       }
       $\mu \gets (\rA - k) / s$ \;
    }
    \Return{$\rho^{*}_e \coloneqq \min\left\{ \mu\widetilde{\rho}_e ,\, 1\right\}$ \text{for all} $e\in \E$}
\caption{Linear-Time Relative Entropy Projection onto $\ThetaA$}
\label{alg:fast_projection}
\end{algorithm}

\begin{theorem}
    \label{thm:fast_projection}
    Given $\widetilde{\rho} \in \R^{\E}_{> 0}$, Algorithm~\ref{alg:fast_projection} computes the projection of $\widetilde{\rho}$ onto $\ThetaA$ with respect to the unnormalized relative entropy in time $O(m)$.
\end{theorem}


Thanks to Theorem~\ref{thm:fast_projection}, we can compute projections under the unnormalized relative entropy onto $\ThetaA$ in linear time. Consequently, the computational cost per iteration of Algorithm~\ref{alg:MWU} is determined by the cost of computing a defender's $\alpha$-approximate best response (Line \ref{alg:best_response_step})---which takes $O(n (n-\rD))$ (resp. $O(n\rD)$) evaluations of the undetection function if we use the reverse (resp. forward) greedy algorithm---and the cost of computing the projection of  $\widetilde{\rho}^{\,(t+1)}$ onto $\ThetaA$ with respect to the unnormalized relative entropy (Line \ref{alg:projection_step}), which takes time $O(m)$ using Algorithm~\ref{alg:fast_projection}. Finally, the convergence of  Algorithm~\ref{alg:MWU} requires $\rA^2 \max\left\{ \ln  (m /\rA), 1\right\} / \varepsilon^2$ iterations (Theorem~\ref{thm:MWU_epsilon_NE}), resulting in an overall running time polynomial in $m$, $n$, $\rA$, $\rD$, and $1/\varepsilon^2$.

%% file: 5_computational_study.tex
\section{Computational Study}
\label{sec:computational study}

We test our solution approaches for solving strategic inspection problems on several natural gas distribution networks across Europe and the United States. The instances from Europe were constructed from the SciGRID\_gas database (\citet{pluta2022scigrid_gas}), representing the countries of France (EU FR), Germany (EU DE), Italy (EU IT), Poland (EU PL), and Great Britain (EU GB). The instances from the United States were obtained from the open datasets of the Homeland Infrastructure Foundation-Level Data (\citet{hifld2024}), and represent the states of Texas (US TX), Oklahoma (US OK), Lousiana (US LA), New York (US NY), and California (US CA).

In this case study, we aim to find inspection strategies for monitoring these networks to detect potential attacks on their pipelines, which can be caused by malicious intervention on their cyber-physical systems or by unauthorized access and trespassing. We consider compressor stations and processing plants from the datasets as the detector locations $\V$, and the pipelines as the components $\E$ that the attacker is interested in targeting. To identify disruptions, such as gas leakages indicative of pipe bursts, the defender leverages detectors in the form of unmanned aerial systems equipped with infrared thermography technology to conduct pipeline patrols around each location. For each location, we define its monitoring set as the collection of pipelines intersecting a circular area of radius $50$ km around its geographic location, and set its detection probability randomly from a Uniform$[0.5, 1]$ distribution. All LPs and MIPs were solved using Gurobi 9.1.2 (\citet{gurobi}) on a 11th Gen Intel(R) Core(TM) i7-1195G7 @ 2.90GHz processor.

We begin by evaluating the performance of the exact column generation (CG) algorithm applied to \eqref{LP: LP1 reformulation}, using \eqref{MIP:pricing problem} to compute exact solutions to the pricing problem \eqref{pb:DBR}. For each network, we set $\rA$ as $2\%$ of the components, and $\rD$ as the minimum number of detectors capable of monitoring $80\%$ of the components, except for EU DE and EU FR, which due to their larger sizes, we adjusted $\rD$ to cover $70\%$ and $57\%$ of the components, respectively. 
The results are shown in Table~\ref{table:results_exact_column_generation}.

\begin{table}[htbp]
\caption{Computational results for CG on different gas distribution networks.}
\label{table:results_exact_column_generation}
\footnotesize
\centering
\setlength{\tabcolsep}{8pt}
\begin{tabular}{l *{6}{d{3.3}}} 
\toprule
\textbf{Network} & \textbf{$|\V|$} & \textbf{$|\E|$} & \textbf{$\rD$} & \textbf{$\rA$} & \mc{$U(\sigmaDEq, \sigmaAEq)$} & \mc{\textbf{Time [s]}} \\
\midrule
\begin{filecontents*}{data/table_results_column_generation_approx.csv}
network;nodes;components;b1;b2;valueExact;timeExact;valueCGSG;timeCGSG;valueCGRG;timeCGRG;valueMWUSG;timeMWUSG;valueMWURG;timeMWURG;relErrorCGSG;relErrorCGRG;relErrorMWUSG;relErrorMWURG
EU GB;103;116;13;3;1.68;1.56;1.68;0.25;1.70;0.22;1.70;3.05;1.72;8.64;0.00;1.13;1.07;2.18
US CA;51;616;9;13;9.54;2.77;9.57;0.72;9.63;0.97;9.65;9.69;9.73;15.33;0.39;1.00;1.23;1.98
US NY;47;1{,}731;7;35;27.07;8.95;27.16;3.45;27.25;4.31;27.36;24.33;27.38;38.58;0.32;0.67;1.07;1.15
US LA;196;2{,}680;11;54;39.67;209.16;39.98;32.89;40.22;37.28;40.01;58.16;40.37;331.95;0.79;1.39;0.85;1.75
EU PL;256;309;16;7;4.31;465.81;4.34;1.25;4.49;1.30;4.36;12.08;4.54;50.97;0.82;4.14;1.27;5.36
US OK;642;3{,}237;12;65;46.68;2{,}759.92;47.10;91.27;47.76;111.22;47.02;137.95;48.41;1{,}833.34;0.90;2.31;0.72;3.70
US TX;304;6{,}762;30;136;106.32;4{,}826.02;106.76;508.13;106.79;478.30;107.22;148.08;107.44;358.38;0.41;0.44;0.84;1.05
EU IT;302;349;21;7;4.01;9{,}600.55;4.07;2.58;4.21;2.72;4.06;18.64;4.28;64.39;1.44;4.91;1.26;6.60
EU FR;828;964;25;20;15.16;15{,}644.50;15.31;39.86;15.52;34.48;15.32;48.56;15.74;341.25;1.02;2.36;1.05;3.82
EU DE;638;866;23;18;12.85;22{,}568.95;12.97;23.55;13.25;18.81;13.00;40.50;13.38;265.31;0.93;3.07;1.16;4.08
\end{filecontents*}
\csvreader[late after line=\\, head to column
        names, separator=semicolon]{data/table_results_column_generation_approx.csv}{}%
        {\network & \nodes & \components & \csvcoliv & \csvcolv & \valueExact & \timeExact }
\bottomrule
\end{tabular}
\end{table}

Table~\ref{table:results_exact_column_generation} shows the exact value of the game $\Gamma$ and the CPU time (in seconds) obtained through CG. Generally, we observe that the running times quickly increase with the size of the network, being the number of detector locations more relevant than the number of components. However, additional factors such as the overlapping structure of the monitoring sets and the magnitude of $\rD$ relative to the number of locations may influence performance, as evidenced by the cases of US TX and US OK. Overall, Table~\ref{table:results_exact_column_generation} shows that CG can solve $\Gamma$ to optimality for real-world networks and a relatively small number of detectors.

Next, we examine the performance of our approximate solution approaches on the same network instances outlined in Table~\ref{table:results_exact_column_generation}. Specifically, we evaluate the performance of the column generation algorithm to solve \eqref{LP: LP1 reformulation} using both forward greedy (CG-FG) and reverse greedy (CG-RG) algorithms to approximate \eqref{pb:DBR}. Similarly, we asses the multiplicative weights update algorithm utilizing both forward greedy (MWU-FG) and reverse greedy (MWU-RG) algorithms for approximating \eqref{pb:DBR}. We set $\varepsilon=0.001 m$ for all methods. The results are shown in Table~\ref{table:results_approx_column_generation_mwu}.


\begin{table}[htbp]
\caption{Computational results for approximate solution methods CG-FG, CG-RG, MWU-FG, and MWU-RG for network instances described in Table~\ref{table:results_exact_column_generation}.}
\label{table:results_approx_column_generation_mwu}
\footnotesize
\centering
\setlength{\tabcolsep}{4pt}
\begin{tabular}{l@{\hspace{3.5em}}*{4}{d{3.3}}@{ \hspace{3.5em}}*{4}{d{3.3}}} 
\toprule
{} & \multicolumn{4}{c@{\hspace{3.5em}}}{\textbf{\% Relative Error}} & \multicolumn{4}{c}{\textbf{Time [s]}} \\
\textbf{Network} & \mc{\textbf{CG-FG}} & \mc{\textbf{CG-RG}} & \mc{\textbf{MWU-FG}} & \multicolumn{1}{c@{\hspace{3.5em}}}{\textbf{MWU-RG}} & \mc{\textbf{CG-FG}} & \mc{\textbf{CG-RG}} & \mc{\textbf{MWU-FG}} & \mc{\textbf{MWU-RG}} \\
\midrule
\csvreader[late after line=\\, head to column
        names, separator=semicolon]{data/table_results_column_generation_approx.csv}{}%
        {\network & \relErrorCGSG & \relErrorCGRG & \relErrorMWUSG & \relErrorMWURG & \timeCGSG & \timeCGRG & \timeMWUSG & \timeMWURG}
\bottomrule
\end{tabular}
\end{table}


Table~\ref{table:results_approx_column_generation_mwu} shows the relative approximation error of the worst-case expected number of undetected attacks with respect to the value of the game $\Gamma$ and the CPU time (in seconds) of the inspection strategies generated by the four approximation methods. We observe a substantial reduction in running times for the five instances that took more than $1{,}000$ seconds to solve under CG---as shown in Table~\ref{table:results_exact_column_generation}---with an average reduction of $94.2\%$ achieved by the four approximation methods. Interestingly, although the forward greedy algorithm provides a theoretically weaker approximate solution compared to the reverse greedy algorithm, both CG-FG and MWU-FG achieve significantly smaller relative errors, typically less than $1.5\%,$ compared to their counterparts, CG-RG and MWU-RG, which yield solutions within $6.6\%$ of optimality. Furthermore, the running times of CG-FG and MWU-FG are also generally shorter than those of CG-RG and MWU-RG, respectively.

We also compare the performance of both the exact and approximate solution methods on a single network for different values of $\rD$ and fixed $\rA$. To this end, we select the gas distribution network of France (EU FR) as our benchmark (see Figure~\ref{fig:FR_natural_gas_layout}).
\begin{figure}[htbp]
    \centering
    \includegraphics[scale=0.2]{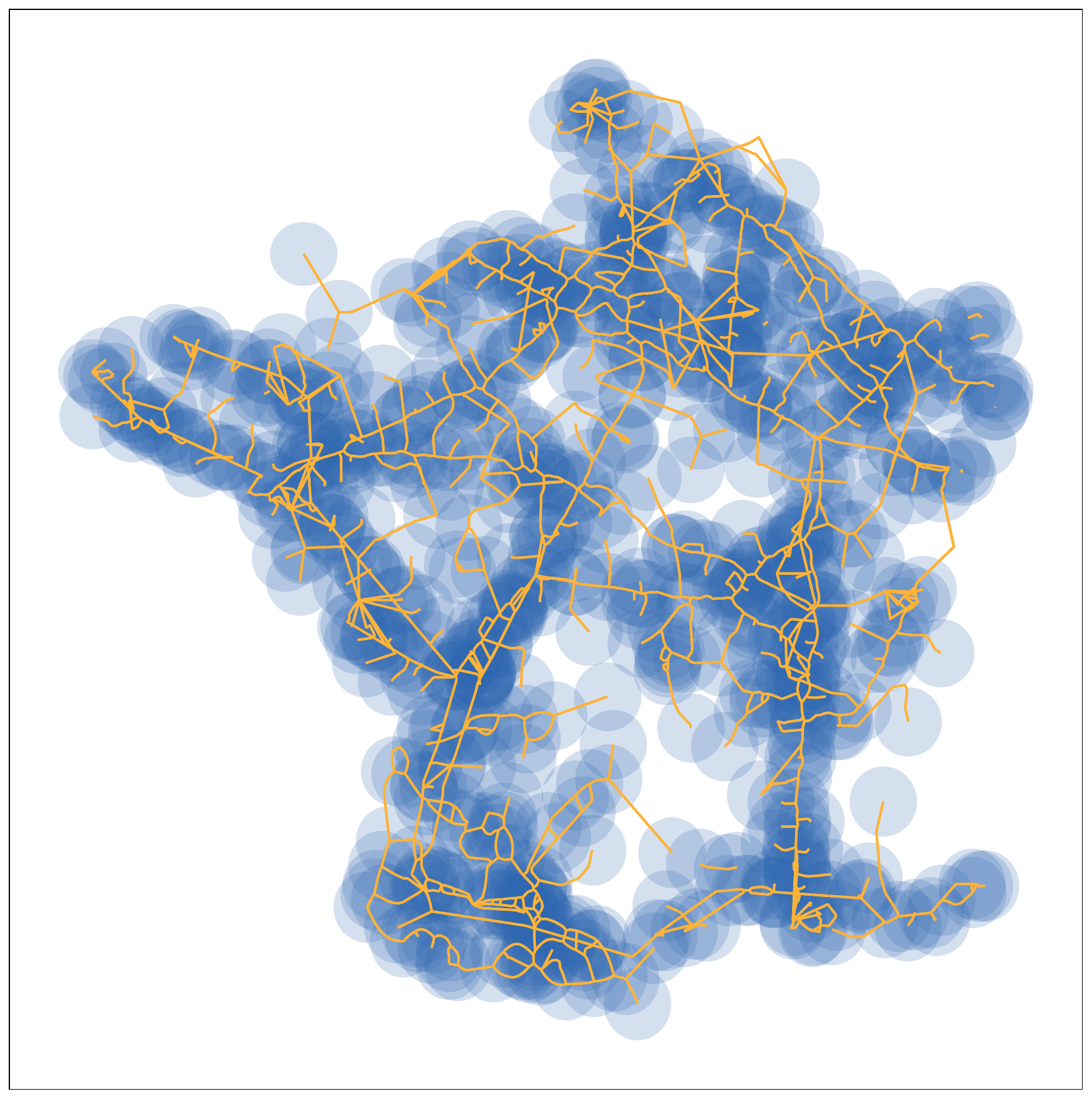}
    \caption{Layout of the benchmark natural gas distribution network EU FR, with circular inspection zones of radius $50$ km around each detector location.}\label{fig:FR_natural_gas_layout}
\end{figure}
We first analyze the results of our experiments on EU FR for a relatively small number of detectors, depicted in Figures~\ref{fig:approximation ratio mwu FR}~and~\ref{fig:running time cg vs mwu FR}.

\begin{figure}[htbp]
    \centering
    \subfloat[Relative error with respect to value of the game ($\rA=20$).]{%
        \input{figures/approx_ratio_colgen_vs_mwu_FR_epsilon.tikz}
        \label{fig:approximation ratio mwu FR}
    }
    \hfill
    \subfloat[Running time (in seconds, $\rA=20$).]{%
        \input{figures/running_time_colgen_vs_mwu_FR_epsilon.tikz}
        \label{fig:running time cg vs mwu FR}%
    }%
    \caption{Computational results for small number of detectors of on benchmark network EU FR.}
    \label{fig:cg vs mwu FR}
\end{figure}
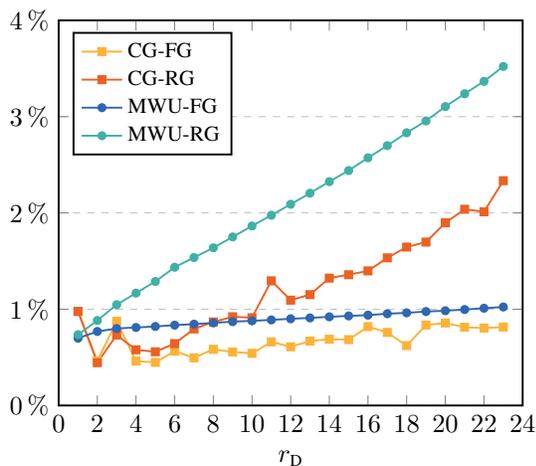
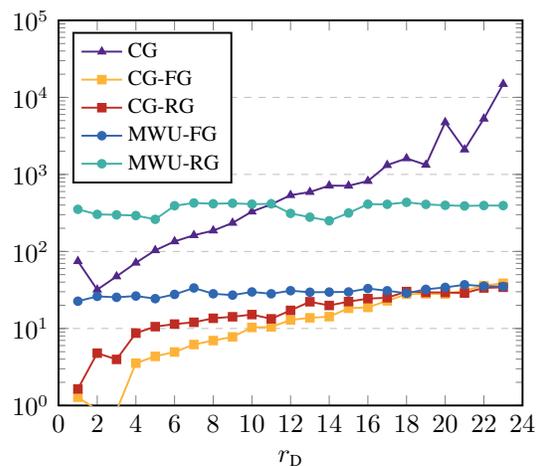

Figure~\ref{fig:approximation ratio mwu FR} shows the relative error of the four approximation methods as a function of the number of detectors. The results illustrate that the error of all methods increase with the number of detectors, reflecting the increasing difficulty to effectively coordinating the detectors by the greedy algorithms for approximate best response. In line with the results from Table~\ref{table:results_approx_column_generation_mwu}, the error of CG-RG and MWU-RG increases at a significantly faster rate with the number of detectors, compared to that of CG-FG and MWU-FG, which consistently remains smaller than $1.1\%$, being CG-FG the method achieving the smallest error for almost every value of $\rD$. 

Figure~\ref{fig:running time cg vs mwu FR} illustrates, using a logarithmic scale, the running times of the exact CG and the four approximation methods. We note that the running time of CG does not scale efficiently with the number of available detectors $\rD$. In fact, CG exhibits running times almost two orders of magnitude larger than those of MWU-FG and MWU-RG, which remain relatively constant. Moreover, MWU-FG outperforms MWU-RG by one order of magnitude. As anticipated in Section~\ref{subsec:approximation_algorithms_for_DBR}, this difference arises due to the small size of $\rD$ relative to $n$, resulting in fewer iterations being required by forward greedy compared to reverse greedy. Nonetheless, both CG-FG and CG-RG exhibit faster, albeit increasing, running times, eventually reaching levels comparable to MWU-FG for the largest values of $\rD$ in this range. This motivates us to further examine the performance of these methods for larger numbers of detectors, as depicted in Figures~\ref{fig:approximation ratio mwu FR large} and~\ref{fig:running time cg vs mwu FR large}.

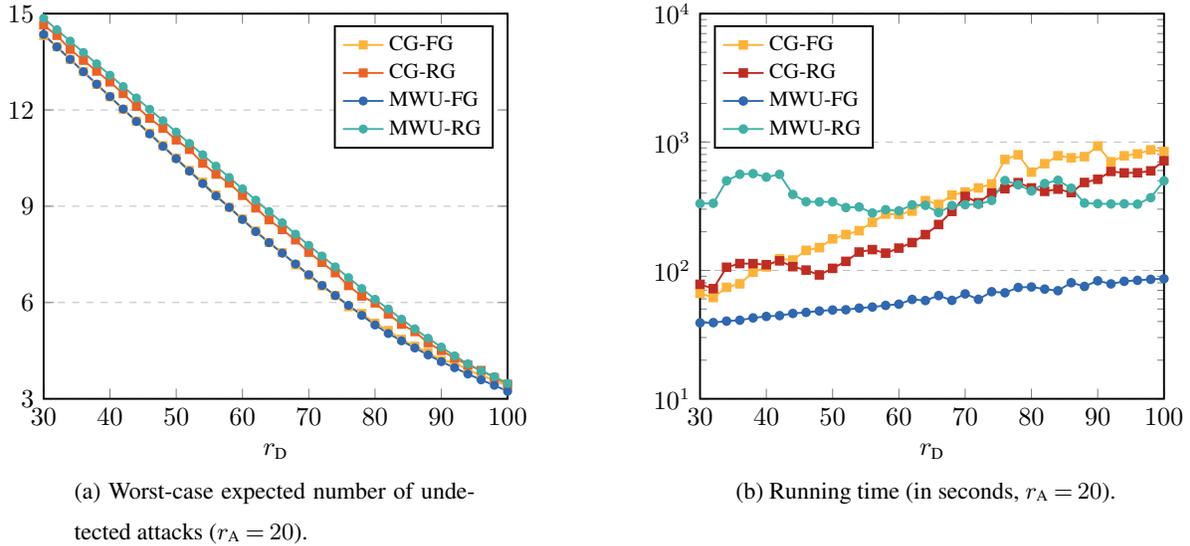
\begin{figure}[htbp]
    \subfloat[Worst-case expected number of undetected attacks ($\rA=20$).]{%
        \input{figures/undetection_colgen_vs_mwu_FR_large_b1.tikz}
        \label{fig:approximation ratio mwu FR large}
    }
    \hfill
    \subfloat[Running time (in seconds, $\rA=20$).]{%
        \input{figures/running_time_colgen_vs_mwu_FR_large_b1.tikz}
        \label{fig:running time cg vs mwu FR large}%
    }%
    \caption{Computational results for large number of detectors on benchmark network EU FR.}
    \label{fig:cg vs mwu FR eps large}
\end{figure}

For larger values of $\rD$, exact solutions of $\Gamma$ become computationally prohibitive, preventing us from plotting the approximation error. Hence, Figure~\ref{fig:approximation ratio mwu FR large} illustrates the worst-case expected number of undetected attacks for the inspection strategies generated by all four approximation methods. Consistent with the trends observed in Figure~\ref{fig:approximation ratio mwu FR}, both CG-FG and MWU-FG outperform CG-RG and MWU-RG. Moreover, in Figure~\ref{fig:running time cg vs mwu FR large}, while the running times of MWU-FG and MWU-RG remain within similar orders of magnitude as those shown in Figure~\ref{fig:running time cg vs mwu FR}, we observe a notable increase in the running times of both CG-FG and CG-RG with the number of detectors. In this scenario, MWU-FG significantly outperforms the other three methods by up to one order of magnitude for the largest values of $\rD$. Overall, MWU-FG emerges as the preferred method for a large number of detectors.


%% file: figures/approx_ratio_colgen_vs_mwu_FR_epsilon.tikz
\begin{tikzpicture}[scale=0.9]
    \definecolor{cb_red}{RGB}{191,044,035}
    \definecolor{cb_blue}{RGB}{047,103,177}
    \definecolor{cb_yellow}{RGB}{253,179,056}
    \definecolor{cb_turquoise}{RGB}{064,176,166}
    \definecolor{cb_orange}{RGB}{235,097,035}
    \definecolor{cb_purple}{RGB}{081,040,136}
    \begin{axis}[align =center,
        xlabel={$\rD$},
        xmin=0, xmax=24,
        ymin=0, ymax=4,
        xtick={0,2,...,24},
        ytick={0, 1, ..., 11},
        yticklabel=\pgfmathprintnumber{\tick}\,$\%$,
        ymajorgrids=true,
        grid style=dashed,
        style={thick},
        mark options={scale=0.75},
        legend pos=north west,
        legend cell align={left},
    ]
    \addplot[
        color=cb_yellow,
        mark=square*,
        ]
        table [x=b1, y=relErrorCGSG, col sep=semicolon] {data/results_FR_b2_20_epsilon_0005b2.csv};
    \addplot[
        color=cb_orange,
        mark=square*,
        ]
        table [x=b1, y=relErrorCGRG, col sep=semicolon] {data/results_FR_b2_20_epsilon_0005b2.csv};
    \addplot[
        color=cb_blue,
        mark=*,
        ]
        table [x=b1, y=relErrorMWUSG, col sep=semicolon] {data/results_FR_b2_20_epsilon_0005b2.csv};
    \addplot[
        color=cb_turquoise,
        mark=*,
        ]
        table [x=b1, y=relErrorMWURG, col sep=semicolon] {data/results_FR_b2_20_epsilon_0005b2.csv};
    \addlegendentry{\footnotesize CG-FG}
    \addlegendentry{\footnotesize CG-RG}
    \addlegendentry{\footnotesize MWU-FG}
    \addlegendentry{\footnotesize MWU-RG}
    \end{axis}
\end{tikzpicture}

%% file: figures/running_time_colgen_vs_mwu_FR_epsilon.tikz
\begin{tikzpicture}[scale=0.9]
    \definecolor{cb_red}{RGB}{191,044,035}
    \definecolor{cb_blue}{RGB}{047,103,177}
    \definecolor{cb_yellow}{RGB}{253,179,056}
    \definecolor{cb_turquoise}{RGB}{064,176,166}
    \definecolor{cb_orange}{RGB}{235,097,035}
    \definecolor{cb_purple}{RGB}{081,040,136}
    \begin{axis}[align =center,
        xlabel={$\rD$},
        xmin=0, xmax=24,
        ymin=1, ymax=100000,
        xtick={0,2,...,24},
        legend pos=north west,
        legend cell align={left},
        ymajorgrids=true,
        grid style=dashed,
        ymode=log,
        style={thick},
        mark options={scale=0.8},
    ]
    \addplot[
        color=cb_purple,
        mark=triangle*,
        ]
        table [x=b1, y=timeCG, col sep=semicolon] {data/results_FR_b2_20_epsilon_0005b2.csv};
    \addplot[
        color=cb_yellow,
        mark=square*,
        ]
        table [x=b1, y=timeCGSG, col sep=semicolon] {data/results_FR_b2_20_epsilon_0005b2.csv};
    \addplot[
        color=cb_red,
        mark=square*,
        ]
        table [x=b1, y=timeCGRG, col sep=semicolon] {data/results_FR_b2_20_epsilon_0005b2.csv};
    \addplot[
        color=cb_blue,
        mark=*,
        ]
        table [x=b1, y=timeMWUSG, col sep=semicolon] {data/results_FR_b2_20_epsilon_0005b2.csv};
    \addplot[
        color=cb_turquoise,
        mark=*,
        ]
        table [x=b1, y=timeMWURG, col sep=semicolon] {data/results_FR_b2_20_epsilon_0005b2.csv};
    \addlegendentry{\footnotesize CG}
    \addlegendentry{\footnotesize CG-FG}
    \addlegendentry{\footnotesize CG-RG}
    \addlegendentry{\footnotesize MWU-FG}
    \addlegendentry{\footnotesize MWU-RG}
    \end{axis}
\end{tikzpicture}

%% file: figures/undetection_colgen_vs_mwu_FR_large_b1.tikz
\begin{tikzpicture}[scale=0.9]
    \definecolor{cb_red}{RGB}{191,044,035}
    \definecolor{cb_blue}{RGB}{047,103,177}
    \definecolor{cb_yellow}{RGB}{253,179,056}
    \definecolor{cb_turquoise}{RGB}{064,176,166}
    \definecolor{cb_orange}{RGB}{235,097,035}
    \definecolor{cb_purple}{RGB}{081,040,136}
    \begin{axis}[align =center,
        xlabel={$\rD$},
        xmin=30, xmax=100,
        ymin=3, ymax=15,
        xtick={30,40,...,100},
        ytick={3, 6, ..., 15},
        yticklabel=\pgfmathprintnumber{\tick},
        ymajorgrids=true,
        grid style=dashed,
        style={thick},
        mark options={scale=0.75},
        legend pos=north east,
        legend cell align={left},
    ]
    \addplot[
        color=cb_yellow,
        mark=square*,
        ]
        table [x=b1, y=valueCGSG, col sep=semicolon] {data/results_FR_large.csv};
    \addplot[
        color=cb_orange,
        mark=square*,
        ]
        table [x=b1, y=valueCGRG, col sep=semicolon] {data/results_FR_large.csv};
    \addplot[
        color=cb_blue,
        mark=*,
        ]
        table [x=b1, y=valueMWUSG, col sep=semicolon] {data/results_FR_large.csv};
    \addplot[
        color=cb_turquoise,
        mark=*,
        ]
        table [x=b1, y=valueMWURG, col sep=semicolon] {data/results_FR_large.csv};
    \addlegendentry{\footnotesize CG-FG}
    \addlegendentry{\footnotesize CG-RG}
    \addlegendentry{\footnotesize MWU-FG}
    \addlegendentry{\footnotesize MWU-RG}
    \end{axis}
\end{tikzpicture}

%% file: figures/running_time_colgen_vs_mwu_FR_large_b1.tikz
\begin{tikzpicture}[scale=0.9]
    \definecolor{cb_red}{RGB}{191,044,035}
    \definecolor{cb_blue}{RGB}{047,103,177}
    \definecolor{cb_yellow}{RGB}{253,179,056}
    \definecolor{cb_turquoise}{RGB}{064,176,166}
    \definecolor{cb_orange}{RGB}{235,097,035}
    \definecolor{cb_purple}{RGB}{081,040,136}
    \begin{axis}[align =center,
        xlabel={$\rD$},
        xmin=30, xmax=100,
        ymin=10, ymax=10000,
        xtick={30,40,...,100},
        legend pos=north west,
        legend cell align={left},
        ymajorgrids=true,
        grid style=dashed,
        ymode=log,
        style={thick},
        mark options={scale=0.8},
    ]
    \addplot[
        color=cb_yellow,
        mark=square*,
        ]
        table [x=b1, y=timeCGSG, col sep=semicolon] {data/results_FR_large.csv};
    \addplot[
        color=cb_red,
        mark=square*,
        ]
        table [x=b1, y=timeCGRG, col sep=semicolon] {data/results_FR_large.csv};
    \addplot[
        color=cb_blue,
        mark=*,
        ]
        table [x=b1, y=timeMWUSG, col sep=semicolon] {data/results_FR_large.csv};
    \addplot[
        color=cb_turquoise,
        mark=*,
        ]
        table [x=b1, y=timeMWURG, col sep=semicolon] {data/results_FR_large.csv};
    \addlegendentry{\footnotesize CG-FG}
    \addlegendentry{\footnotesize CG-RG}
    \addlegendentry{\footnotesize MWU-FG}
    \addlegendentry{\footnotesize MWU-RG}
    \end{axis}
\end{tikzpicture}

%% file: 6_conclusions.tex
\section{Conclusion}
\label{sec:conclusion}

Motivated by real-world applications in safeguarding critical infrastructures, we investigated a large-scale zero-sum inspection game where a defender strategically allocates multiple detectors in locations of a system with imperfect and heterogeneous detection capabilities, while an attacker targets multiple system components. The defender (resp. attacker) aims to minimize (resp. maximize) the expected number of undetected attacks. 

We leveraged the inherent structure of the game, which is supermodular for the defender and additive for the attacker, to develop both exact and approximate solution methods utilizing column generation (CG) and multiplicative weights update (MWU) algorithms. These approaches employed efficient algorithms as subroutines to compute either exact or approximate defender's best responses. Specifically, we achieved exact equilibria using CG with an exact MIP formulation for the defender's best response problem. For approximate equilibria, we utilized both CG and MWU in conjunction with greedy algorithms for approximate defender's best responses. Across all methods, we represented the attacker's strategies in terms of unidimensional marginal probabilities. In the case of the MWU algorithm, this required addressing an additional relative entropy projection problem within each iteration, for which we derived a closed-form solution that can be computed in linear time.


Our computational study on real-world gas distribution networks revealed that our exact solution method can efficiently compute equilibrium inspection strategies when the defender has access to a relatively small number of detectors. However, for a larger number of detectors, our approximation methods provide the only scalable solution approaches. Among these, the MWU algorithm exhibits better performance than a CG algorithm when both utilize the forward greedy algorithm for computing approximate defender's best responses. 

Our solution approaches can be applied to other combinatorial zero-sum games featuring a submodular or supermodular structure for one player and an additive structure for the adversary. Examples of such settings include the network inspection model with heterogeneous component criticalities proposed by \citet{milosevic2023strategic}, or the flow interdiction model introduced by \citet{guo2016optimal}.

We propose the following directions for future work. First, exploring approximate equilibrium strategies for structured instances, such as monitoring sets forming minimum set covers or maximum set packings, could enhance our comprehension of strategic interactions in scenarios with complex system structures and potentially lead to the development of more efficient algorithms. Second, we advocate investigating nonzero-sum variants of our game, for which the minmax formulation underlying our CG algorithm, and the MWU algorithm do not necessarily converge to Nash equilibria. Hence, alternative techniques capable of providing satisfactory approximations in these scenarios need to be developed. Finally, another interesting research direction involves analyzing a variant of the game under incomplete information, where players face uncertainty regarding their adversary's number of resources or the detection probabilities, which can be learned through repeated interactions.


%% file: 7_appendix.tex

%
%


\clearpage

\ECHead{Proofs of Statements}

\section*{Proof of Lemma \ref{lem:undetection_function}}
    Let $(\sigmaD,\sigmaA)\in \DeltaD \times \DeltaA$. Then,
    \begin{align*}
        U(\sigmaD,\sigmaA) = \sum_{S\in \AD} \sum_{T\in \AA} \sigmaD_S \,\sigmaA_T \, u(S,T) &= \sum_{S\in \AD} \sum_{T\in \AA} \sigmaD_S \,\sigmaA_T \sum_{e\in T} u(S,e) \\
        &= \sum_{S\in \AD} \sigmaD_S \sum_{e\in \E}  \sum_{T\in \AA:\, e\in T} \sigmaA_T \, u(S,e) \\
        &= \sum_{S\in \AD} \sigmaD_S \sum_{e\in \E}  \rho_{e}(\sigmaA) \,u(S,e) \\
        &= \sum_{e\in \E} \rho_{e}(\sigmaA)\, U(\sigmaD,e).
    \end{align*}
   \hfill\Halmos

\section*{Proof of Proposition \ref{prop:LP1 reformulation}}

    Let $z^*\in \R$ be the optimal value of the variable $z$ in \eqref{LP: LP1}. Then,  
    \begin{align*}
        z^* = \underset{\sigmaD\in \DeltaD}{\min} \, \underset{T\in \AA}{\max} \, U(\sigmaD, T) = \underset{\sigmaD\in \DeltaD}{\min} \, \underset{\sigmaA\in \DeltaA}{\max} \, U(\sigmaD, \sigmaA) = \underset{\sigmaD\in \DeltaD}{\min} \, \underset{\sigmaA\in \DeltaA}{\max} \, \sum_{e\in \E} \rho_{e}(\sigmaA)\, U(\sigmaD,e).
    \end{align*}
    The first equality follows by the constraints of \eqref{LP: LP1}.
    The second equality holds since for every fixed $\sigmaD \in \DeltaD$, the function $U(\sigmaD,\sigmaA)=\sum_{T\in \AA}\sigmaA_T\, U(\sigmaD, T) $ is linear in $\sigmaA$, and therefore it attains its maximum over the set of extreme points of $\DeltaA$, given by the characteristic vectors of the sets in $\AA$. The third equality follows from Lemma \ref{lem:undetection_function}. Thus, \eqref{LP: LP1} is equivalent to
    \begin{align}
        \label{minmax reformulation}
        \underset{\sigmaD\in \DeltaD}{\min} \, \underset{\sigmaA\in \DeltaA}{\max} \, \sum_{e\in \E} \rho_{e}(\sigmaA)\, U(\sigmaD,e).
    \end{align}
    Next, for fixed $\sigmaD \in \DeltaD$, the inner maximization in \eqref{minmax reformulation} is equivalent to the following LP:
    \begin{align}
        \begin{alignedat}{3} \label{LP: inner LP reformulation}
        \underset{\rhoA \in \ThetaA}{\max} \sum_{e\in \E} \rhoA_e \, U(\sigmaD,e) =
        & \max_{\rhoA \in \R^{\E}} \quad && \sum_{e\in \E} \rhoA_{e} \, \sum_{S\in \AD} \sigmaD_S \, u(S,e) \\
        & \text{subject to} \quad && \sum_{e\in \E} \rhoA_{e} \leq \rA   &  \\
        &  && 0 \leq \rhoA_{e} \leq 1 & \forall\, e \in \E.
        \end{alignedat}
    \end{align}
    Indeed, by Lemma \ref{lem:convex_hull}, for any feasible solution of the inner maximization of \eqref{minmax reformulation}, there exists a feasible solution of \eqref{LP: inner LP reformulation} with the same expected number of undetected attacks and vice versa. Now, we consider the dual of \eqref{LP: inner LP reformulation}:
    \begin{align}
        \label{LP: dual inner LP reformulation}
        \begin{alignedat}{4}
            & \underset{\lambda \in \R^{\E},\, \gamma \in \R}{\min} & & \rA\gamma + \sum_{e\in \E}\lambda_{e} \\
            & \textnormal{subject to}& \quad & \begin{aligned}[t]
                                                \gamma + \lambda_e & \geq \sum_{S\in \AD} \sigmaD_S \, u(S,e)  & \forall\, & e\in \E,\\
                                                \lambda_e & \geq 0 & \forall\, & e\in \E,\\
                                                \gamma & \geq 0. &  &
                                            \end{aligned}
        \end{alignedat}
    \end{align}
    By strong duality, the optimal values of \eqref{LP: inner LP reformulation} and \eqref{LP: dual inner LP reformulation} coincide. Thus, the reformulation \eqref{LP: LP1 reformulation} is obtained by substituting the inner maximization in \eqref{minmax reformulation} by \eqref{LP: dual inner LP reformulation}. In particular, given an optimal solution $(\sigmaDEq, \lambda^{*}, \gamma^{*}) \in \R^{\AD} \times \R^{\E} \times \R$ of \eqref{LP: LP1 reformulation}, the value of the game is given by the optimal value $\rA \gamma^{*} + \sum_{e\in \E} \lambda^{*}_e$.

    Next, let $(\rhoAEq, \nu^*)\in \R^{\E} \times \R$ be an optimal solution of the dual of \eqref{LP: LP1 reformulation}, which is given by 
    \begin{align}
        \begin{alignedat}{3} \label{LP: dual of LP reformulation}
        & \max_{\rhoA \in \R^{\E},\, \nu \in \R} & & \nu \\
        & \text{subject to}& \quad & \begin{aligned}[t]
                                        & \nu  \leq \sum_{e\in \E} \rhoA_e \, u(S, e) \quad & \forall\, S\in \AD,\\
                                        & \sum_{e\in \E} \rhoA_{e}  \leq \rA & &\\
                                        & 0 \leq \rhoA_{e} \leq 1 \quad & \forall\, e \in \E.
                                    \end{aligned}
        \end{alignedat}
    \end{align}
    From the equivalence between \eqref{LP: LP1} and \eqref{LP: LP1 reformulation}, it follows that $\sigmaDEq$ in an optimal solution of $\min_{\sigmaD \in \DeltaD} \max_{\rhoA \in \ThetaA} U(\sigmaD, \rhoA)$. On the other hand, we can similarly show that the dual of \eqref{LP: LP1} is equivalent to \eqref{LP: dual of LP reformulation}, which in turn is equivalent to $\max_{\rhoA \in \ThetaA} \min_{\sigmaD \in \DeltaD} U(\sigmaD, \rhoA)$, for which $\rhoAEq$ is an optimal solution. Therefore, $(\sigmaDEq, \rhoAEq)$ is an equilibrium of the game $\Gamma$.
   \hfill\Halmos

\section*{Proof of Proposition~\ref{prop:NP_hardness_Defender_BR_2}}
    We show that Problem \eqref{pb:DBR} is NP-hard even for instances where $\rA=1$, every component belongs to at most two monitoring sets, and the detection probabilities are homogeneous. To this aim, we reduce the NP-hard problem VERTEX COVER to \eqref{pb:DBR}. Given a graph $G=(V,E)$ and a parameter $k\in \Z_{>0}$, the VERTEX COVER problem consists in finding a subset of nodes $S\subseteq V$ of size at most $k$, such that every edge of $G$ is incident to at least one node of $S$.

    Given an instance of VERTEX COVER, we construct an instance of \eqref{pb:DBR} by setting $\V = V$ and $\E = E$. Then, we define the monitoring sets and detection probabilities respectively as $\C_v = \{ e\in E:\, v \in e \}$ and $p_v=1-1/(m+1)$ for all $v \in \V$ (we recall that $m\coloneqq |\E|$). We note that every component belongs to exactly two monitoring sets, and the detection probabilities are homogeneous. Next, we set $\rA=1$ and let $\rhoA$ be the uniform distribution in $\E$, that is $\rhoA_e=1/m$ for every $e\in \E$. Finally, we let $\rD=k$, so $\AD = \{ S \subseteq \V:\, |S| \leq k \}$. From the selection of the detection probabilities, for every detector positioning $S\in \AD$ and component $e\in \E$, we have $u(S,e)\in \{1,1/(m+1),1/(m+1)^2\}$. Specifically, $u(S,e)=1$ if $e \notin \cup_{v \in S}\, \C_v$, $u(S,e)=1/(m+1)$ if there is exactly one $v \in S$ such that  $e\in\C_v$, and $u(S,e)=1/(m+1)^2$ if there are exactly two $v,w\in S$ such that $e \in \C_v \cap \C_w$. Next, we argue that $G$ has a vertex cover of size at most $k$ if and only if there exists a detector positioning $S\in \AD$ such that $U(S,\rhoA)\leq 1/(m+1)$.

    Suppose $G$ has a vertex cover $S$ of size at most $k$. Then, $S\in \AD$, and since all edges are covered by $S$, we must have $u(S,e) \leq 1/(m+1)$ for every $e\in \E$. It follows that $U(S,\rhoA)= (1/m)\sum_{e\in \E}u(S,e) \leq 1/(m+1)$. Conversely, suppose there is a detector positioning $S\in \AD$ such that $U(S,\rhoA) \leq 1/(m+1)$. Then, $|S|\leq k$ and $S$ is a vertex cover in $G$. Indeed, suppose there is an edge $e^{\prime}\in E$ which is not incident to any node in $S$. Then, we must have $u(S,e^{\prime})=1$. If $m=1$, then $e^{\prime}$ is the only edge of $E$ and $S=\emptyset$, so $U(S,\rhoA) = u(S,e^{\prime})=1$, which is a contradiction. If $m>1$, then it follows that 
    \begin{align*}
        \frac{1}{m+1} \geq U(S,\rhoA) = \frac{1}{m} \sum_{e\in \E} u(S,e) = \frac{1}{m}u(S,e^{\prime}) +  \frac{1}{m}\sum_{\substack{e\in \E: \\ e\neq e^{\prime}}} u(S,e) \geq \frac{1}{m} +  \frac{1}{m}\cdot \frac{m-1}{(m+1)^2} > \frac{1}{m}, 
    \end{align*}
    which is again a contradiction.
   \hfill\Halmos




\section*{Proof of Proposition~\ref{prop:MIP_optimality}}
    For a fixed enumeration of $\V = \{v_1,\ldots,v_n\}$, let $S\in \AD$ be a detector positioning and $(x,u)\in \R^{n} \times \R^{\E \times \{1,\ldots,n\} }$ be defined as
    \begin{align}
        x_i &= \mathds{1}_{S}(v_i) \quad & \forall\, v_i \in \V, \label{MIP:feasible_solution_1} \\
        u_{e,i} &= \prod_{v_j \in \V:\, j \leq i} \left(1 - p_{v_j} \mathds{1}_{ \C_{v_j} }(e) x_{j} \right) \quad &\forall \, e\in \E, \forall\, v_i\in \V,   \label{MIP:feasible_solution_2}
    \end{align}
    where $\mathds{1}_{S}(v_i) = 1$ (resp. $\mathds{1}_{ \C_{v_j} }(e)=1$) if $v_i\in S$ (resp. $e\in \C_{v_j}$) and $\mathds{1}_{S}(v_i)= 0$ (resp. $\mathds{1}_{ \C_{v_j} }(e)=0$) otherwise. We first show that $(x,u)$ is a feasible solution of \eqref{MIP:pricing problem}. To this end, let us consider its constraints:
    \begin{align}
            \sum_{v_i\in \V} x_{i} &\leq \rD, \label{MIP:linearizing_ineq_1}\\
            1-p_{v_1} \mathds{1}_{ \C_{v_1} }(e) x_{1} &\leq u_{e,1} \quad &\forall\, e\in \E, \label{MIP:linearizing_ineq_2}\\
            u_{e,i} \left(1-p_{v_{i+1}} \mathds{1}_{ \C_{v_{i+1}} }(e) \right) &\leq u_{e,i+1}  \quad &\forall\, e\in \E, \, \forall\, v_i\in \V \setminus \left\{v_n \right\}, \label{MIP:linearizing_ineq_3} \\
            u_{e,i} - x_{i+1} &\leq u_{e,i+1}  \quad  &\forall\, e\in \E, \, \forall\, v_i\in \V \setminus \left\{v_n\right\}. \label{MIP:linearizing_ineq_4} \\
            0 \leq u_{e,i} &\leq 1 \quad &  \forall \,  v_i  \in \V, \label{MIP:linearizing_ineq_5}\\
            x_{i} &\in \{0,1\} & \forall\, v_i \in \V. \label{MIP:linearizing_ineq_6}
    \end{align}
    By definition, $(x,u)$ clearly satisfies \eqref{MIP:linearizing_ineq_1}, \eqref{MIP:linearizing_ineq_2}, \eqref{MIP:linearizing_ineq_5} and \eqref{MIP:linearizing_ineq_6}. Then, it remains to show that $(x,u)$ also satisfies \eqref{MIP:linearizing_ineq_3} and \eqref{MIP:linearizing_ineq_4}. From \eqref{MIP:feasible_solution_2}, it holds that $(x,u)$ satisfies the following (nonlinear) equalities:    
    \begin{align}
        u_{e,i} \left( 1-p_{v_{i+1}} \mathds{1}_{ \C_{v_{i+1}} }(e) x_{i+1}\right) &= u_{e,i+1}  \quad &\forall\, e\in \E, \, \forall\, v_i \in \V \setminus \left\{v_n\right\}. \label{MIP:linearizing_ineq_equivalent}
    \end{align}    
    We note that the left hand side of \eqref{MIP:linearizing_ineq_equivalent} is equivalent to $\max\left\{ u_{e,i} \left( 1-p_{v_{i+1}} \mathds{1}_{ \C_{v_{i+1}} }(e) \right),\, u_{e,i} - x_{i+1} \right\}$. Indeed, if $x_{i+1}=0$, then from the inequalities $u_{e,i} \geq 0$ and $p_{v_{i+1}} \mathds{1}_{ \C_{v_{i+1}} }(e) \leq 1$ it follows that $ u_{e,i} \left( 1-p_{v_{i+1}} \mathds{1}_{ \C_{v_{i+1}} }(e) \right) \leq u_{e,i} = u_{e,i} - x_{i+1}$.  On the other hand, if $x_{i+1}=1$, then the inequalities $u_{e,i} \leq 1$ and $p_{v_{i+1}} \mathds{1}_{ \C_{v_{i+1}} }(e) \leq 1$ imply that $u_{e,i} - x_{i+1} = u_{e,i} - 1 \leq u_{e,i} \left( 1-p_{v_{i+1}} \mathds{1}_{ \C_{v_{i+1}} }(e) \right)$. Thus, \eqref{MIP:linearizing_ineq_equivalent} translates to
    \begin{align*}
        \max\left\{ u_{e,i} \left( 1-p_{v_{i+1}} \mathds{1}_{ \C_{v_{i+1}} }(e) \right),\, u_{e,i} - x_{i+1} \right\} &= u_{e,i+1}  \quad &\forall\, e\in \E, \, \forall\, v_i \in \V \setminus \left\{v_n\right\},
    \end{align*}
    from which we deduce that $(x,u)$ satisfies \eqref{MIP:linearizing_ineq_3} and \eqref{MIP:linearizing_ineq_4}. Therefore, $(x,u)$ is feasible for \eqref{MIP:pricing problem}. We also note that $u_{e,n}=u(S,e)$ for every $e\in \E$.

    Next, let $(x^*,u^*)\in \R^{n}\times \R^{\E \times \{1,\ldots,n\} }$ be an optimal solution of \eqref{MIP:pricing problem}, and let $S^{*} \coloneqq \{ v_i \in \V:\, x^*_i=1\}$ be its associated detector positioning. By similar arguments as above, we can show that \eqref{MIP:linearizing_ineq_2}-\eqref{MIP:linearizing_ineq_4} imply that $(x^*,u^*)$ satisfies
    \begin{align}
         1-p_{v_1} \mathds{1}_{ \C_{v_1} }(e) x^*_{1} &\leq u^*_{e,1} \quad &\forall\, e\in \E,  \label{MIP:proof_inequality_1} \\ 
        u^*_{e,i} \left( 1-p_{v_{i+1}} \mathds{1}_{ \C_{v_{i+1}} }(e) x^*_{i+1}\right) &\leq u^*_{e,i+1}  \quad &\forall\, e\in \E, \, \forall\, v_i \in \V \setminus \left\{v_n\right\}. \label{MIP:proof_inequality_2}
    \end{align}
    By repeatedly applying these inequalities, it follows that $(x^*, u^*)$ satisfies 
    \begin{align*}
        u^*_{e,n} &\geq u^*_{e, n-1} \left(1 - p_{v_n} \mathds{1}_{ \C_{v_n} }(e) x^*_n \right) \\
        &\geq u^*_{e, n-2} \left(1 - p_{v_{n-1}} \mathds{1}_{ \C_{v_{n-1}} }(e) x^*_{n-1} \right) \left(1 - p_{v_n} \mathds{1}_{ \C_{v_n} }(e) x^*_n \right) \\
        &\geq \cdots \\
        &\geq \prod_{v_i \in \V} \left(1 - p_{v_i} \mathds{1}_{ \C_{v_i} }(e) x^*_{i} \right) = u(S^*, e).
    \end{align*}
    
    Furthermore, since \eqref{MIP:pricing problem} is a minimization problem, for every $e\in \E$ with $\rhoA_{e}>0$, $u^{*}_{e,n}$ attains its lower bound $u(S^{*}, e)$ by tightening the inequalities \eqref{MIP:proof_inequality_1} and \eqref{MIP:proof_inequality_2}. Therefore,
    \begin{align*}
        U(S^{*}, \rhoA) = \sum_{e\in \E} \rhoA_{e} \, u(S^{*}, e) = \sum_{e\in \E} \rhoA_{e} \, u^{*}_{e,n} \leq \sum_{e\in \E} \rhoA \, u_{e,n} = \sum_{e\in \E} \rhoA \, u(S, e) =  U(S, \rhoA).
    \end{align*}
    Since $S\in \AD$ is arbitrary, we conclude that $S^*$ is an optimal solution for \eqref{pb:DBR}. 
   \hfill\Halmos

\section*{Proof of Lemma~\ref{prop:supermodularity}}
    We first show that for every fixed $e\in \E$, the function $u(S,e) = \prod_{v \in S:\, e\in \C_v} (1-p_v)$, defined for every $S\subseteq \V$, is nonincreasing and supermodular.
    
    \begin{itemize}
        \item[--] Let $S,S^{\prime}\subseteq\V$ be such that $S\subseteq S'$. Then,
        \begin{align*}
            u(S,e) = \prod_{v \in S:\, e\in \C_v} (1-p_v) \geq \prod_{v \in S^{\prime}:\, e\in \C_v} (1-p_v) = u(S^{\prime},e).
        \end{align*}
        Thus, $u(\cdot,e)$ is nonincreasing.
        \item[--] Let $S,S^{\prime}\subseteq \V$ be such that $S\subseteq S'$, and let $v \notin S^{\prime}$. Then,
        \begin{align*}
            u(S,e) - u(S \cup \{v\}, e) &= \prod_{w \in S:\, e\in \C_w} (1-p_w) - \prod_{w \in S \cup \{v\}:\, e\in \C_w} \hspace{-1.5em} (1-p_w) \\
            &= p_v \mathds{1}_{ \C_{v} }(e) \prod_{w \in S:\, e\in \C_w} (1-p_w) \\
            &\geq p_v \mathds{1}_{ \C_{v} }(e) \prod_{w \in S^{\prime}:\, e\in \C_w} (1-p_w) \\
            &= u(S^{\prime},e) - u(S^{\prime} \cup \{ v \}, e).
        \end{align*}
        Therefore, $u(\cdot,e)$ is supermodular.
    \end{itemize}
    Finally, we recall that for every fixed $\rhoA \in \ThetaA$, $U(S, \rhoA) = \sum_{e\in \E} \rhoA_e\, u(S, e)$. In particular, $U(\cdot, \rhoA)$ is a nonnegative linear combination of nonincreasing and supermodular set functions. This implies that $U(\cdot,\rhoA)$ is nonincreasing and supermodular as well.
   \hfill\Halmos

\section*{Proof of Lemma~\ref{lem:upper_bound_curvature}}
    For every $\rhoA \in \ThetaA$, we have
    \begin{align*}
        U_{\V \setminus \{ v \}}(v, \rhoA) &= U(\V \setminus \{ v \}, \rhoA) - U(\V, \rhoA) \\
        &= \sum_{e\in \E} \rhoA_e \, u( \V \setminus \{ v \}, e ) - \sum_{e\in \E} \rhoA_e \, u( \V \setminus \{ v \}, e ) \left( 1 - p_v \mathds{1}_{ \C_{v} }(e) \right) \\
        &= p_v \sum_{e\in \C_v} \rhoA_e \, u( \V \setminus \{ v \}, e ) \\
        &= p_v \sum_{e\in \C_v} \rhoA_e \hspace{-0.5em} \prod_{ \substack{ w \in \V \setminus \{ v \}: \\ e\in \C_w } } \hspace{-0.5em} (1-p_w) \geq \left( 1 - \max_{w \in  \V} p_w \right)^{d} p_v \sum_{e\in \C_v} \rhoA_e,
    \end{align*}
    where the inequality is due to the assumption that every component belongs to at most $d$ monitoring sets. On the other hand, 
    \begin{align*}
        U_{\emptyset}(v, \rhoA) = U(\emptyset, \rhoA) - U(v, \rhoA) = \sum_{e\in \E} \rhoA_e - \sum_{e\in \E} \rhoA_e \left( 1 - p_v \mathds{1}_{ \C_{v} }(e) \right) = p_v \sum_{e\in \C_v} \rhoA_e.
    \end{align*}
    Let us assume that $\rhoA \neq \boldsymbol{0}$. Then, since $p_v >0$ for every $v\in \V$, the set $\left\{v \in \V:\,  U_{\emptyset}(v, \rhoA) > 0 \right\}$ is nonempty. It follows that
    \begin{align*}
        c = 1 - \underset{v \in \V:\,  U_{\emptyset}(v, \rhoA) > 0}{\min}\, \frac{ U_{\V \setminus \{ v \}}(v, \rhoA)  }{ U_{\emptyset}(v, \rhoA) } 
         \leq 1 - \left( 1 - \max_{w \in  \V} p_w \right)^{d}.
    \end{align*}
    Finally, if $\rhoA = \boldsymbol{0}$, $U(\cdot, \rhoA)$ is identical to zero. Therefore, $c=0$ by definition, ensuring that the bound also holds in this case. 
   \hfill\Halmos

\section*{Proof of Proposition~\ref{prop:inapproximability}}
    We first show that if $c=1$, deciding whether the optimal value of \eqref{pb:DBR} is zero is NP-hard. 
    Let us consider the reduction from VERTEX COVER in the proof of Proposition~\ref{prop:NP_hardness_Defender_BR_2}, but set $p_v=1$ for all $v \in \V$. In this case, every $v \in \V$ satisfies 
    \begin{align*}
        U_{\V \setminus \{ v \}}(v, \rhoA) &= U(\V \setminus \{v\},\rhoA ) - U(\V,\rhoA) = 0, \\
        U_{\emptyset}(v, \rhoA) &= U(\emptyset,\rhoA) - U(\{ v \}, \rhoA) = \frac{|\C_v|}{m+1}.
    \end{align*}
    Therefore, $c=1$. Clearly, $G$ has a vertex cover of size at most $k$ if and only if there exists a detector positioning $S\in \AD$ such that $U(S,\rhoA)=0$, which concludes the reduction.
    
    Next, let $S^*\in \AD$ be a defender's best response of an instance of \eqref{pb:DBR} with $c=1$. Given any $\alpha\geq 1$, suppose that there exists a polynomial-time algorithm that computes a detector positioning $\widehat{S}$ satisfying
    \begin{align*}
        U(S^{*}, \rhoA) \leq  U(\widehat{S},\rhoA) \leq \alpha \, U(S^{*}, \rhoA).
    \end{align*}
    These inequalities imply that $U(S^{*}, \rhoA)=0$ if an only if $U(\widehat{S},\rhoA)=0$. In other words, we could use such algorithm to decide whether the optimal value of \eqref{pb:DBR} is zero in polynomial time, which, from the above reduction, is not possible unless P$=$NP.
   \hfill\Halmos  



\section*{Proof of Proposition~\ref{prop:forward_greedy_guarantee}}
    Let $\rhoA \in \ThetaA$ be fixed. First, we observe that $S^*\in \AD$ minimizes the expected number of undetected attacks $U(S, \rhoA)$ over $\AD$ if and only if $S^*$ maximizes the expected number of \emph{detected} attacks $D(S, \rhoA) \coloneqq \sum_{e\in \E} \rhoA_e - U(S, \rhoA)$ over $\AD$. 
    
    Since $U(\cdot, \rhoA)$ is upper bounded by $\sum_{e\in \E} \rhoA_e$, as well as nonincreasing and supermodular (Lemma~\ref{prop:supermodularity}), it follows that $D(\cdot, \rhoA)$ is nonnegative, nondecreasing and submodular. The curvature parameter for the latter function is defined as
    \begin{align*}
        c \coloneqq 1 - \underset{v \in \V:\,  D_{\emptyset}(v, \rhoA) > 0}{\min}\, \frac{ D_{\V \setminus \{v\}}(v, \rhoA)  }{ D_{\emptyset}(v, \rhoA) },
    \end{align*}
    where $D_{S}(v, \rhoA) \coloneqq D(S \cup \{ v \}, \rhoA) - D(S, \rhoA)$ is the marginal \emph{increase} in the expected number of detected attacks when $v \in \V$ is added to the set of detector locations $S\subseteq \V$. We note that $U_{S}(v, \rhoA) = D_{S}(v, \rhoA)$ for every $S\subseteq \V$ and $v \in \V$, so both $U(\cdot, \rhoA)$ and $D(\cdot, \rhoA)$ share the same curvature parameter $c$.
    
    Let $S$ be the set of nodes selected by the forward greedy algorithm in a given iteration. If $|S| < \rD$, in the next iteration the algorithm selects a node $v \in \argmax_{w \in \V} U_{S}(w, \rhoA) = \argmax_{w \in \V} D_{S}(w, \rhoA)$, and updates $S$ to $S \cup \{v\}$. In other words, the algorithm selects a node maximizing the marginal increase in the expected number of detected attacks against $\rhoA$. Therefore, running the forward greedy algorithm on $U(\cdot, \rhoA)$ is equivalent to running it on $D(\cdot, \rhoA)$. By \citet{conforti1984submodular}, the forward greedy algorithm returns a detector positioning $\widehat{S}\in \AD$ satisfying
    \begin{align*}
        D(\widehat{S}, \rhoA) \geq \left( \frac{1 - e^{-c}}{c} \right) D(S^*, \rhoA). 
    \end{align*}
    Rearranging terms, we obtain 
    \begin{align*}
        U(\widehat{S}, \rhoA) \hspace{-0.1em} \leq \hspace{-0.1em}\left( \frac{1 - e^{-c}}{c} \right) \hspace{-0.1em} U(S^*, \rhoA) + \left( \hspace{-0.2em} 1 - \frac{1 - e^{-c}}{c} \right) \hspace{-0.1em} \sum_{e\in \E}  \rhoA_e   \leq  \left( \frac{1 - e^{-c}}{c} \right) \hspace{-0.1em} U(S^*, \rhoA) + \left(\hspace{-0.2em} 1 - \frac{1 - e^{-c}}{c} \right) \hspace{-0.2em} \rA.
    \end{align*}
   \hfill\Halmos

\section*{Proof of Theorem~\ref{thm:column_generation_with_approx_best_response}}


    Let $(\sigmaDHat, \widehat{\lambda}, \widehat{\gamma}) \in \R^{\AD} \times \R^{\E} \times \R$ and $(\rhoAHat, \widehat{\nu}) \in \R^{\E} \times \R$ be optimal primal and dual solutions of the restricted master problem \RMP{$\I$} upon termination of the column generation algorithm with $\alpha$-approximate best response, that is, when the $\alpha$-approximate best response $\widehat{S}$ satisfies $\bar{c}_{\widehat{S}} \geq -\varepsilon$.
    
    First, we argue that $\widehat{\nu} = U(\sigmaDHat, \rhoAHat)$. To this aim, we recall that \RMP{$\I$} is equivalent to \eqref{LP: LP1 reformulation} with the constraint that the support of the variable $\sigmaD$ is contained within the set of columns $\I$. We let $\Supp(\sigmaD) \coloneqq \{S\in \AD:\, \sigmaD_S >0 \}$ denote the support of $\sigmaD$. Then, from similar arguments to those given in the proof of Proposition~\ref{prop:LP1 reformulation}, we can show that \RMP{$\I$} and its dual are respectively equivalent to
    \begin{align*}
        \underset{\substack{ \sigmaD \in \DeltaD:\\ \Supp(\sigmaD) \subseteq \I }}{\min}\, \underset{\rhoA \in \ThetaA}{\max} U(\sigmaD, \rhoA) \quad\text{ and }\quad \underset{\rhoA \in \ThetaA}{\max}\, \underset{\substack{ \sigmaD \in \DeltaD:\\ \Supp(\sigmaD) \subseteq \I }}{\min}\,  U(\sigmaD, \rhoA).
    \end{align*}
    Furthermore, $\sigmaDHat$ and $\rhoAHat$ are respectively optimal solutions of the above minmax and maxmin problems. Therefore,
    \begin{align}
        \label{ineq:for_approx_value_of_game_colgen}
        \widehat{\nu} =  \underset{\substack{ \sigmaD \in \DeltaD:\\ \Supp(\sigmaD) \subseteq \I }}{\min} U(\sigmaD, \rhoAHat) \leq U(\sigmaDHat, \rhoAHat) \leq \underset{\rhoA \in \ThetaA}{\max} U(\sigmaDHat, \rhoA) = \rA \widehat{\gamma} + \sum_{e\in \E} \widehat{\lambda}_e,
    \end{align}
    where the equalities hold since $\rA \widehat{\gamma} + \sum_{e\in \E} \widehat{\lambda}_e$ and $\widehat{\nu}$ are the optimal values of \RMP{$\I$} and its dual, respectively. By strong duality, we have $\widehat{\nu} = \rA \widehat{\gamma} + \sum_{e\in \E} \widehat{\lambda}_e$, and therefore all the inequalities in \eqref{ineq:for_approx_value_of_game_colgen} are in fact \emph{equalities}. Hence, $\widehat{\nu} = U(\sigmaDHat, \rhoAHat)$.
    
    Next, let $S^{*}\in \AD$ be a pure best response against $\rhoAHat$, and let $\widehat{S}\in \AD$ be the approximate best response returned by the $\alpha$-approximation algorithm for \eqref{pb:DBR}. Then,
    \begin{align}
        \label{ineq:upper_bound_for_approx_value_of_game_1}
        U(\sigmaDHat, \rhoAHat) = \widehat{\nu} \leq U(\widehat{S}, \rhoAHat) + \varepsilon \leq \alpha \, U(S^*, \rhoAHat) + \varepsilon. 
    \end{align}
    The first inequality holds by the algorithm's termination criterion, which is $\bar{c}_{\widehat{S}} = -\widehat{\nu} + U(\widehat{S}, \rhoAHat) \geq -\varepsilon$, whereas the second inequality is due to $\widehat{S}$ being an $\alpha$-approximation of \eqref{pb:DBR}. On the other hand, from the optimality of $S^*$, it follows that
    \begin{align}
        \label{ineq:upper_bound_for_approx_value_of_game_2}
        U(S^*, \rhoAHat) \leq \sum_{S\in \AD} \sigmaD_S\, U(S, \rhoAHat) = U(\sigmaD, \rhoAHat), \quad \forall\, \sigmaD \in \DeltaD.
    \end{align}
    Therefore, we have the following bounds for unilateral deviation from $(\sigmaDHat, \rhoAHat)$:
    \begin{align}
        U(\sigmaDHat, \rhoAHat) \overset{\eqref{ineq:upper_bound_for_approx_value_of_game_1}}&{\leq} \alpha \, U(S^*, \rhoAHat) + \varepsilon \overset{\eqref{ineq:upper_bound_for_approx_value_of_game_2}}{\leq} \alpha\,  U(\sigmaD, \rhoAHat) + \varepsilon, \quad &\forall\, \sigmaD \in \DeltaD, \label{ineq:colgen_unilateral_deviation_upper} \\
        U(\sigmaDHat, \rhoAHat) \overset{\eqref{ineq:for_approx_value_of_game_colgen}}&{=} \underset{\rhoA \in \ThetaA}{\max} U(\sigmaDHat, \rhoA) \geq U(\sigmaDHat, \rhoA), \quad &\forall \, \rhoA \in \ThetaA. \label{ineq:colgen_unilateral_deviation_lower}
    \end{align}
    Now, let $(\sigmaDEq, \rhoAEq)$ be an equilibrium of the game $\Gamma$. Using the definition of equilibrium, we have
    \begin{align}
        \label{ineq:bounds_for_approx_value_of_game_colgen}
        U(\sigmaDEq, \rhoAEq) \leq U(\sigmaDHat, \rhoAEq) \overset{\eqref{ineq:colgen_unilateral_deviation_lower}}{\leq} U(\sigmaDHat, \rhoAHat) \overset{\eqref{ineq:colgen_unilateral_deviation_upper}}{\leq} \alpha \, U(\sigmaDEq, \rhoAHat) + \varepsilon \leq \alpha \, U(\sigmaDEq, \rhoAEq) + \varepsilon.
    \end{align}
    Thus, $\sigmaDHat$ and $\rhoAHat$ respectively satisfy the following bounds for the worst-case expected number of undetected attacks with respect to the value of the game $U(\sigmaDEq, \rhoAEq)$:
    \begin{align*}
        U(\sigmaDEq, \rhoAEq) \overset{\eqref{ineq:bounds_for_approx_value_of_game_colgen}}&{\leq} U(\sigmaDHat, \rhoAHat) \overset{\eqref{ineq:for_approx_value_of_game_colgen}}{=} \underset{\rhoA \in \ThetaA}{\max} U(\sigmaDHat, \rhoA) \overset{\eqref{ineq:bounds_for_approx_value_of_game_colgen}}{\leq} \alpha\, U(\sigmaDEq, \rhoAEq) + \varepsilon, \\
        \frac{1}{\alpha} \left( U(\sigmaDEq, \rhoAEq) - \varepsilon \right) \hspace{-0.25em} \overset{\eqref{ineq:bounds_for_approx_value_of_game_colgen}}&{\leq} \hspace{-0.25em} \frac{1}{\alpha} \left( U(\sigmaDHat, \rhoAHat) - \varepsilon \right) \hspace{-0.5em} \overset{\eqref{ineq:colgen_unilateral_deviation_upper}}{\leq} \hspace{-0.5em} \underset{\sigmaD \in \DeltaD}{\min} U(\sigmaD, \rhoAHat) \leq U(\sigmaDEq, \rhoAHat) \leq U(\sigmaDEq, \rhoAEq).
    \end{align*}
   \hfill\Halmos

\section*{Lemmas for Theorem~\ref{thm:MWU_epsilon_NE}}


Before delving into the proof of Theorem~\ref{thm:MWU_epsilon_NE}, we establish some key lemmas. First, we present a standard bound essential for analyzing MWU algorithms. Our proof mirrors that given in \citet{arora2012multiplicative} for the MWU algorithm applied to restricted distributions. 

\begin{lemma}
    \label{lem:MWU_guarantee}
    Algorithm \ref{alg:MWU} satisfies the following bound for the sum of the expected number of undetected attacks across $\tau$ iterations:
    \begin{align*}
       \sum_{t=1}^{\tau} U(S^{(t)}, \rho^{(t)}) \geq  \max_{\rho \in \ThetaA} \left\{ \sum_{t=1}^{\tau} U(S^{(t)}, \rho) - \frac{1}{\eta} \DRE{ \rho }{ \rho^{(1)} } \right\} - \eta \sum_{t=1}^{\tau} \sum_{e\in \E} \rho^{(t)}_e \left(u(S^{(t)}, e)\right)^2.
    \end{align*}
\end{lemma}

\begin{proof}{Proof of Lemma~\ref{lem:MWU_guarantee}.}
    Let $\rho \in \ThetaA$. First, we consider the following Generalized Pythagorean inequality (see \eg, \citet{csiszar1975divergence,herbster2001tracking}):
    \begin{align}
        \label{ineq:pythagorean_ineq_relative_entropy}
        \DRE{ \rho }{ \rho^{(t+1)} } + \DRE{ \rho^{(t+1)} }{ \widetilde{\rho}^{\,(t+1)} } \leq \DRE{ \rho }{ \widetilde{\rho}^{\,(t+1)} }.
    \end{align}
    Since $\DRE{ \rho^{(t+1)} }{ \widetilde{\rho}^{\,(t+1)} } \geq 0$, it follows that
    \begin{align*}
        \DRE{ \rho }{ \rho^{(t+1)} } - \DRE{ \rho }{ \rho^{(t)} } & =  \DRE{ \rho }{ \rho^{(t+1)} } + \DRE{ \rho^{(t+1)} }{ \widetilde{\rho}^{\,(t+1)} } \\
        &\quad - \DRE{ \rho^{(t+1)} }{ \widetilde{\rho}^{\,(t+1)} } - \DRE{ \rho }{ \rho^{(t)} } \\
         \overset{\eqref{ineq:pythagorean_ineq_relative_entropy}}&{\leq} \DRE{ \rho }{ \widetilde{\rho}^{\,(t+1)} } - \DRE{ \rho^{(t+1)} }{ \widetilde{\rho}^{\,(t+1)} } - \DRE{ \rho }{ \rho^{(t)} } \\
        &\leq \DRE{ \rho }{ \widetilde{\rho}^{\,(t+1)} } - \DRE{ \rho }{ \rho^{(t)}}.
    \end{align*}
    Thus,
    \begin{align}
        \label{ineq:for_mwu_bound_1}
        \DRE{ \rho }{ \rho^{(t+1)} } - \DRE{ \rho }{ \rho^{(t)} } \leq \DRE{ \rho }{ \widetilde{\rho}^{\,(t+1)} } - \DRE{ \rho }{ \rho^{(t)}}.
    \end{align}
    Next, we derive an upper bound for the right hand side $\DRE{ \rho }{ \widetilde{\rho}^{\,(t+1)} } - \DRE{ \rho }{ \rho^{(t)}}$. From the update step of the algorithm (Lines \ref{alg:start_update_step}-\ref{alg:end_update_step}), we have $\widetilde{\rho}^{\,(t+1)}_e = \rho^{(t)}_e \exp\left(\eta u(S^{(t)}, e)\right)$ for every $e\in \E$. Then,
    \begin{align*}
        \DRE{ \rho }{ \widetilde{\rho}^{\,(t+1)} }  - \DRE{ \rho }{ \rho^{(t)} } &= \sum_{e\in \E} \rho_e \ln \frac{ \rho^{(t)}_e }{ \widetilde{\rho}^{\,(t+1)}_e} + \widetilde{\rho}^{\,(t+1)}_e - \rho^{(t)}_e \\
        &= \sum_{e\in \E} \rho_e \ln \exp\left( -\eta u(S^{(t)}, e) \right)  + \rho^{(t)}_e \left( \exp\left( \eta u(S^{(t)}, e)  \right) -1  \right) \\
        &\leq -\eta \sum_{e\in \E} \rho_e u(S^{(t)}, e)  \\
        &\phantom{\leq} + \sum_{e\in \E}\rho^{(t)}_e \left(  1 + \eta u(S^{(t)}, e) + \eta^2 \left(u(S^{(t)}, e)\right)^2  -1  \right) \\
        &=-\eta  U(S^{(t)}, \rho) + \eta U(S^{(t)}, \rho^{(t)}) + \eta^2 \sum_{e\in \E} \rho^{(t)}_e \left(u(S^{(t)}, e)\right)^2,
    \end{align*}
    where we used the inequality $e^{x} \leq 1+x+x^2$ for $x\in [-1,1]$. Hence,
    \begin{align}
        \label{ineq:for_mwu_bound_2}
        \DRE{ \rho }{ \widetilde{\rho}^{\,(t+1)} } - \DRE{ \rho }{ \rho^{(t)} } \leq -\eta  U(S^{(t)}, \rho) + \eta U(S^{(t)}, \rho^{(t)}) + \eta^2 \sum_{e\in \E} \rho^{(t)}_e \left(u(S^{(t)}, e)\right)^2.
    \end{align}
    Combining inequalities \eqref{ineq:for_mwu_bound_1} and \eqref{ineq:for_mwu_bound_2}, we obtain
    \begin{align*}
        \DRE{ \rho }{ \rho^{(t+1)} } - \DRE{ \rho }{ \rho^{(t)} } \leq -\eta  U(S^{(t)}, \rho) + \eta U(S^{(t)}, \rho^{(t)}) + \eta^2 \sum_{e\in \E} \rho^{(t)}_e \left(u(S^{(t)}, e)\right)^2.
    \end{align*}
    Summing over $t$, rearranging terms and using that $\DRE{ \rho }{ \rho^{(\tau+1)} } \geq 0$, we get
    \begin{align*}
        \eta \sum_{t=1}^{\tau} U(S^{(t)}, \rho^{(t)}) \geq \eta  \sum_{t=1}^{\tau} U(S^{(t)}, \rho) - \DRE{ \rho }{ \rho^{(1)} } - \eta^2 \sum_{t=1}^{\tau}\sum_{e\in \E} \rho^{(t)}_e \left(u(S^{(t)}, e)\right)^2.
    \end{align*}
    Finally, since $\rho\in \ThetaA$ is arbitrary, we divide by $\eta>0$ to conclude the result.
   \hfill\Halmos
\end{proof}


The next lemma provides an upper bound for the unnormalized relative entropy between an arbitrary vector in $\ThetaA$ and the initial marginal attack strategy of Algorithm~\ref{alg:MWU}.

\begin{lemma}
    \label{lem:upper_bound_DKL}
    Let $\rho^{(1)}\in \ThetaA$ given by $\rho^{(1)}_{e} \coloneqq \rA / m$ for every $e\in \E$. Then, 
    \begin{align*}
        \DRE{ \rho }{ \rho^{(1)} } \leq  \rA \max\left\{ \ln   \frac{m}{\rA}, 1\right\}, \quad \forall\, \rho\in \ThetaA.
    \end{align*}    
\end{lemma}
\begin{proof}{Proof of Lemma~\ref{lem:upper_bound_DKL}.}
    The convexity of the function $x \to x\ln(x)$ in $\R_{\geq 0}$---where we let $0\ln(0) \coloneqq 0$--- implies that the function $\DRE{ \cdot }{ \rho^{(1)} } = \sum_{e\in \E} \left( \rho_{e} \ln \left( \frac{ m}{\rA} \rho_e \right) + \frac{\rA}{m} - \rho_e \right)$ is convex in $\ThetaA$, which is a convex set. Furthermore, $\ThetaA$ is the convex hull of the characteristic vectors of the sets in $\AA$ (see, for example, \citet{bahamondes2022network}). Let $T\in \AA$, and $\mathds{1}_{T} \in \{0,1\}^{\E}$ be its characteristic vector, such that for every $e\in \E$, $(\mathds{1}_{T})_e = 1$ if $e\in T$ and $(\mathds{1}_{T})_e = 0$ otherwise. Then,
    \begin{align*}
        \DRE{ \mathds{1}_{T} }{ \rho^{(1)} } = \sum_{e\in T} \left( \ln \left( \frac{ m}{\rA}  \right) + \frac{\rA}{m} - 1 \right) + \sum_{e \notin T} \frac{\rA}{m} = |T|\ln \left( \frac{ m}{\rA}  \right) - |T| + \rA.
    \end{align*}
    The last term is upper bounded by $\rA$ if $\ln(m/\rA) <1$, and by $\rA \ln(m/\rA)$ otherwise. Since $T$ is arbitrary, we obtain
    \begin{align}
        \label{eq:rel_entropy_characteristic_attack_plan}
        \forall\,T \in \AA,\quad \DRE{ \mathds{1}_{T} }{ \rho^{(1)} } \leq \rA \max\left\{ \ln   \frac{m}{\rA}, 1\right\}.
    \end{align}
    Now, let $\rho \in \ThetaA$. Then, $\rho = \sum_{i=1}^{k} \lambda_i \mathds{1}_{T_{i}}$ for some sets $T_1,\ldots,T_k\in \AA$ and scalars $\lambda_1,\cdots,\lambda_k \in [0,1]$ such that $\sum_{i=1}^{k} \lambda_i = 1$. From the convexity of $\DRE{ \cdot }{ \rho^{(1)} }$, it follows that
    \begin{align*}
        \DRE{ \rho }{ \rho^{(1)} } &= \DRE{ \sum_{i=1}^{k} \lambda_i \mathds{1}_{T_{i}} }{ \rho^{(1)} } \leq \sum_{i=1}^{k} \lambda_i \DRE{ \mathds{1}_{T_i} }{ \rho^{(1)} } \overset{\eqref{eq:rel_entropy_characteristic_attack_plan}}{\leq} \rA \max\left\{ \ln   \frac{m}{\rA}, 1\right\},
    \end{align*}
    which concludes the proof.
   \hfill\Halmos
\end{proof}

Now we are ready to prove Theorem \ref{thm:MWU_epsilon_NE}.

\section*{Proof of Theorem~\ref{thm:MWU_epsilon_NE}}
    From Lemma~\ref{lem:MWU_guarantee}, after $\tau$ iterations of Algorithm~\ref{alg:MWU}, we have
    \begin{align*}
       \sum_{t=1}^{\tau} U(S^{(t)}, \rho^{(t)}) \geq  \max_{\rho \in \ThetaA} \left\{ \sum_{t=1}^{\tau} U(S^{(t)}, \rho) - \frac{1}{\eta}\DRE{ \rho }{ \rho^{(1)} } \right\} - \eta \sum_{t=1}^{\tau} \sum_{e\in \E} \rho^{(t)}_e \left(u(S^{(t)}, e)\right)^2.
    \end{align*}
    Using Lemma~\ref{lem:upper_bound_DKL} to bound the term $\DRE{ \rho }{ \rho^{(1)}}$, and the facts that $u(S^{(t)},e) \in [0,1]$ and $\rho^{(t)} \in \ThetaA$, we obtain
    \begin{align*}
       \frac{1}{\tau}\sum_{t=1}^{\tau} U(S^{(t)}, \rho^{(t)}) \geq  \underset{\rho \in \ThetaA}{\max}\, \frac{1}{\tau} \sum_{t=1}^{\tau} U(S^{(t)}, \rho) - \frac{1}{\eta \tau} \rA \max\left\{ \ln   \frac{m}{\rA}, 1\right\}  - \eta \rA.
    \end{align*}
    Therefore,
    \begin{align}
        \label{ineq:regret_bound_mwu}
        \underset{\rho \in \ThetaA}{\max}\,  \frac{1}{\tau}\sum_{t=1}^{\tau} U(S^{(t)}, \rho) - \frac{1}{\tau}\sum_{t=1}^{\tau} U(S^{(t)}, \rho^{(t)}) \leq \frac{1}{\eta\tau} \rA \max\left\{ \ln   \frac{m}{\rA}, 1\right\}  + \eta  \rA.
    \end{align}
    The right hand side in \eqref{ineq:regret_bound_mwu} is minimized when $\eta = \sqrt{\max\left\{ \ln  (m /\rA), 1\right\} / \tau}$, which yields the bound 
    \begin{align*}
        \underset{\rho \in \ThetaA}{\max}\, \frac{1}{\tau} \sum_{t=1}^{\tau} U(S^{(t)}, \rho) - \frac{1}{\tau}\sum_{t=1}^{\tau} U(S^{(t)}, \rho^{(t)}) \leq 2 \rA \sqrt{  \frac{ \max\left\{ \ln   \frac{m}{\rA}, 1\right\} }{\tau} }.
    \end{align*}
    Thus, the right hand side in \eqref{ineq:regret_bound_mwu} is at most $\varepsilon$ if 
    \begin{align*}
        \tau \geq \frac{ 4\rA^2  \max\left\{ \ln   \frac{m}{\rA}, 1\right\} }{ \varepsilon^2}.
    \end{align*}
    Thus, after $\tau \geq 4\rA^2 \max\left\{ \ln  (m /\rA), 1\right\} / \varepsilon^2$ iterations, we obtain the following lower bound:
    \begin{align}
        \label{ineq:avg_lower_bound_mwu}
        \underset{\rho \in \ThetaA}{\max}\, \frac{1}{\tau} \sum_{t=1}^{\tau} U(S^{(t)}, \rho) - \varepsilon \leq \frac{1}{\tau}\sum_{t=1}^{\tau} U(S^{(t)}, \rho^{(t)}).
    \end{align}
    On the other hand, since $S^{(t)}$ is an $\alpha$-approximation of D's best response against $\rho^{(t)}$, it follows that
    \begin{align*}
        \frac{1}{\tau} \sum_{t=1}^{\tau} U(S^{(t)}, \rho^{(t)}) \leq \frac{1}{\tau}  \sum_{t=1}^{\tau} \alpha\, U(S, \rho^{(t)}), \qquad \forall\, S\in \AD.
    \end{align*}
    Then, by linearity of the expectation, we deduce that
    \begin{align}
        \label{ineq:avg_upper_bound_mwu}
        \frac{1}{\tau} \sum_{t=1}^{\tau} U(S^{(t)}, \rho^{(t)}) \leq \alpha \min_{\sigma \in \DeltaD} \frac{1}{\tau}  \sum_{t=1}^{\tau}  U(\sigma, \rho^{(t)}).
    \end{align}
    Thus, combining \eqref{ineq:avg_lower_bound_mwu} and \eqref{ineq:avg_upper_bound_mwu}, after $\tau \geq 4\rA^2 \max\left\{ \ln  (m /\rA), 1\right\} / \varepsilon^2$ iterations we obtain
    \begin{align}
        \label{ineq:avg_bound_mwu}
         \underset{\rho \in \ThetaA}{\max}\, \frac{1}{\tau} \sum_{t=1}^{\tau} U(S^{(t)}, \rho) - \varepsilon \leq \frac{1}{\tau}\sum_{t=1}^{\tau} U(S^{(t)}, \rho^{(t)}) \leq \alpha  \min_{\sigma \in \DeltaD} \frac{1}{\tau}  \sum_{t=1}^{\tau}  U(\sigma, \rho^{(t)}).
    \end{align}
    Next, using that $\sigmaDHat \coloneqq \frac{1}{\tau}\sum_{t=1}^{\tau} \mathds{1}_{S^{(t)}}$ and $\rhoAHat \coloneqq \frac{1}{\tau}\sum_{t=1}^{\tau} \rho^{(t)}$, we have the following identities:
    \begin{align}
        \frac{1}{\tau} \sum_{t=1}^{\tau} U(S^{(t)}, \rho) &= \sum_{S\in \AD} \frac{1}{\tau}  \sum_{t=1}^{\tau} \mathds{1}_{S^{(t)}}(S) U(S, \rho)  = \sum_{S\in \AD} \sigmaDHat_{S} \, U(S, \rho) = U(\sigmaDHat, \rho), \hspace{0.5em} &\forall \, \rho \in \ThetaA. \label{eq:mwu_identity_1} \\
        \frac{1}{\tau}  \sum_{t=1}^{\tau}  U(\sigma, \rho^{(t)}) &=  \sum_{e\in \E} \frac{1}{\tau} \sum_{t=1}^{\tau}  \rho^{(t)}_e U(\sigma, e) = \sum_{e\in \E} \rhoAHat_e \, U(\sigma, e) = U(\sigma, \rhoAHat), \quad &\forall\, \sigma \in \DeltaD. \label{eq:mwu_identity_2}
    \end{align}
    Let $(\sigmaDEq, \rhoAEq)$ be an equilibrium of the game $\Gamma$. Then,
    \begin{align}
        \underset{ \rho \in \ThetaA}{\max} \, U(\sigmaDHat, \rho) \geq U(\sigmaDHat, \rhoAEq) \geq U(\sigmaDEq, \rhoAEq), \label{ineq:for_bounds_value_of_game_mwu_1} \\
        \underset{ \sigma \in \DeltaD}{\min} \, U(\sigma, \rhoAHat) \leq U(\sigmaDEq, \rhoAHat) \leq U(\sigmaDEq, \rhoAEq). \label{ineq:for_bounds_value_of_game_mwu_2}
    \end{align}
    Therefore,
    \begin{align}
        \label{ineq:final_bounds_mwu}
         U(\sigmaDEq, \rhoAEq) - \varepsilon \overset{\eqref{ineq:for_bounds_value_of_game_mwu_1}}{\leq}  \underset{\rho \in \ThetaA}{\max}\, U(\sigmaDHat, \rho) - \varepsilon  \overset{\eqref{ineq:avg_bound_mwu}, \eqref{eq:mwu_identity_1}, \eqref{eq:mwu_identity_2}}{\leq}  \alpha  \min_{\sigmaD \in \DeltaD} U(\sigmaD, \rhoAHat)  \overset{\eqref{ineq:for_bounds_value_of_game_mwu_2}}{\leq}  \alpha U(\sigmaDEq, \rhoAEq),
    \end{align} 
    from which we deduce the following worst-case bounds for the expected number of undetected attacks for $\sigmaDHat$ and $\rhoAHat$ respectively:
    \begin{align*}
        U(\sigmaDEq, \rhoAEq) \leq \underset{\rhoA \in \ThetaA}{\max}\, U(\sigmaDHat, \rhoA) &\leq \alpha U(\sigmaDEq, \rhoAEq) + \varepsilon, \\
        \frac{1}{\alpha} \left( U(\sigmaDEq, \rhoAEq) - \varepsilon \right) \leq \underset{\sigmaD \in \DeltaD}{\min}\, U(\sigmaD, \rhoAHat) &\leq U(\sigmaDEq, \rhoAEq).
    \end{align*}  
    Finally, the bounds for the unilateral deviation from $(\sigmaDHat, \rhoAHat)$ follow from 
    \begin{align*}
        U(\sigmaDHat, \rhoAHat) \leq \underset{ \rho \in \ThetaA}{\max} \, U(\sigmaDHat, \rho) \overset{\eqref{ineq:final_bounds_mwu}}&{\leq} \alpha  \min_{\sigmaD \in \DeltaD} U(\sigmaD, \rhoAHat) + \varepsilon \leq \alpha U(\sigmaDHat, \rhoAHat) + \varepsilon, \\
        U(\sigmaDHat, \rhoAHat) \geq \min_{\sigmaD \in \DeltaD} U(\sigmaD, \rhoAHat) \overset{\eqref{ineq:final_bounds_mwu}}&{\geq} \frac{1}{\alpha} \left( \underset{\rho \in \ThetaA}{\max}\, U(\sigmaDHat, \rho) - \varepsilon \right) \geq \frac{1}{\alpha} \left( U(\sigmaDHat, \rhoAHat) - \varepsilon \right).
    \end{align*}
   \hfill\Halmos

\section*{Lemmas for Theorem~\ref{thm:projection_closed_form}}


\begin{lemma}
    \label{lem:minimizer_xlogx}
    For every $a>0$ and $b\in \R$, the real function
    \begin{align*}
        f(x) \coloneq 
        \begin{cases}
            x \ln \frac{x}{a} + a - (1-b)x & \text{if } x>0,\\
            a & \text{if } x=0,
        \end{cases}
    \end{align*}
    is strictly convex in $\R_{\geq 0}$ and has a unique minimizer, given by $y^*=\exp(-b)a$. Consequently, it has a unique minimizer in the interval $[0,1]$, given by $x^*=\min \{\exp(-b)a,1\}$.
\end{lemma}
\begin{proof}{Proof of Lemma~\ref{lem:minimizer_xlogx}.}
    The first and second derivatives of $f$ in $\big( 0, +\infty \big)$ are given by $f^{\prime}(x) = \ln (x/a) - b$ and $f^{\dprime}(x) = 1/x$, respectively. Therefore, $f$ is strictly convex and its global minimizer can be obtained by setting $f^{\prime}(y^{*}) = 0$, which yields $y^{*}=\exp(-b)a$. 
   \hfill\Halmos
\end{proof}

\begin{lemma}
    \label{lem:properties_of_g_k}
    Let $\widetilde{\rho} \in \R^{\E}_{>0}$, and let us sort $\E$ such that $\widetilde{\rho}_{e_1} \geq \cdots \geq \widetilde{\rho}_{e_{m}}$, breaking ties arbitrarily. Then:
    \begin{itemize}
        \item[--] The function $g: \{0,\ldots,m\} \to \R_{\geq 0}$ defined by $g(k) \coloneqq k + (1/\widetilde{\rho}_{e_{k}}) \sum_{j=k+1}^{m} \widetilde{\rho}_{e_{j}}$ (where we let $\widetilde{\rho}_{e_{0}} \coloneqq +\infty$, so $1 / \widetilde{\rho}_{e_{0}} \coloneqq 0$) is nondecreasing. Moreover, for every $k\in \{0,\ldots, m-1\}$, $g(k) < g(k+1)$ if and only if $\widetilde{\rho}_{e_{k}} > \widetilde{\rho}_{e_{k+1}}$.

        \item[--] Let $k^* \coloneqq \max \left\{ k\in \{0,\ldots,\rA\}:\, g(k) \leq \rA \right\}$. If $k^{*} \in \{0,\ldots,m-1\}$, then $\widetilde{\rho}_{e_{k^*}} > \widetilde{\rho}_{e_{k^*+1}}$.
    \end{itemize}    
\end{lemma}
\begin{proof}{Proof of Lemma~\ref{lem:properties_of_g_k}.}
    For every $k\in \{0,\ldots,m-1\}$, it holds that
    \begin{align*}
        g(k+1) = k+1 + \frac{1}{\widetilde{\rho}_{e_{k+1}}} \sum_{j=k+2}^{m} \widetilde{\rho}_{e_{j}} = k + \frac{1}{\widetilde{\rho}_{e_{k+1}}} \sum_{j=k+1}^{m} \widetilde{\rho}_{e_{j}} \geq k + \frac{1}{\widetilde{\rho}_{e_{k}}} \sum_{j=k+1}^{m} \widetilde{\rho}_{e_{j}} = g(k).
    \end{align*}
    Therefore, $g(k) \leq g(k+1)$. Furthermore, the inequality above is strict if an only if  $\widetilde{\rho}_{e_{k}} > \widetilde{\rho}_{e_{k+1}}$. Next, let us assume that $k^{*} \in \{0,\ldots,m-1\}$. By definition, $k^*$ satisfies $g(k^{*}) \leq \rA < g(k^{*}+1)$; hence $\widetilde{\rho}_{e_{k^*}} > \widetilde{\rho}_{e_{k^*+1}}$.
   \hfill\Halmos
\end{proof}

The following lemma is standard for results involving Bregman projections (see, for example \citet{yasutake2011online,suehiro2012online}), and shows that such projections preserve the relative ordering of the entries within the projected vector. 

\begin{lemma}
\label{lem:projection_preserves_order}
Let $\widetilde{\rho} \in \R^{\E}_{>0}$, and let $\rho^{*} \in \R^{\E}$ be its projection with respect to the unnormalized relative entropy. For every $e, e^{\prime} \in \E$, if $\widetilde{\rho}_e \geq \widetilde{\rho}_{e^{\prime}}$, then $\rho^{*}_{e} \geq  \rho^{*}_{e^{\prime}}$.
\end{lemma}
\begin{proof}{Proof of Lemma~\ref{lem:projection_preserves_order}.}
    Suppose there exist $e, e^{\prime} \in \E$ with $\widetilde{\rho}_e \geq \widetilde{\rho}_{e^{\prime}}$ and $\rho^{*}_{e} <  \rho^{*}_{e^{\prime}}$. Let ${\rho}^{\prime} \in \R^{\E}$ be defined as
    \begin{align*}
        {\rho}^{\prime}_{f} \coloneqq \begin{cases}
        \rho^{*}_f & \text{if } f \neq e, f \neq e^{\prime},\\
        \rho^{*}_{e^{\prime}} & \text{if } f = e, \\
        \rho^{*}_e & \text{if } f = e^{\prime},
        \end{cases} \quad \forall\, f\in \E.
    \end{align*} 
    In other words, ${\rho}^{\prime}$ is obtained from $\rho^{*}$ by exchanging the coordinates $e$ and $e^{\prime}$. Then,
    \begin{align*}
        \DRE{ \rho^{*} }{ \widetilde{\rho}\, } - \DRE{ {\rho}^{\prime}\, }{ \widetilde{\rho}\, } &= \rho^{*}_{e^{\prime}} \ln \frac{\widetilde{\rho}_e }{ \widetilde{\rho}_{e^{\prime}} } + \rho^{*}_e \ln \frac{ \widetilde{\rho}_{e^{\prime}} }{ \widetilde{\rho}_{e} }  = \left( \rho^{*}_{e^{\prime}} - \rho^{*}_{e} \right) \ln \frac{ \widetilde{\rho}_e }{ \widetilde{\rho}_{e^{\prime}} } \geq  0.
    \end{align*}
    Hence, ${\rho}^{\prime} \in \ThetaA$ and $\DRE{ \rho^{*} }{ \widetilde{\rho}\, } \geq  \DRE{ {\rho}^{\prime}\, }{ \widetilde{\rho}\, }$. Since $\rho^{*}$ minimizes $\DRE{ \cdot }{ \widetilde{\rho}\, }$ over $\ThetaA$, we must have $\DRE{ \rho^{*} }{ \widetilde{\rho}\, } =  \DRE{ {\rho}^{\prime}\, }{ \widetilde{\rho}\, }$, and thus ${\rho}^{\prime} \neq \rho^{*}$ is another minimizer, contradicting the uniqueness of the projection.
   \hfill\Halmos
\end{proof}


\begin{lemma}
    \label{lemma:for_projection_case_2}
    Let $\widetilde{\rho} \in \R^{\E}_{>0}$ satisfying $\sum_{e\in \E} \min\{ \widetilde{\rho}_e, 1\} > \rA$, and let $\rho^{*}\in \R^{\E}$ be its projection into $\ThetaA$ with respect to the unnormalized relative entropy. Then, $\rho^*$ satisfies $\sum_{e\in \E}\rho^{*}_e = \rA$.
\end{lemma}
\begin{proof}{Proof of Lemma~\ref{lemma:for_projection_case_2}.}
    By definition of the projection, we have $\rho^{*}\in \ThetaA$. In particular, $\rho^{*}$ satisfies $\sum_{e\in \E}\rho^{*}_e \leq \rA$. Let us uppose that $\sum_{e\in \E}\rho^{*}_e < \rA$. Then, there exists $e^{\prime}\in \E$ such that $\rho^{*}_{e^{\prime}} < \min\{ \widetilde{\rho}_{e^{\prime}}, 1\}$; otherwise we would have $\sum_{e\in \E}\rho^{*}_e \geq \sum_{e\in \E} \min\{ \widetilde{\rho}_{e}, 1\} > \rA$, which contradicts that $\rho^{*} \in \ThetaA$. Consequently, there exists $\varepsilon>0$ such that $\rho^{*}_{e^{\prime}} + \varepsilon < \min\{ \widetilde{\rho}_{e^{\prime}}, 1\}$ and the vector $\rho^{\prime} \in \R^{\E}$ defined as
    \begin{align*}
        \rho^{\prime}_{e} \coloneqq \begin{cases}
        \rho^{*}_e  & \text{if } e \neq e^{\prime},\\
        \rho^{*}_e + \varepsilon & \text{if } e = e^{\prime},
        \end{cases}\quad \forall\, e\in \E, 
    \end{align*}
    belongs to $\ThetaA$. Then, it follows that
    \begin{align*}
        \DRE{ \rho^{*} }{ \widetilde{\rho}\, } \hspace{-0.1em} - \hspace{-0.1em}\DRE{ \rho^{\prime}\, }{ \widetilde{\rho}\, } &= \left( \rho^{*}_{e^{\prime}} \ln \frac{\rho^{*}_{e^{\prime}} }{ \widetilde{\rho}_{e^{\prime}} } + \widetilde{\rho}_{e^{\prime}} - \rho^{*}_{e^{\prime}} \right) \hspace{-0.1em} - \hspace{-0.1em} \left(  \left( \rho^{*}_{e^{\prime}} + \varepsilon \right) \ln \frac{\rho^{*}_{e^{\prime}} + \varepsilon }{ \widetilde{\rho}_{e^{\prime}} } + \widetilde{\rho}_{e^{\prime}} - \left( \rho^{*}_{e^{\prime}} + \varepsilon \right)  \right) \hspace{-0.1em} > \hspace{-0.1em} 0,
    \end{align*}
    where we used the fact that for every $a>0$, the function 
    \begin{align*}
        f(x) \coloneq 
        \begin{cases}
            x \ln \frac{x}{a} + a - x & \text{if } x>0,\\
            a & \text{if } x=0
        \end{cases}
    \end{align*}
    is strictly decreasing in the interval $[0,\min\{a, 1\}]$ (Lemma~\ref{lem:minimizer_xlogx}). Thus, $\DRE{ \rho^{*} }{ \widetilde{\rho}\, } > \DRE{ \rho^{\prime}\, }{ \widetilde{\rho}\, }$ and $\rho^{\prime} \in \ThetaA$, contradicting the optimality of the projection $\rho^{*}$.
   \hfill\Halmos
\end{proof}

Now we are ready to prove Theorem~\ref{thm:projection_closed_form}.

\section*{Proof of Theorem~\ref{thm:projection_closed_form}}

    We split the proof of Theorem~\ref{thm:projection_closed_form} into the two cases given in its statement.
    \begin{itemize}
        \item[1.] We first consider the case $\sum_{e\in \E} \min\{ \widetilde{\rho}_e, 1\} \leq \rA$. The vector $\rho^*\in \R^{\E}$ given by $\rho^{*}_e = \min\left\{\widetilde{\rho}_e,\, 1 \right\}$ for every $e\in \E$ satisfies $\rho^* \in \ThetaA$, since $\rho^*\in [0,1]^{\E}$ and $\sum_{e\in \E} \rho^{*}_e = \sum_{e\in \E} \min\{ \widetilde{\rho}_e, 1\} \allowbreak \leq \rA$. Then,
        \begin{align*}
            \DRE{ \rho }{ \widetilde{\rho}\, } &= \sum_{e \in \E } \left( \rho_e \ln \frac{ \rho_e }{ \widetilde{\rho}_e } + \widetilde{\rho}_e - \rho_e \right) \\
            &\geq \sum_{e \in \E} \left( \min \{ \widetilde{\rho}_e, 1\} \ln \frac{  \min \{ \widetilde{\rho}_e, 1\} }{ \widetilde{\rho}_e } + \widetilde{\rho}_e -  \min \{ \widetilde{\rho}_e, 1\} \right) = \DRE{ \rho^{*} }{ \widetilde{\rho}\, },
        \end{align*}
        where the inequality follows by Lemma \ref{lem:minimizer_xlogx}. Thus, $\rho^*$ minimizes $\DRE{ \cdot }{ \widetilde{\rho}\, }$ over $\ThetaA$ and therefore it is the projection of $\widetilde{\rho}$.
    
        \item[2.] We next consider the case $\sum_{e\in \E} \min\{ \widetilde{\rho}_e, 1\} > \rA$. First, we note that in this case it holds that $\rA < m$; otherwise, we would have $m = \rA < \sum_{e\in \E} \min\{ \widetilde{\rho}_e, 1\} \leq m$, which is a contradiction. Moreover, from Lemma~\ref{lemma:for_projection_case_2}, the projection $\rho^{*}$ must satisfy $\sum_{e\in \E} \rho^{*}_e = \rA$. Thus, $\rho^{*}$ is the optimal solution of the following optimization problem:
        \begin{align}
            \label{pb:projection_problem_reformulated}
            \tag{P}
            \begin{aligned}
                \underset{\rho \in [0,1]^{\E}}{\min} \quad& \sum_{e \in \E} \left( \rho_e \ln \frac{ \rho_e }{ \widetilde{\rho}_e } + \widetilde{\rho}_e - \rho_e \right) &\\
                \text{s.t.} \quad& \sum_{e\in \E}\rho_e =  \rA,
            \end{aligned}
        \end{align}   
        where we let $0\log 0 \coloneqq 0$ to ensure the objective is continuous. We note that the optimal solution of \eqref{pb:projection_problem_reformulated} is unique as the objective is strictly convex. The Lagrangian function of \eqref{pb:projection_problem_reformulated} is given by
        \begin{align*}
            \mathcal{L}(\rho,\nu) &= \sum_{e \in \E} \left( \rho_e \ln \frac{ \rho_e }{ \widetilde{\rho}_e } + \widetilde{\rho}_e - \rho_e \right) + \nu \left( \sum_{e\in \E}\rho_e - \rA \right) \\
            &=\sum_{e \in \E} \left( \rho_e \ln \frac{ \rho_e }{ \widetilde{\rho}_e } + \widetilde{\rho}_e - (1-\nu)\rho_e \right), \qquad \forall\, (\rho,\nu) \in [0,1]^{\E} \times \R.
        \end{align*}
        The Lagrangian dual of \eqref{pb:projection_problem_reformulated} is the problem
        \begin{align}
            \tag{D}
            \label{pb:lagrangian_dual}
            \underset{\nu \in \R}{\max}\, \underset{\rho \in [0,1]^{\E}}{\min} \, \mathcal{L}(\rho,\nu).
        \end{align}
        Problem \eqref{pb:projection_problem_reformulated} is convex and satisfies the Slater condition (consider, for example, $\rho \in \R^{\E}$ given by $\rho_e \coloneqq \rA / m$ for every $e\in \E$) . Therefore, strong duality holds, and the optimal values of \eqref{pb:projection_problem_reformulated} and \eqref{pb:lagrangian_dual} are identical. Furthermore, if $\nu^{*}$ is an optimal solution of \eqref{pb:lagrangian_dual}, then $\rho^{*}$ minimizes $\mathcal{L}(\cdot,\nu^{*})$ over $[0,1]^{\E}$.  
    
        
        Given the value of $\nu^{*}$, we compute the minimizer of $\mathcal{L}(\cdot,\nu^{*})$ over $[0,1]^{\E}$ (\ie, $\rho^{*}$) analytically. Since $\mathcal{L}(\rho,\nu^{*})$ is separable over the variables $\rho_e$, the minimum is obtained by minimizing each real function $f_{e}(\rho_e) \coloneqq \rho_e \ln ( \rho_e / \widetilde{\rho}_e ) + \widetilde{\rho}_e - (1-\nu^{*})\rho_e$ over the interval $[0,1]$. This yields $\rho^{*}_e = \min \{\exp(-\nu^*)\widetilde{\rho}_e,1\}$ for every $e\in \E$ (Lemma \ref{lem:minimizer_xlogx}). Then, to determine the value of $\nu^{*}$, we use primal feasibility. From the constraint $\sum_{e\in \E} \rho^{*}_e=\rA$ in \eqref{pb:projection_problem_reformulated}, it follows that $\nu^*$ must satisfy
        \begin{align}
            \label{eq:for_nu_star}
            \sum_{e\in \E} \min \{\exp(-\nu^*)\widetilde{\rho}_e,1\} = \rA.
        \end{align}
    
        To solve equation \eqref{eq:for_nu_star}, let us sort the components in $\E$ such that $\widetilde{\rho}_{e_1} \geq \cdots \geq \widetilde{\rho}_{e_{m}}$. Then, by Lemma~\ref{lem:projection_preserves_order}, the projection $\rho^{*}$ satisfies $\rho^{*}_{e_1} \geq \cdots \geq \rho^{*}_{e_m}$. Let $k^*$ be the largest index in $\{1,\ldots,m\}$ such that $\rho^{*}_{e_{k^*}}=1$, if such index exists, and $k^*\coloneqq 0$ otherwise. We note that $k^* \in \{0,\ldots,\rA\}$, as the constraints $\rho^{*} \in [0,1]^{\E}$ and $\sum_{e\in \E} \rho^{*}_e=\rA$  prevent $\rho^{*}$ from having more than $\rA$ entries being equal to $1$.
        Then, equation \eqref{eq:for_nu_star} translates into $k^* + \exp(-\nu^*) \sum_{j=k^*+1}^{m} \widetilde{\rho}_{e_j} = \rA$,
        which yields
        \begin{align*}
            \exp(-\nu^*) = \frac{\rA - k^*}{ \sum_{j=k^*+1}^{m} \widetilde{\rho}_{e_j} }.
        \end{align*}
        Thus, we can write
        \begin{align}
            \label{eq:projection_closed_form}
            \rho^{*}_{e} = \begin{dcases}
                1 & \text{if }  e=e_i,\, i=1,\ldots,k^*,\\
                 \frac{\rA-k^*}{ \sum_{j=k^*+1}^{m} \widetilde{\rho}_{e_j}} \widetilde{\rho}_{e}& \text{if } e=e_i,\, i=k^*+1,\ldots,m.
            \end{dcases}
        \end{align}
        Next, it remains to determine the value of $k^*$. We observe that the following inequalities hold:
        \begin{align}
            \label{ineq:for_k_star}
            \frac{\rA-k^*}{ \sum_{j=k^*+1}^{m} \widetilde{\rho}_{e_{j}} } \widetilde{\rho}_{e_{k^*+1}} < 1  \leq \frac{\rA-k^*}{ \sum_{j=k^*+1}^{m} \widetilde{\rho}_{e_{j}} } \widetilde{\rho}_{e_{k^*}}.
        \end{align}
        Indeed, the first inequality in \eqref{ineq:for_k_star} is equivalent to $\rho^{*}_{e_{k^*+1}} < 1$, which follows by definition of $k^*$. On the other hand, since $\rho^ {*}_{e_{k^*}} = \min \{\exp(-\nu^*)\widetilde{\rho}_{e_{k^ *}},1\}$ and $\rho^ {*}_{e_{k^*}}=1$, we must have $1 \leq \exp(-\nu^*)\widetilde{\rho}_{e_{k^ *}}$, which yields the second inequality in \eqref{ineq:for_k_star}. Rearranging terms in \eqref{ineq:for_k_star}, we obtain the following inequalities for $\rA$:
        \begin{align}
            \label{ineq:for_k_star_2}
             k^* + \frac{1}{\widetilde{\rho}_{e_{k^*}}} \sum_{j=k^*+1}^{m} \widetilde{\rho}_{e_{j}} \leq \rA < k^*+1 + \frac{1}{\widetilde{\rho}_{e_{k^*+1}}} \sum_{j=k^*+2}^{m} \widetilde{\rho}_{e_{j}},
        \end{align}
        where we let $\widetilde{\rho}_{e_0} \coloneqq +\infty$. Let $g: \{0,\ldots,m\} \to \R_{\geq 0}$ be the function defined by $g(k) \coloneqq k + (1/\widetilde{\rho}_{e_{k}}) \sum_{j=k+1}^{m} \widetilde{\rho}_{e_{j}}$ for $k\in \{0,\ldots,m \}$. 
        Therefore, \eqref{ineq:for_k_star_2} is equivalent to the inequalities $g(k^*) \leq \rA < g(k^ *+1)$. Since $g$ is nondecreasing (Lemma~\ref{lem:properties_of_g_k}), we conclude that $k^*$ is the maximum integer in $\{0,\ldots,\rA\}$ whose image under $g$ is at most $\rA$, that is,
        \begin{align*}
            k^* = \max\left\{ k \in \{0,\ldots,\rA\} :\, k + \frac{1}{\widetilde{\rho}_{e_k}} \sum_{j=k+1}^{m} \hspace{-0.3em} \widetilde{\rho}_{e_j}  \leq \rA \right\}.
        \end{align*}
        Finally, we note that $k^*$, as well as the constant $(\rA - k^*) / \sum_{j=1}^{m} \widetilde{\rho}_{e_j}$, do not depend on the tie-breaking rule selected for sorting the entries of $\widetilde{\rho}$. Indeed, if $\pi$ is any other permutation of $\{1,\ldots,m\}$ satisfying $\widetilde{\rho}_{e_{\pi(1)}} \geq \cdots \geq \widetilde{\rho}_{e_{\pi(m)}}$, then $\widetilde{\rho}_{e_{\pi(j)}} = \widetilde{\rho}_{e_j}$ for every $j \in \{1,\ldots,m\}$.
    \end{itemize}
   \hfill\Halmos

\section*{Lemmas for Theorem~\ref{thm:fast_projection}}

Before proceeding with the proof of Theorem~\ref{thm:fast_projection}, we remark that when $\sum_{e\in \E} \min\{ \widetilde{\rho}_e, 1\} \leq \rA$, the fact that Algorithm~\ref{alg:fast_projection} returns the projection of $\widetilde{\rho}$ follows directly from the first case of Theorem~\ref{thm:projection_closed_form}. Therefore, we focus our analysis on the case $\sum_{e\in \E} \min\{ \widetilde{\rho}_e, 1\} > \rA$.

To account for edge cases, throughout this section we assume that both $e_{0}$ and $e_{m+1}$ are two components \emph{not} belonging to $\E$ satisfying $\widetilde{\rho}_{e_0} \coloneqq +\infty$ and $\widetilde{\rho}_{e_{m+1}} \coloneqq 0$. We recall that the value of $k^*$, as well as $\sum_{j=k^*+1}^{m}\widetilde{\rho}_{e_j}$ involved in the closed form of the projection \eqref{eq:projection_closed_form} do not depend on the tie-breaking rule, as noted in the proof of Theorem~\ref{thm:projection_closed_form}.

We denote by $\tau$ the number of iterations of the while loop of Algorithm~\ref{alg:fast_projection} (Lines \ref{alg:while_start}-\ref{alg:while_end}). We use the superscript $(t)$ to denote the iterates generated by the algorithm, which we define as follows.

At the initialization of Algorithm~\ref{alg:fast_projection}, it sets $\F^{(1)} \coloneqq  \E$, $k^{(1)} \coloneqq 0$ and $s^{(1)} \coloneqq 0$. For $t \in \{1,\ldots,\tau\}$, the $t$-th iteration of the while loop of Algorithm~\ref{alg:fast_projection} starts with $\F^{(t)} \neq \emptyset$. Then, the algorithm selects $e^{(t)} \in \F^{(t)}$ corresponding to the $\ceil*{ |\F^{(t)}| / 2 }$-th largest entry of $\left( \widetilde{\rho}_e \right)_{e\in \F^{(t)}}$, breaking ties arbitrarily. 
Then, the algorithm defines the sets $\Fh^{(t)} \coloneqq \left\{ e \in \F^{(t)}:\, \widetilde{\rho}_{e} > \widetilde{\rho}_{e^{(t)}}  \right\}$, $\Fe^{(t)} \coloneqq \left\{ e \in \F^{(t)}:\, \widetilde{\rho}_{e} = \widetilde{\rho}_{e^{(t)}}  \right\}$ and $\Fl^{(t)} \coloneqq \left\{ e \in \F^{(t)}:\, \widetilde{\rho}_{e} < \widetilde{\rho}_{e^{(t)}}  \right\}$. Next, it computes the quantity $\gamma^{(t)} \coloneqq k^{(t)} + |\Fh^{(t)} | + |\Fe^{(t)} | + \left( 1 /\widetilde{\rho}_{e^{(t)}} \right) \left( \sum_{e \in \Fl^{(t)}} \widetilde{\rho}_{e} + s^{(t)} \right)$, and generates the next iterates $\F^{(t+1)}$, $k^{(t+1)}$ and $s^{(t+1)}$ by considering two cases:
\begin{enumerate}
    \item[1.] If $\gamma^{(t)} \leq \rA$, then $\F^{(t+1)} \coloneqq \Fl^{(t)}$, $k^{(t+1)} \coloneqq k^{(t)} + | \Fh^{(t)} | + | \Fe^{(t)} |$, and $s^{(t+1)} \coloneqq s^{(t)}$.
    \item[2.] If $\gamma^{(t)} > \rA$, then $\F^{(t+1)} \coloneqq \Fh^{(t)}$, $k^{(t+1)} \coloneqq k^{(t)}$, and $s^{(t+1)} \coloneqq s^{(t)} + \widetilde{\rho}_{e^{(t)}} | \Fe^{(t)} | + \sum_{e \in \Fl^{(t)}} \widetilde{\rho}_{e}$.
\end{enumerate}
 We note that as a consequence of the final iteration, we have $\F^{(\tau+1)} = \emptyset$, and the algorithm returns the vector $\rho^{*} \in \R^{\E}$ defined as $\rho^{*}_e \coloneqq \min \left\{ \mu \widetilde{\rho}_e ,\, 1 \right\}$ for every $e\in \E$, with $\mu \coloneqq (\rA - k^{(\tau+1)}) / s^{(\tau+1)}$.

First, we show that Algorithm~\ref{alg:fast_projection} runs in linear time.

\begin{lemma}
    \label{lem:alg_linear_time}
    Algorithm~\ref{alg:fast_projection} runs in time $O(m)$.
\end{lemma}
\begin{proof}{Proof of Lemma~\ref{lem:alg_linear_time}.}
    Since $e^{(t)}$ is the index of the $\ceil*{ |\F^{(t)}| / 2 }$-th largest component of $( \widetilde{\rho}_{e} )_{e\in \F^{(t)}}$, from the update rule of Algorithm~\ref{alg:fast_projection}, we have $|\F^{(t+1)}| \leq  \max\left\{  |\Fl^{(t)}|, |\Fh^{(t)}|  \right\} \leq |\F^{(t)}| / 2$. We also note that $e^{(t)}$ can be computed in time $O(m)$ using the median-of-medians algorithm (\citet{blum1973time}).
    
    Furthermore, since the $\tau$-th (\ie, the last) iteration of the while loop generates $\F^{(\tau+1)} = \emptyset$, it follows that the total number of iterations of the while loop is then $\tau \leq \floor*{ \log( |\F^{(1)}| ) + 1}$. By iteratively applying the inequality $|\F^{(t+1)}| \leq |\F^{(t)}| / 2$, we obtain 
    \begin{align*}
        |\F^{(t)}| \leq \frac{ |\F^{(t-1)}| }{ 2 } \leq \cdots \leq \frac{ |\F^{(1)}| }{ 2^{t-1} }, \quad \forall\, t\in \{2,\ldots,\tau\}.
    \end{align*}
    Therefore, the total running time of the while loop in Algorithm~\ref{alg:fast_projection} is given by
    \begin{align*}
        \sum_{t=1}^{\tau} O(|\F^{(t)}|) = \sum_{t=1}^{\floor*{ \log( |\F^{(1)}| ) + 1}} \hspace{-0.5em} O\left( \frac{ |\F^{(1)}| }{2^{t-1}} \right) = O(|\F^{(1)}|) = O(m).
    \end{align*}
    Finally, the initialization and return steps of the algorithm take $O(m)$ as well. Therefore, the overall running time of Algorithm~\ref{alg:fast_projection} is $O(m)$.
   \hfill\Halmos
\end{proof}

Next, we show that when $\sum_{e\in \E} \min\{ \widetilde{\rho}_e, 1\} > \rA$, Algorithm~\ref{alg:fast_projection} satisfies the following invariants throughout the execution of the while loop:

\begin{itemize}
    \item[$\LI^{(t)}$.] The elements of $\F^{(t)}$ are consecutive, that is, for every $i<j<k \in \{1,\ldots,m\}$, $e_{i}, e_{k} \in \F^{(t)}$ implies $e_{j} \in \F^{(t)}$.
    
    \item[$\LII^{(t)}$.] $k^{(t)} = \ell^{(t)}-1$, where $\ell^{(t)} \coloneqq \min \left\{ i\in \{1,\ldots,m\}:\, e_i \in \F^{(t)}  \right\}$.

    \item[$\LIII^{(t)}$.] $s^{(t)} = \sum_{j=k^{(t)}+1}^{m} \widetilde{\rho}_{e_{j}} - \sum_{e\in \F^{(t)}} \widetilde{\rho}_{e}$.

    \item[$\LIV^{(t)}$.] $\F^{(t)}$ contains $e_{k^*}$ or $e_{k^*+1}$. 
\end{itemize}




The following lemma will serve as a useful tool prior to demonstrating the maintenance of these invariants throughout the algorithm.

\begin{lemma}
    \label{lem:alg_gamma_equality}
    Assume $\LI^{(t)}$, $\LII^{(t)}$ and $\LIII^{(t)}$ hold, and let $i_{t} \in \{1,\ldots,m\}$ be such that $e^{(t)} = e_{i_t}$. Then, $\gamma^{(t)} = i_{t} + (1 / \widetilde{\rho}_{e_{i_t}} ) \sum_{j=i_{t}+1}^{m}  \widetilde{\rho}_{e_j}$. Furthermore, $\gamma^{(t)} \leq \rA$ if and only if $i_t \in \{1,\ldots, k^*\}$.
\end{lemma}

\begin{proof}{Proof of Lemma~\ref{lem:alg_gamma_equality}.}
    Let $r^{(t)}$ be defined as $r^{(t)} \coloneqq \min \left\{ j\in \{1,\ldots,m\}:\, e_j \in \Fl^{(t)}  \right\}$ if $\Fl^{(t)} \neq \emptyset$, and $r^{(t)} \coloneqq m+1$ if  $\Fl^{(t)} = \emptyset$. Since $\LI^{(t)}$ holds and $\ell^{(t)}$ is the index of the first element of $\F^{(t)}$, we can write $\Fh^{(t)}  \cup \Fe^{(t)} = \left\{ e_{\ell^{(t)}}, e_{\ell^{(t)}+1},
 \ldots, e_{r^{(t)} - 1} \right\}$. Moreover, since $e_{i_{t}} \in \Fe^{(t)}$, we have $\{ e_{i_{t}+1},\ldots, e_{r^{(t)}-1} \} \subseteq \Fe^{(t)}$. Then, it follows that
    \begin{align*}
        \gamma^{(t)} &= k^{(t)} + |\Fh^{(t)} | + |\Fe^{(t)} | + \frac{ 1 }{\widetilde{\rho}_{e^{(t)}} } \left( \sum_{e \in \Fl^{(t)}} \widetilde{\rho}_{e} + s^{(t)} \right) &\\
        \overset{ (\LII^{(t)}, \LIII^{(t)} ) }&{=} \ell^{(t)} - 1 + |\Fh^{(t)} | + |\Fe^{(t)} | + \frac{ 1 }{\widetilde{\rho}_{e^{(t)}} } \left( \sum_{e \in \Fl^{(t)}} \widetilde{\rho}_{e} + \sum_{j=\ell^{(t)}}^{m} \widetilde{\rho}_{e_{j}} - \sum_{e\in \F^{(t)}} \widetilde{\rho}_{e} \right) \\
        &= \ell^{(t)} - 1 + |\Fh^{(t)} | + |\Fe^{(t)} | + \frac{ 1 }{\widetilde{\rho}_{e^{(t)}} } \left( \sum_{j=\ell^{(t)}}^{m} \widetilde{\rho}_{e_{j}} - \sum_{e \in \Fh^{(t)} \cup \Fe^{(t)}  } \widetilde{\rho}_{e} \right) \\
        &= | \left\{ e_{1},\ldots, e_{\ell^{(t) - 1}} \right\} | + | \left\{ e_{\ell^{(t)}},\ldots, e_{r^{(t)} - 1} \right\} | + \frac{ 1 }{\widetilde{\rho}_{e^{(t)}} } \left( \sum_{j=\ell^{(t)}}^{m} \widetilde{\rho}_{e_{j}} - \sum_{j=\ell^{(t)}}^{r^{(t)} - 1} \widetilde{\rho}_{e_{j}} \right) \\
        &= | \left\{ e_{1},\ldots, e_{\ell^{(t)}-1} \right\} | + | \left\{ e_{\ell^{(t)}},\ldots, e_{i_{t}} \right\} | + \frac{ 1 }{\widetilde{\rho}_{e_{i_{t}} }} \left( \sum_{j= i_{t}+1 }^{ r^{(t)} - 1 } \widetilde{\rho}_{e_{j}}  +  \sum_{j=r^{(t)}}^{m} \widetilde{\rho}_{e_{j}}  \right) \\
        &= i_{t} + \frac{1}{\widetilde{\rho}_{e_{i_t}} } \sum_{j=i_{t}+1}^{m}  \widetilde{\rho}_{e_j}.
    \end{align*}
    Finally, the fact that $\gamma^{(t)} = i_{t} + (1 / \widetilde{\rho}_{e_{i_t}} ) \sum_{j=i_{t}+1}^{m}  \widetilde{\rho}_{e_j} \leq \rA$ if and only if $i_t \in \{1,\ldots, k^*\}$ follows directly from the definition of $k^*$ and the fact that the function $g(k) \coloneqq k + (1/\widetilde{\rho}_{e_{k}}) \sum_{j=k+1}^{m} \widetilde{\rho}_{e_j}$ (with $\widetilde{\rho}_{e_0} \coloneqq +\infty$) is nondecreasing (Lemma~\ref{lem:properties_of_g_k}).
   \hfill\Halmos
\end{proof}

\begin{lemma}
    \label{lem:alg_invariants}
    Algorithm~\ref{alg:fast_projection} satisfies invariants $\LI^{(t)}$, $\LII^{(t)}$, $\LIII^{(t)}$  and $\LIV^{(t)}$ for every $t \in \{1,\ldots, \tau\}$.
\end{lemma}
\begin{proof}{Proof of Lemma~\ref{lem:alg_invariants}.}
    We prove the lemma by induction on $t$. At the initialization of Algorithm~\ref{alg:fast_projection}, we have $\F^{(1)} = \E$, $k^{(1)} = 0$, and $s^{(1)}= 0$. In particular, invariants $\LI^{(1)}$, $\LII^{(1)}$, and $\LIII^{(1)}$ hold. Moreover, if $k^*>0$, then $\F^{(1)}$ contains $e_{k^*}$, and if $k^*=0$, $\F^{(1)}$ contains $e_1 = e_{k^*+1}$, so $\LIV^{(1)}$ holds as well.

    Now, for $t \in \{1,\ldots,\tau-1\}$, let us suppose invariants $\LI^{(t)}$, $\LII^{(t)}$, $\LIII^{(t)}$  and $\LIV^{(t)}$ hold. At iteration $t$ of the while loop (Lines \ref{alg:while_start}-\ref{alg:while_end}), the algorithm selects a component $e^{(t)} \in \F^{(t)}$, and defines $\Fh^{(t)} = \{ e\in \F^{(t)}:\, \widetilde{\rho}_{e} > \widetilde{\rho}_{e^{(t)}} \}$, $\Fl^{(t)} = \{ e\in \F^{(t)}:\, \widetilde{\rho}_{e} < \widetilde{\rho}_{e^{(t)}} \}$ and $\Fe^{(t)} = \{ e\in \F^{(t)}:\, \widetilde{\rho}_{e} = \widetilde{\rho}_{e^{(t)}} \}$. Let $i_{t} \in \{1,\ldots,m\}$ be such that $e^{(t)}=e_{i_{t}}$. Then, we have two cases:

    \begin{itemize}
        \item[1.] $\gamma^{(t)} \leq \rA$. Then, the algorithm sets $\F^{(t+1)} = \Fl^{(t)}$, $k^{(t+1)} = k^{(t)} + | \Fh^{(t)} | + | \Fe^{(t)} |$, and $s^{(t+1)} = s^{(t)}$. From $\LI^{(t)}$, the elements of $\F^{(t)}$ are consecutive, so the elements of $\F^{(t+1)}$ are consecutive as well, hence $\LI^{(t+1)}$ holds. On the other hand, by definition, $\ell^{(t)}$ is the index of the largest element of $\F^{(t)}$, and $\ell^{(t+1)}$ is the index of the largest element of $\F^{(t+1)} = \Fl^{(t)}$. Since the elements of $\F^{(t)}$ are consecutive, we can write $\Fh^{(t)}  \cup \Fe^{(t)} = \left\{ e_{\ell^{(t)}},e_{\ell^{(t)}+1},\ldots, e_{\ell^{(t+1)} - 1} \right\}$, and thus $\ell^{(t+1)} = \ell^{(t)} + | \Fh^{(t)} | + | \Fe^{(t)} |$. Then, it follows that
        \begin{align}
            \label{eq:for_loop_invariant_case_1}
            k^{(t+1)} = k^{(t)} + | \Fh^{(t)} | + | \Fe^{(t)} | 
            \overset{ (\LII^{(t)}) }{=} \ell^{(t)} - 1 + | \Fh^{(t)} | + | \Fe^{(t)} | = \ell^{(t+1)} - 1.
        \end{align}
        Therefore, $\LII^{(t+1)}$ holds. Moreover, 
        \begin{align*}
            s^{(t+1)} = s^{(t)} \overset{(\LIII^{(t)})}&{=} \sum_{j=k^{(t)}+1}^{m} \widetilde{\rho}_{e_j} - \sum_{e\in \F^{(t)}} \widetilde{\rho}_{e} \\
            \overset{(\LII^{(t)})}&{=} \sum_{j=\ell^{(t)}}^{m} \widetilde{\rho}_{e_j} - \sum_{e\in \F^{(t)}} \widetilde{\rho}_{e} \\
            &= \sum_{j=\ell^{(t)}}^{\ell^{(t+1)}-1} \widetilde{\rho}_{e_j} + \sum_{j=\ell^{(t+1)}}^{m} \widetilde{\rho}_{e_j} - \sum_{e\in \F^{(t)}} \widetilde{\rho}_{e}  \\
            &= \sum_{e \in \Fh^{(t)} \cup \Fe^{(t)}} \widetilde{\rho}_{e} + \sum_{j=\ell^{(t+1)}}^{m} \widetilde{\rho}_{e_j} - \sum_{e\in \F^{(t)}} \widetilde{\rho}_{e} \\
            &= \sum_{j=\ell^{(t+1)}}^{m} \widetilde{\rho}_{e_j} - \sum_{e\in \Fl^{(t)}} \widetilde{\rho}_{e} \\
            \overset{\eqref{eq:for_loop_invariant_case_1}}&{=} \sum_{j=k^{(t+1)}+1}^{m} \widetilde{\rho}_{e_j} - \sum_{e\in \F^{(t+1)}} \widetilde{\rho}_{e}.
        \end{align*}
        Thus, $\LIII^{(t+1)}$ holds. Next, we show that $\LIV^{(t+1)}$ holds. From $\LI^{(t)}$, $\LII^{(t)}$, $\LIII^{(t)}$ and Lemma~\ref{lem:alg_gamma_equality}, we have $\gamma^{(t)} = i_{t} + (1 / \widetilde{\rho}_{e_{i_t}} ) \sum_{j=i_{t}+1}^{m}  \widetilde{\rho}_{e_j} \leq \rA$, and therefore, $i_{t} \in  \{1,\ldots,k^*\}$. Moreover, from Lemma~\ref{lem:properties_of_g_k}, it holds that $\widetilde{\rho}_{e_{k^*}} >  \widetilde{\rho}_{e_{k^*+1}}$. Hence, we have $\widetilde{\rho}_{e_{i_{t}}} \geq \widetilde{\rho}_{e_{k^*}} >  \widetilde{\rho}_{e_{k^*+1}}$. From $\LIV^{(t)}$, we have $e_{k^*} \in \F^{(t)}$ or $e_{k^*+1} \in \F^{(t)}$. We then consider the following three cases. First, if $e_{k^*+1} \in \F^{(t)}$, then the inequality $\widetilde{\rho}_{e_{i_{t}}}  >  \widetilde{\rho}_{e_{k^*+1}}$ implies $e_{k^*+1} \in \Fl^{(t)} = \F^{(t+1)}$. Second, if $e_{k^*}\in \F^{(t)}$ and $\widetilde{\rho}_{e_{i_{t}}} > \widetilde{\rho}_{e_{k^*}}$, then $e_{k^*} \in \Fl^{(t)} = \F^{(t+1)}$. Third, if $e_{k^*}\in \F^{(t)}$ and $\widetilde{\rho}_{e_{i_{t}}} = \widetilde{\rho}_{e_{k^*}}$, we have $e_{k^{*}} \in \Fe^{(t)}$. Moreover, $\Fl^{(t)} \neq \emptyset$, as $t < \tau$ and $\Fl^{(t)} = \F^{(t+1)} \neq \emptyset$. From $\LI^{(t)}$, the elements of $\F^{(t)}$ are consecutive, and $\widetilde{\rho}_{e_{k^*}} > \widetilde{\rho}_{e_{k^*+1}}$ necessarily implies that $e_{k^*+1} \in \Fl^{(t)} = \F^{(t+1)}$. Hence, $\LIV^{(t+1)}$ holds.

        \item[2.] $\gamma^{(t)} > \rA$. Then, $\F^{(t+1)} = \Fh^{(t)}$, $k^{(t+1)} = k^{(t)}$, and $s^{(t+1)} = s^{(t)} + \widetilde{\rho}_{e^{(t)}} | \Fe^{(t)} | + \sum_{e \in \Fl^{(t)}} \widetilde{\rho}_{e}$. From $\LI^{(t)}$, the elements of $\F^{(t)}$ are consecutive, so the elements of $\Fh^{(t)}$ are consecutive as well, so $\LI^{(t+1)}$ holds. On the other hand, by definition, $\ell^{(t)}$ is the index of the largest element of $\F^{(t)}$, and $\ell^{(t+1)}$ is the index of the largest element of $\F^{(t+1)} = \Fh^{(t)}$. Therefore, $\ell^{(t+1)} = \ell^{(t)}$. From, $\LII^{(t)}$, it follows that
        \begin{align*}
            k^{(t+1)} = k^{(t)} \overset{(\LII^{(t)})}{=} \ell^{(t)} - 1 = \ell^{(t+1)} - 1.
        \end{align*}
        Therefore, $\LII^{(t+1)}$ holds. Moreover,
        \begin{align*}
            s^{(t+1)} &= s^{(t)} + \widetilde{\rho}_{e^{(t)}} | \Fe^{(t)} | + \sum_{e \in \Fl^{(t)}} \widetilde{\rho}_{e} \\
            \overset{(\LIII^{(t)})}&{=} \sum_{j=k^{(t)} + 1}^{m} \widetilde{\rho}_{e_{j}} - \sum_{e \in \F^{(t)}} \widetilde{\rho}_{e}     + \widetilde{\rho}_{e^{(t)}} | \Fe^{(t)} | + \sum_{e \in \Fl^{(t)}} \widetilde{\rho}_{e} \\
            &= \sum_{j=k^{(t)} + 1}^{m} \widetilde{\rho}_{e_{j}} - \sum_{e \in \F^{(t)}} \widetilde{\rho}_{e}    + \sum_{e\in \Fe^{(t)}} \widetilde{\rho}_{e}  + \sum_{e \in \Fl^{(t)}} \widetilde{\rho}_{e} \\
            &= \sum_{j=k^{(t)} + 1}^{m} \widetilde{\rho}_{e_{j}} - \sum_{e \in \Fh^{(t)}} \widetilde{\rho}_{e} \\
            &= \sum_{j=k^{(t+1)} + 1}^{m} \widetilde{\rho}_{e_{j}} - \sum_{e \in \F^{(t+1)}} \widetilde{\rho}_{e}.
        \end{align*}
        Thus, $\LIII^{(t)}$ holds. Next, we show that $\LIV^{(t)}$ holds. From $\LI^{(t)}$, $\LII^{(t)}$, $\LIII^{(t)}$ and Lemma~\ref{lem:alg_gamma_equality}, we have $\gamma^{(t)} = i_{t} + (1 / \widetilde{\rho}_{e_{i_t}} ) \sum_{j=i_{t}+1}^{m}  \widetilde{\rho}_{e_j} > \rA$, and therefore, $i_{t} \in  \{k^{*}+1,\ldots,m\}$. This implies $\widetilde{\rho}_{e_{k^*}} > \widetilde{\rho}_{e_{k^*+1}} \geq \widetilde{\rho}_{e_{i_{t}}}$. From $\LIV^{(t)}$, we have $e_{k^*} \in \F^{(t)}$ or $e_{k^*+1} \in \F^{(t)}$. We then consider the following three cases. First, if $e_{k^*} \in \F^{(t)}$, then the inequality $\widetilde{\rho}_{e_{k^*}} > \widetilde{\rho}_{e_{i_{t}}}$ implies $e_{k^*} \in \Fh^{(t)} = \F^{(t+1)}$. Second, if $e_{k^*+1}\in \F^{(t)}$ and $\widetilde{\rho}_{e_{k^*+1}} > \widetilde{\rho}_{e_{i_{t}}}$, then $e_{k^*+1} \in \Fh^{(t)} = \F^{(t+1)}$. Third, if $e_{k^*+1}\in \F^{(t)}$ and $\widetilde{\rho}_{e_{k^*+1}} = \widetilde{\rho}_{e_{i_{t}}}$, we have $e_{k^{*}+1} \in \Fe^{(t)}$. Moreover, $\Fh^{(t)} \neq \emptyset$, as $t < \tau$ and $\Fh^{(t)} = \F^{(t+1)} \neq \emptyset$. From $\LI^{(t)}$, the elements of $\F^{(t)}$ are consecutive, and $\widetilde{\rho}_{e_{k^*}} > \widetilde{\rho}_{e_{k^*+1}}$ necessarily implies that $e_{k^*} \in \Fh^{(t)} = \F^{(t+1)}$. Hence, $\LIV^{(t+1)}$ holds.
    \end{itemize}
   \hfill\Halmos
\end{proof}

Now we are ready to prove Theorem~\ref{thm:fast_projection}.

\section*{Proof of Theorem~\ref{thm:fast_projection}}
    If $\sum_{e\in \E} \min\{ \widetilde{\rho}_e, 1\} \leq \rA$, the projection of $\widetilde{\rho}$ returned by Algorithm~\ref{alg:fast_projection} directly follows from the first case of Theorem~\ref{thm:projection_closed_form}. Therefore, for the rest of the proof, we assume that $\sum_{e\in \E} \min\{ \widetilde{\rho}_e, 1\} > \rA$.

    Let us consider the $\tau$-th (\ie, the last) iteration of the while loop (lines \ref{alg:while_start}-\ref{alg:while_end}). From Lemma~\ref{lem:alg_invariants}, Algorithm~\ref{alg:fast_projection} satisfies invariants $\LI^{(\tau)}$, $\LII^{(\tau)}$, $\LIII^{(\tau)}$, and $\LIV^{(\tau)}$. At iteration $\tau$, the algorithm selects $e^{(\tau)} \in \F^{(\tau)}$, and defines $\Fh^{(\tau)} = \{ e\in \F^{(\tau)}:\, \widetilde{\rho}_{e} > \widetilde{\rho}_{e^{(\tau)}} \}$, $\Fl^{(\tau)} = \{ e\in \F^{(\tau)}:\, \widetilde{\rho}_{e} < \widetilde{\rho}_{e^{(\tau)}} \}$ and $\Fe^{(\tau)} = \{ e\in \F^{(\tau)}:\, \widetilde{\rho}_{e} = \widetilde{\rho}_{e^{(\tau)}} \}$. Let $i_{\tau} \in \{1,\ldots,m\}$ be such that $e^{(\tau)}=e_{i_{\tau}}$. Then, we have the following two cases:

    \begin{itemize}
        \item[1.] $\gamma^{(\tau)} \leq \rA$. Then, the algorithm sets $\F^{(\tau+1)} = \Fl^{(\tau)}$, $k^{(\tau+1)} = k^{(\tau)} + | \Fh^{(\tau)} | + | \Fe^{(\tau)} |$, and $s^{(\tau+1)} = s^{(\tau)}$. Moreover, from Lemma~\ref{lem:alg_gamma_equality}, $\gamma^{(\tau)} = i_{\tau} + (1 / \widetilde{\rho}_{e_{i_{\tau}}} ) \sum_{j=i_{\tau}+1}^{m}  \widetilde{\rho}_{e_j} \leq \rA$, so $i_{\tau} \in \{1,\ldots,k^*\}$, which implies $\widetilde{\rho}_{i_{\tau}} \geq \widetilde{\rho}_{e_{k^*}}$. We also recall that the definition of $k^*$ implies $\widetilde{\rho}_{e_{k^*}} > \widetilde{\rho}_{e_{k^*}+1}$ (Lemma~\ref{lem:properties_of_g_k}). On the other hand, since the $\tau$-th iteration is the last one of the while loop, we must have $\Fl^{(\tau)} = \F^{(\tau+1)} = \emptyset$.

        We next argue that $e_{k^*} \in \F^{(\tau)}$, but $e_{k^*+1} \notin \F^{(\tau)}$. From $\LIV^{(\tau)}$, we know that $\F^{(\tau)}$ contains $e_{k^*}$ or $e_{k^{*}+1}$. Suppose $e_{k^*+1} \in \F^{(\tau)}$. Then, $\Fl^{(\tau)} = \emptyset$ implies $\widetilde{\rho}_{i_{\tau}} \leq \widetilde{\rho}_{e_{k^*+1}} < \widetilde{\rho}_{e_{k^*}}$, which is a contradiction. Thus, $e_{k^*} \in \F^{(\tau)}$. Additionally, $\Fl^{(\tau)} = \emptyset$ implies $\widetilde{\rho}_{i_{\tau}} \leq \widetilde{\rho}_{e_{k^*}}$, so $\widetilde{\rho}_{i_{\tau}} = \widetilde{\rho}_{e_{k^*}}$. In particular, $e_{k^*} \in \Fe^{(\tau)}$. Moreover, due to the facts that $\widetilde{\rho}_{e_{k^*}} > \widetilde{\rho}_{e_{k^*+1}}$, the elements of $\F^{(\tau)}$ are consecutive by $\LI^{(\tau)}$, and $\Fl^{(\tau)}=\emptyset$, it follows that $e_{k^*}$ is the smallest element of $\F^{(\tau)}$. Therefore, we can write $\F^{(\tau)} = \{ e_{\ell^{(\tau)}}, e_{\ell^{(\tau)}+1},\ldots, e_{k^*}\}$ and then $\ell^{(\tau)} = k^* - |\F^{(\tau)}| + 1$. Then,
        \begin{align*}
            k^{(\tau+1)} = k^{(\tau)} + | \Fh^{(\tau)} | + | \Fe^{(\tau)} | = \ell^{(\tau)} - 1 + | \F^{(\tau)} |=  k^{*}.
        \end{align*}
        In addition, 
        \begin{align*}
            s^{(\tau+1)} = s^{(\tau)} \overset{(\LIII^{(\tau)})}&{=} \sum_{j=k^{(\tau)}+1}^{m} \widetilde{\rho}_{e_j} - \sum_{e\in \F^{(\tau)}} \widetilde{\rho}_{e} \\ 
             \overset{(\LII^{(\tau)})}&{=} \sum_{j=\ell^{(\tau)}}^{m} \widetilde{\rho}_{e_j} - \sum_{e\in \F^{(\tau)}} \widetilde{\rho}_{e} \\
             &= \sum_{j=\ell^{(\tau)}}^{k^*} \widetilde{\rho}_{e_j} + \sum_{j=k^*+1}^{m} \widetilde{\rho}_{e_j} - \sum_{e\in \F^{(\tau)}} \widetilde{\rho}_{e} \\
             &= \sum_{j=k^*+1}^{m} \widetilde{\rho}_{e_j}.
        \end{align*}

        \item[2.] $\gamma^{(\tau)} > \rA$. Then, the algorithm sets $\F^{(\tau+1)} = \Fh^{(\tau)}$, $k^{(\tau+1)} = k^{(\tau)}$, and $s^{(\tau+1)} = s^{(\tau)} + \widetilde{\rho}_{e^{(\tau)}} | \Fe^{(\tau)} | + \sum_{e \in \Fl^{(\tau)}} \widetilde{\rho}_{e}$. Moreover, from Lemma~\ref{lem:alg_gamma_equality}, $\gamma^{(\tau)} = i_{\tau} + (1 / \widetilde{\rho}_{e_{i_{\tau}}} ) \sum_{j=i_{\tau}+1}^{m}  \widetilde{\rho}_{e_j} > \rA$, so $i_{\tau} \in \{k^*+1,\ldots,m\}$, which implies $\widetilde{\rho}_{e_{k^*+1}} \geq \widetilde{\rho}_{i_{\tau}}$. On the other hand, since the $\tau$-th iteration is the last one of the while loop, we must have $\Fh^{(\tau)} = \F^{(\tau+1)} = \emptyset$.

        We next argue that $e_{k^*+1} \in \F^{(\tau)}$, but $e_{k^*} \notin \F^{(\tau)}$. From $\LIV^{(\tau)}$, we know that $\F^{(\tau)}$ contains $e_{k^*}$ or $e_{k^{*}+1}$. Suppose $e_{k^*} \in \F^{(\tau)}$. Then, $\Fh^{(\tau)} = \emptyset$ implies $\widetilde{\rho}_{i_{\tau}} \geq \widetilde{\rho}_{e_{k^*}} > \widetilde{\rho}_{e_{k^*+1}}$, which is a contradiction. Thus, $e_{k^*+1} \in \F^{(\tau)}$. Additionally, $\Fh^{(\tau)} = \emptyset$ implies $\widetilde{\rho}_{i_{\tau}} \leq \widetilde{\rho}_{e_{k^*+1}}$, so $\widetilde{\rho}_{i_{\tau}} = \widetilde{\rho}_{e_{k^*+1}}$. In particular, $e_{k^*+1} \in \Fe^{(\tau)}$. Moreover, due to the facts that $\widetilde{\rho}_{e_{k^*}} > \widetilde{\rho}_{e_{k^*+1}}$, the elements of $\F^{(\tau)}$ are consecutive by $\LI^{(\tau)}$, and $\Fh^{(\tau)}=\emptyset$, it follows that $e_{k^*+1}$ is the largest element of $\F^{(\tau)}$, so $\ell^{(\tau)} = k^*+1$. Then,
        \begin{align}
            \label{eq:for_projection_alg_correcteness_cae_2}
            k^{(\tau+1)} = k^{(\tau)} = \ell^{(\tau)} - 1 =  k^{*}.
        \end{align}
        In addition, 
        \begin{align*}
            s^{(\tau+1)} &= s^{(\tau)} + \widetilde{\rho}_{e^{(\tau)}} | \Fe^{(\tau)} | + \sum_{e \in \Fl^{(\tau)}} \widetilde{\rho}_{e} \\
            \overset{(\LIII^{(\tau)})}&{=} \sum_{j=k^{(\tau)}+1}^{m} \widetilde{\rho}_{e_j} - \sum_{e\in \F^{(\tau)}} \widetilde{\rho}_{e} + \sum_{e \in \Fe^{(\tau)}} \widetilde{\rho}_{e} + \sum_{e \in \Fl^{(\tau)}} \widetilde{\rho}_{e} \\ 
            &= \sum_{j=k^{(\tau)}+1}^{m} \widetilde{\rho}_{e_j}    \\
            \overset{\eqref{eq:for_projection_alg_correcteness_cae_2}}&{=} \sum_{j=k^*+1}^{m} \widetilde{\rho}_{e_j}.
        \end{align*}
    \end{itemize}
    Therefore, after the the execution of the while loop, we have $k^{(\tau+1)} = k^*$ and $s^{(\tau+1)} = \sum_{j=k^*+1}^{m} \widetilde{\rho}_{e_j}$. This ensures that $\mu = (\rA - k^*) / \sum_{j=k^*+1}^{m} \widetilde{\rho}_{e_j}$, and by the second case of Theorem~\ref{thm:projection_closed_form}, the vector returned by Algorithm~\ref{alg:fast_projection} is the projection of $\widetilde{\rho}$.
    
    Finally, the fact that Algorithm~\ref{alg:fast_projection} runs in time $O(m)$ is shown in Lemma~\ref{lem:alg_linear_time}.
   \hfill\Halmos